\documentclass[12pt]{article}

\usepackage{hyperref}
\usepackage{graphicx}
\usepackage{epsfig}
\usepackage{epstopdf}
\usepackage[]{algorithm2e}
\usepackage{verbatim}
\usepackage{caption}
\usepackage{lineno}

\usepackage{amsmath,amssymb,amsthm}
\usepackage{xcolor}

\usepackage{fullpage}

\usepackage{hyperref}
\usepackage{cite}
\usepackage{url}
\newtheorem{theorem}{Theorem}
\newtheorem{lemma}[theorem]{Lemma}
\newtheorem{claim}{Claim}
\newtheorem{proposition}[theorem]{Proposition}
\newtheorem{corollary}[theorem]{Corollary}

\newtheorem{observation}{Observation}
\newtheorem{conjecture}{Conjecture}
\theoremstyle{definition}

\newcommand{\remove}[1]{{}}
\usepackage[normalem]{ulem}

\newcommand{\note}[1]{{#1}}
\newcommand{\rednote}[1]{{\color{red} #1}}

\newcommand{\changed}[1]{{#1}}
\newcommand{\Anna}[1]{{#1}}
\newcommand{\megan}[1]{{#1}}
\newcommand{\revised}[1]{{#1}}
\newcommand{\rerevised}[1]{{\color{red} #1}}

\newcommand\bR{\mathbb{R}}
\newcommand\CH{{\rm CH}}

\newcommand\LCH{{\rm LCH}}

\newcommand{\T}{\mathcal{T}}

\title{Shortest Paths and Convex Hulls in 2D Complexes with Non-Positive Curvature}
\author{%
  Anna Lubiw%
    \thanks{David R. Cheriton School of Computer Science,
      University of Waterloo, Waterloo, Ontario N2L 3G1, Canada,
      \protect\url{alubiw@uwaterloo.ca, daniela.maftuleac@uwaterloo.ca};
      $\dag$ Department of Mathematics, Lehman College, City University of New York, United States, 
      \protect\url{megan.owen@lehman.cuny.edu}
      }
\and
  Daniela Maftuleac$^*$%
\and
  Megan Owen$^{\dag}$%
}

\date{~}

\begin{document}

\maketitle

\begin{abstract}
Globally non-positively curved, or CAT(0), polyhedral complexes arise in a number of applications, including evolutionary biology and robotics.  These spaces have unique shortest paths and are composed of Euclidean polyhedra, yet many 
algorithms and properties of shortest paths and convex hulls 
in Euclidean space fail to transfer over. 
\note{We give an algorithm, using linear programming, to compute the convex hull of a set of points in a 2-dimensional CAT(0) polyhedral complex with a single vertex.}
%We give an algorithm for computing convex hulls using linear programming for 
%2-dimensional CAT(0) polyhedral complexes with a single vertex.
\changed{We explore the use of shortest path maps to answer single-source
shortest path queries in 2-dimensional CAT(0) polyhedral complexes,
and we %\note{\sout{improve and}} 
unify \note{efficient} solutions for 2-manifold and rectangular cases.}
%
%properties of convex hulls in Euclidean space fail to transfer over.  We give examples of some such properties.  
%
%For 2-dimensional CAT(0) polyhedral complexes, we give: 
%(1) an algorithm for computing convex hulls using linear programming;  
%and (2) an algorithm for answering shortest path queries from a given source point. 
\end{abstract}

%%%%%%%%%%%%%%%%%%%%%%%%%%%%%%%%%%%%%%%%%%%%%%%%%%%%%%%%
%%%%%%%  INTRODUCTION
%%%%%%%%%%%%%%%%%%%%%%%%%%%%%%%%%%%%%%%%%%%%%%%%%%%%%%%%
\section{Introduction}
\label{sec:introduction}

Convex hulls and shortest paths---and algorithms to find them---are very well understood in Euclidean spaces, but less so in non-Euclidean spaces.
%Algorithms to find  convex hulls and shortest paths are very well-studied in Euclidean spaces, but less so in non-Euclidean spaces.
We consider these two problems in finite \emph{polyhedral complexes} which are formed by joining a finite number of $d$-dimensional convex polyhedra along lower dimensional faces.  We will primarily be concerned with the 2D case of triangles or rectangles joined at edges.

We will restrict our attention to polyhedral complexes that are 
globally non-positively curved, or CAT(0).  
Introduced by Gromov in 1987~\cite{Gromov}, CAT(0) metric spaces (or spaces of global non-positive curvature) constitute a far-reaching common generalization of Euclidean spaces, hyperbolic spaces and simple polygons. 
The initials ``CAT''  stand for Cartan, Alexandrov, and Toponogov, three researchers who made substantial contributions to the theory of comparison geometry.  In a CAT(0) space, in contrast to a space of positive curvature, there is a unique geodesic (locally shortest path) between any two points and this property characterizes CAT(0) complexes.
% and explains why we have chosen this restriction.

The impact of CAT(0) geometry on mathematics 
is significant
%plays a significant role, 
especially in the field of geometric group theory
where the particular case of  CAT(0) polyhedral complexes formed by cubes---the so-called ``CAT(0) cube complexes''---are particularly relevant~\revised{\cite{agol2013virtual,HaglundWise,sageev1995ends}.}
%\cite{BridsonHaefliger, Gromov, HaglundWise}.
%The particular case of a CAT(0) polyhedral complex formed by cubes---a so-called ``CAT(0) cube complex''--- 
%has recently received 
%much attention from the geometric group theory community \cite{BridsonHaefliger, Gromov, HaglundWise}. 
 Most of the work on CAT(0) metric spaces so far has been mathematical.
 Algorithmic aspects remain relatively unexplored 
apart from a few results for some particular CAT(0) spaces~\cite{Ch_CAT,ChepoiDraganVaxes,ElderMcCammond,Maftuleac}. 
%MO:  removed BauesPeyerimhoff from the above reference list, since they seem to be talking about a different type of curvature (on graphs) 
% the following sentence are word-for-word from Daniela's paper, so we need new wording here. AL 
%Presently, most of the known results on CAT(0) metric spaces are mathematical. To the best of our knowledge, from the algorithmic point of view, these spaces remain relatively unexplored. Still there are algorithmic results for some particular CAT(0) spaces \cite{BauesPeyerimhoff,Ch_CAT,ChepoiDraganVaxes}. %\note{check have right reference for ChepoiDraganVaxes}.

%However, 
This paper is about algorithms for finite CAT(0) polyhedral complexes, which we will call ``CAT(0) complexes'' from now on.
We are primarily interested in the algorithmic properties of CAT(0) 
complexes because of their applications, particularly to computational evolutionary biology.  The (moduli) space of all phylogenetic (evolutionary) trees with $n$ leaves can be modelled as a CAT(0) cube complex with a single vertex~\cite{BHV}, and being able to compute convex hulls in this space would give a method for computing confidence intervals for sets of trees (see Section~\ref{s:tree_applications} for more details).  A second application of CAT(0) cube complexes is to reconfigurable systems \cite{GhristPeterson}, a large family of systems which change according to some local rules, e.g.~robotic motion planning, the motion of non-colliding particles in a graph, and phylogenetic tree mutation, etc.
In many reconfigurable systems, the parameter space of all possible positions of the system can be seen as a CAT(0) cube complex \cite{GhristPeterson}.  \revised{CAT(0) cube complexes are also in bijection with median graphs \cite{Ch_CAT}, which have been applied to phylogenetics \cite{bandelt1999median} as well, and with domains of event structures \cite{barthelemy1993median}.}
%\note{AL.  With a single vertex?  MO:  No, generally there are many vertices (and they have no special significance).}
%  and computing geodesics in these complexes is equivalent to finding the optimal way to get the system from one position to another one under the corresponding metric.  \note{leaving out since usually don't use the L2 metric}

% SUMMARY OF OUR RESULTS
%Since locally shortest paths are unique in CAT(0) spaces, it is not so surprising that finding shortest paths is relatively easy for 2D CAT(0) complexes.

\medskip\noindent{\bf Main Results.}
In this paper we study the shortest path problem and the convex hull problem in 2D CAT(0) complexes formed by triangles or rectangles.
%
%We give an algorithm to solve the single-source shortest path problem in 2D CAT(0) complexes.
%Given a 2D CAT(0) complex of $n$ triangles and a source point $s$, our algorithm uses $O(n^2)$ preprocessing time and $O(n)$ storage, and finds the shortest path from $s$ to any query point $t$ in time proportional to the number of triangles traversed by the path.
%This generalizes and improves two previous results: an algorithm by Maftuleac~\cite{Maftuleac} for the case of 
%2D CAT(0) complexes where every edge is in at most two faces; and an algorithm by 
%Chepoi and Maftuleac~\cite{Chepoi-Maftuleac} for the case of 2D CAT(0) rectangular complexes.
%
%The problem of finding the convex hull of a set of points in a 2D CAT(0) complex is 
%%significantly more 
%challenging.
For any set of points $P$ we define the \emph{convex hull} to be
the minimal set containing $P$ that is closed under taking the shortest path between any two points in the set.
We show that convex hulls in 2D CAT(0) complexes fail to satisfy some of the properties we take for granted in Euclidean spaces.
Our main result is an 
%\changed{``efficient''}
%a polynomial time 
algorithm to find the convex hull of a finite set of points in a 2D CAT(0) complex with a single vertex.   In general, for any CAT(0) complex, the convex hull of a set of points is the union of a convex set in each cell of the complex.  For the case of 2D CAT(0) complexes, these convex sets are polygons \revised{which may be open or closed on parts of their boundaries}.  
%\rednote{Should we say ``open or closed on boundary''?  Megan:  yes, probably.}
For the special case when there is a single vertex, we show how to find these polygons using linear programming.
Our algorithm runs in polynomial time (in bit complexity) %when we can compute sines of the angles of the complex in polynomial time---in particular, it runs in polynomial time 
for  a cube complex. %, which has angles of $90^\circ$.
For more general inputs we must use the real-RAM model of computation, and the bottleneck in our running time is the time required for linear programming in an algebraic model, which is not known to be polynomially bounded, but is considered efficient via the simplex method.

\changed{In the single-source shortest path problem, we are given 
a 2D CAT(0) complex of $n$ triangles and a source point $s$, and we wish to preprocess the complex in order to find the shortest path from $s$ to any query point $t$ quickly.  
We explore the shortest path map, which divides the space into regions where shortest paths from $s$ are combinatorially the same (i.e.~traverse the same sequence of edges and faces).  We show that the shortest path map may have exponential size.  An alternative, the ``last step shortest path map,'' has linear size and can be used to find shortest paths from $s$ in time proportional to the number of faces traversed by the path.
We show how to construct the last step shortest path map in $O(n^2)$ preprocessing time and space for special cases.    
This generalizes \note{and unifies} two previous results: 
an algorithm by Chepoi and Maftuleac~\cite{Chepoi-Maftuleac} for the case of 2D CAT(0) rectangular complexes; and \note{an algorithm by Chen and Han~\cite{Chen-Han} specialized to the case of a 2D CAT(0) complex that is a \emph{topological 2-manifold with boundary} (i.e.~every edge is incident to at most two faces). } 
}
%\rednote{need to refer to Chen and Han instead}
%an algorithm by Maftuleac~\cite{Maftuleac} for the case of a 2D CAT(0) complex that is a \emph{topological 2-manifold with boundary} (i.e.~every edge is incident to at most two faces). 
%\note{\sout{In the latter case we improve the time bounds.}}}

The rest of the paper is organized as follows.  Section~\ref{sec:background} contains further background on the problem, including existing algorithmic results for CAT(0) polyhedral complexes and applications to phylogenetics.  Section~\ref{sec:preliminaries} reviews the relevant mathematics and tree space notation.  Section~\ref{sec:convex-hull} gives our results for convex hulls, and Section~\ref{sec:shortest-paths} gives our results for shortest paths in 2D CAT(0) polyhedral complexes.  Finally we give our conclusions in Section~\ref{sec:conclusions}.

\section{Background}
\label{sec:background}

In this section we describe background work on shortest path and convex hull algorithms, and discuss the application of our work to phylogenetic trees. 

One of the most basic CAT(0) spaces is 
%the interior of 
any simple polygon (interior plus boundary) in the plane. 
This can be viewed as a 2D CAT(0) complex once the polygon is triangulated. 
%This is a 2D CAT(0) complex with one cell.  
The fact that geodesic paths are unique is at the heart of efficient algorithms for shortest paths and related problems.
On the other hand, generalizing a polygon to a polygonal domain (a polygon with holes) or a polyhedral terrain yields spaces that are not CAT(0), since geodesic paths are no longer unique.  
This helps explain why shortest path and convex hull problems are more difficult in these more general settings.

\subsection{Shortest Paths}

The shortest path problem is a fundamental algorithmic problem with many applications, both in discrete settings like graphs and networks (see, e.g.,~\cite{AhujaMagnantiOrlin}) as well as in geometric settings like polygons, polyhedral surfaces, or 3-dimensional space with obstacles (see, e.g.,~Mitchell \cite{Mitchell}).  
%\note{These first sentences are word-for-word from Chepoi-Maftuleac, so we should change them.}
%The shortest path problem is one of the best-known algorithmic problems with many applications in routing, %robotics, operations research, motion planning, urban transportation, and terrain navigation.
%This fundamental problem has been intensively studied both in discrete settings like graphs and networks (see, e.g., \cite{AhujaMagnantiOrlin}) as well as in geometric spaces (simple polygons, polygonal domains with obstacles, polyhedral surfaces, terrains; see, e.g., Mitchell \cite{Mitchell}).

All variants of the shortest path problem can be solved efficiently for a polygon once it is triangulated, and triangulation can be done in linear time with Chazelle's algorithm~\cite{Chazelle}.  The shortest path (the unique geodesic) between two given points can be found in linear time~\cite{LeePreparata}.  
For query versions, linear space and linear preprocessing time allow us to  
answer single-source queries~\cite{Guibas-sh-path-87,Hershberger-Snoeyink}, and more general all-pairs queries~\cite{GuibasHershberger}, where answering a query means returning the distance in 
logarithmic time, and the actual path in time proportional to its number of edges. 

By contrast, in a polygonal domain, where geodesic paths are no longer unique, the best 
single-source shortest path algorithm uses a continuous-Dijkstra approach in which paths are explored by order of distance.  For a polygonal domain of $n$ vertices, this method takes $O(n \log n)$ [preprocessing] time~\cite{HershbergerSuri} (see the survey by Mitchell~\cite{Mitchell}).
For a polyhedral terrain the continuous-Dijkstra approach gives $O(n^2 \log n)$ time~\cite{Mitchell-Mount-Papadimitriou}, and the best-known run-time of $O(n^2)$ is achieved  by Chen and Han's algorithm~\cite{Chen-Han} that uses a breadth-first-search approach.

There are no shortest path algorithms 
for the general setting of CAT(0) polyhedral complexes, although there are some for certain specializations.
For  2D CAT(0) complexes that are 2-manifolds,
\note{the algorithm of Chen and Han~\cite{Chen-Han} applies (in fact, they do not need the CAT(0) property), and solves the single-source shortest 
path problem with preprocessing time $O(n^2)$, space $O(n)$ and query time (to produce the path) proportional to the size of the output path.
Maftuleac~\cite{Maftuleac} explicitly discussed this as a problem in a CAT(0) space and 
gave a Dijkstra-like algorithm with the same space and query time, but preprocessing time of $O(n^2 \log n)$.}
%
%for the single-source shortest path query problem, with preprocessing time $O(n^2 \log n)$, space $O(n^2)$ and query time proportional to the size of the output path.
%\note{This special case can also be solved using Chan and Han's shortest path algorithm~\cite{Chen-Han} with preprocessing time $O(n^2)$, space $O(n)$ and query time (to produce the path) proportional to the size of the output path.}
Chepoi and Maftuleac \cite{Chepoi-Maftuleac} used different methods to give a polynomial time algorithm for all-pair shortest path queries in any 2D CAT(0) rectangular complex, with preprocessing time $O(n^2)$, space $O(n^2)$, and query time  
proportional to the size of the output path.

%Maftuleac \cite{Maftuleac} presented efficient algorithms for single-point distance queries and convex hulls in CAT(0) planar 2D complexes, where no edge is in more than two faces. Given a source point, the single-point distance queries algorithm uses as a preprocessing step the construction of a geodesic map consisting of geodesics connecting the source point to all the vertices of the complex. This geodesic map can be seen as a collection of 2D cones. The query point is localized in one of the cones and the shortest path connecting it to the source point is then efficiently computed by embedding the cone in the plane as an acute triangle.

% Anna: how about sticking to historical order here?
There are also some %partial 
results on finding shortest paths when we restrict the CAT(0) polyhedral complex to be composed of cubes or rectangles.
The space of phylogenetic trees mentioned in the introduction is a \megan{special type of} CAT(0) cube complex.  For these ``tree spaces,'' Owen and Provan~\cite{OwenProvan} %\megan{had previously given}
gave an algorithm to compute shortest paths (geodesics) with a running time of 
$O(d^4)$, where $d$ is the dimension of the maximal cubes.  The algorithm is much faster in practice for realistic phylogenetic trees.
The result was extended to a polynomial time algorithm for computing geodesics in any \emph{orthant space}~\cite{MillerOwenProvan}, where an \emph{orthant space} is a CAT(0) cube complex with a single vertex.
%\remove{The algorithm for tree space was also adapted by Ardila et al.~\cite{ArdilaOwenSullivant} to 
%compute the shortest path between two points in a general CAT(0) cube complex, however this algorithm is likely not polynomial time.}
Very recently, Hayashi~\cite{hayashi2018polynomial} gave a polynomial time algorithm to compute approximate shortest paths in a general CAT(0) cube complex.  \revised{In fact, his algorithm will be polynomial time in any CAT(0) space in which certain conditions are met, mainly being able to compute shortest paths between any pair of points within a fixed distance $D$ of each other in polynomial time, and having an initial path that consists of a sequence of shortest paths, each of length less than $D$.  It is not clear how to meet his preconditions in a 2D CAT(0) complex.}
%In fact, his algorithm applies to any CAT(0) space in which distances between close points can be computed exactly.
It is not possible to compute exact shortest paths in CAT(0) cube complexes of dimension three or more because Ardila et al. \cite{ArdilaOwenSullivant} showed that, in general, the coordinates of the points where geodesics crosses orthant boundaries are the solutions to higher ordered algebraic equations, and thus cannot be expressed by closed form formulas.
%\megan{If we restrict the CAT(0) polyhedral complex to contain only cubes, then a recent result by Hayashi \cite{hayashi2017polynomial} shows that shortest paths can be computed in polynomial time.} 

%%%%%%%%%%%%%%%%%%%%%%%%%%%%%%%%%%%%%%%%%
\subsection{Convex Hulls}

The problem of computing the convex hull of a set of points is fundamental to geometric computing, especially because of the connection to Voronoi diagrams and Delaunay triangulations~\cite{CHDelaunayVoronoi}.  

The convex hull of a set of points in the plane can be found in provably optimal time $O(n \log h)$ where $n$ is the number of points and $h$ is the number of points on the convex hull.  The first such algorithm was developed by Kirkpatrick and Seidel \cite{KirkpatrickSeidel} and a simpler algorithm was given by Chan \cite{ChanCH}.
An optimal algorithm for computing convex hulls in higher dimension $d$ in time $O(n^{\lfloor \frac{d}{2} \rfloor})$ was given by Chazelle \cite{Chazelle93}.  See the survey by Seidel~\cite{Seidel}.

A simple polygon, triangulated by chords, is the most basic example of a 2D CAT(0) complex.  
In this setting, the convex hull of a set of points $P$ (i.e.~the smallest set containing $P$ and closed under taking geodesics) is referred to as the \emph{relative} (or \emph{geodesic}) convex hull.  
Toussaint gave an $O(n \log n)$ time algorithm to compute the relative convex hull of a set of points in a simple polygon~\cite{Toussaint1}, and studied properties of such convex hulls~\cite{Toussaint2}. 
Ishaque and T\'oth~\cite{Ishaque-Toth} considered the case of line segments that separate the plane into simply connected regions (thus forming a CAT(0) space) and gave an semi-dynamic algorithm to maintain the convex hull of a set of points as  line segments are added and points are deleted.

Moving beyond polygons to polygonal domains or terrains, geodesic paths are no longer unique, so there is no single natural definition of convex hull (one could take the closure under geodesic paths, or the closure under shortest paths).  We are unaware of algorithmic work on these variants. 

However, in a polyhedral surface with unique geodesics
%---that is, in a planar 2D CAT(0) complex---
the convex hull is well defined, and Maftuleac~\cite{Maftuleac} gave an algorithm to compute the convex hull of a set of points 
in $O(n^2 \log n)$ time, where $n$ is the number of vertices in the complex plus the number of points in the set.
%$O(n^2 \log n + nk \log k)$ time, using a data structure of size $O(n^2 + k)$, where $n$ is the number of vertices in the complex and $k$ is the number of points. 
%\note{citation.  Also, why space?} 

In all the above cases the boundary of the convex hull is composed of segments of shortest paths between the given points, which---as we shall see in Section~\ref{sec:CH-examples}---is not true in our setting of 2D CAT(0) complexes.

Beyond polyhedral complexes, convex hulls become much more complicated.  Indeed it is still an open question if the convex hull of 3 points on a general Riemannian manifold of dimension 3 or higher is closed \cite[Note~27]{Berger}.  Bowditch~\cite{Bowditch} and Borb\'ely~\cite{Borbely} give some results for convex hulls on manifolds of pinched negative curvature, but our complexes need not be manifolds.  In the space of positive definite matrices, which is a CAT(0) Riemannian manifold, Fletcher et al.~\cite{horoballs} give an algorithm to compute generalized convex hulls using horoballs, which are generalized half-spaces.  
Lin et al. \cite{LinSturmfelsTangYoshida} look at convex hulls of three points in an \megan{orthant space}. %CAT(0) cubical complex that generalizes the space of phylogenetic trees.  
They prove that there are such spaces where the top-dimensional cells have dimension $2d$, and there exist 3 points in the space such that their convex hull contains a $d$-dimensional simplex.   Bridson and Haefliger \cite[Proposition~II.2.9]{BridsonHaefliger} give conditions for when the convex hull of three points in a CAT(0) space is ``flat", or 2-dimensional.  Finally, for a survey of convexity results in complete CAT(0) (aka Hadamard) spaces, of which the space of phylogenetic trees is one, see \cite{Bacakbook}.  \megan{As an alternative to geodesically closed convex hulls in CAT(0) orthant spaces, Nye et al. \cite{nye2017principal} propose the locus of the Fr\'echet mean, which generalizes the Euclidean definition of a convex hull as a weighted combination of points.}

%useful book:  Metric spaces, convexity, and nonpositive curvature by Athanase Papadopoulos  - talks about geodesically convex sets, which is what we have/want.~\cite{papadopoulos}

%%%%%%%%%%%%%%%%%%%%%%%%%%%%%%%%%%%%%%%%%%%%%%
\subsection{Application to phylogenetic trees} \label{s:tree_applications}

While this paper will look at arbitrary 2D complexes with non-positive curvature, our work is motivated by a particular complex with non-positive curvature, namely the space of phylogenetic trees introduced by Billera, Holmes, and Vogtmann \cite{BHV}, called the \emph{BHV tree space}, and described in more detail in Section~\ref{s:BHV_detail}. Phylogenetic trees are ubiquitous in biology, and each one depicts a possible evolutionary history of a set of organisms, represented as the tree's leaves.
Once we fix a set of leaves, the BHV tree space is a complex of Euclidean orthants (the higher dimensional version of quadrants and octants), in which each point in the space represents a different phylogenetic tree on exactly that set of leaves.  
%Moved above: Phylogenetic trees represent the evolutionary history of a set of organisms and are ubiquitous in biology.

One area of active phylogenetics research is how to statistically analyze sets of phylogenetic trees on the same, or roughly the same, set of species.  Such sets can arise in various ways:
from sampling a known distribution of trees, such as that generated by the Yule process \cite{Yule}; from tree inference programs, such as the posterior distribution returned by performing Bayesian inference \cite{mrbayes} or the bootstrap trees from conducting a maximum likelihood search \cite{bootstrap}; or from improvements in genetic sequencing technologies that lead to large sets of gene trees, each of which represents the evolutionary history of a single gene, as opposed to the species' evolutionary history as a whole.  Traditionally, most of the research in this area focused on summarizing the set of trees, although recent work has included computing variance \cite{Bacak, MillerOwenProvan} and principal components \cite{Nye2014}.

It is an open question to find 
a good way to compute confidence regions for a set of phylogenetic trees. \megan{Willis \cite{willis2017confidence} recently proposed a method for constructing confidence sets based on the Central Limit Theorem for BHV tree space \cite{CLT2, CLTcomplete}.}
%Thus we must assume that the underlying distribution is unknown.  
\megan{An alternative, non-parametric approach was proposed by Holmes \cite{Holmes2005}, who suggested applying the data depth approach of peeling convex hulls for Euclidean space \cite{Tukey} to the BHV tree space.} 
\remove{In \cite{Holmes2005},  Holmes proposes applying the data depth approach of Tukey \cite{Tukey} for Euclidean space to the BHV tree space, namely peeling convex hulls. }
The convex hull is the minimum set that contains all the data points, as well as all geodesics between points in the convex hull.  By peeling convex hulls, we mean to compute the convex hull for the data set, and then remove all data points that lie on the convex hull.  This can then be repeated.  To get the 95\% confidence region, for example, one would remove successive convex hulls until only 95\% of the original data points remain.

If we keep peeling convex hulls until all remaining points lie on the boundary of the convex hull, then we can take their Fr\'echet mean \cite{Bacak, MillerOwenProvan} to get an analog of
% until a single point remains, we have 
the univariate median of Tukey \cite{Tukey}.  This could also be a useful one-dimensional summary statistic for a set of trees.  Many of the most-used tree summary statistics have a tendency to yield a degenerate or non-binary tree, which is a tree in which some of the ancestor relationships are undefined.  This is considered a problem by biologists, but such a univariate median tree found by peeling convex hulls would likely be 
%guaranteed to be     MO:  changes in this paragraph are because strictly pealing convex hulls 
binary if all trees in the data set are.

Currently, these methods cannot be used, because it is not known how to compute convex hulls in BHV tree space.  We show several examples of how Euclidean intuition and properties for convex hulls do not carry over to convex hulls in the BHV tree space.  
Our algorithm to find convex hulls applies to the space of trees with five leaves which is described in more detail in Section~\ref{s:BHV_detail}.
%We also give an algorithm for computing convex hulls in the space of trees with five leaves.

\section{Preliminaries}
\label{sec:preliminaries}

%copied from intro.  I think we might already say this here, but if not, we should add it in in the appropriate place
%A polyhedral complex has non-positive curvature if and only if it is simply connected and there is no cycle of faces incident to a vertex such that the sum of the face angles at the vertex is less than $2\pi$ \cite{BridsonHaefliger}[Theorem 5.4, Lemma 5.6].

A metric space $(X,d)$ is \emph{geodesic} if every 
two points $x,y \in X$ are connected by a locally shortest or \emph{geodesic} path.
% AL: I don't think "geodesic" path means "shortest path"  MO:  yes, you are right
%pair of points $x,y \in X$ are connected by a geodesic, or shortest path.
A geodesic metric space $X$ is \emph{CAT(0)} if 
its triangles satisfy the following
CAT(0) inequality.  For any triangle $ABC$ in $X$ with geodesic segments for its sides, construct a 
\emph{comparison triangle} $A'B'C'$ in the Euclidean plane such $d(A,B) = \changed{|A'B'|}$, $d(A,C) = \changed{|A'C'|}$, and $d(B,C) = \changed{|B'C'|}$. Let $Y$ be a point on the geodesic between $B$ and $C$, and let $Y'$ be a \emph{comparison point} on the line between $B'$ and $C'$ such that $d(B,Y) = \changed{|B'Y'|}$.  Then triangle $ABC$ satisfies the CAT(0) inequality if 
%and only if    -- only use "if" for definitions. AL.
$d(Y, A) \leq \changed{|Y'A'|}$ for any $Y$.  Intuitively, this corresponds to all triangles in $X$ being at least as skinny as the corresponding triangle in Euclidean space (Figure~\ref{fig:comparison_triangles}).

%\remove{
%\emph{comparison triangle} $\bar{A}\bar{B}\bar{C}$ in the Euclidean plane such $d(A,B) = d(\bar{A},\bar{B})$, $d(A,C) = d(\bar{A},\bar{C})$, and $d(B,C) = d(\bar{B},\bar{C})$.  Let $Y$ be a point on the geodesic between $A$ and $B$, and let $\bar{Y}$ be a \emph{comparison point} on the line between $\bar{A}$ and $\bar{B}$ such that $d(A,Y) = d(\bar{A},\bar{Y})$.  Then triangle $ABC$ satisfies the CAT(0) inequality if 
%and only if    -- only use "if" for definitions. AL.
%$d(Y, C) \leq d(\bar{Y},\bar{C})$ for any $Y$.  Intuitively, this corresponds to each triangle in $X$ being at least as skinny as the corresponding triangle in Euclidean space (Figure~\ref{fig:comparison_triangles}).
%}

\begin{figure}[h]
\centering
\includegraphics[width=4in]{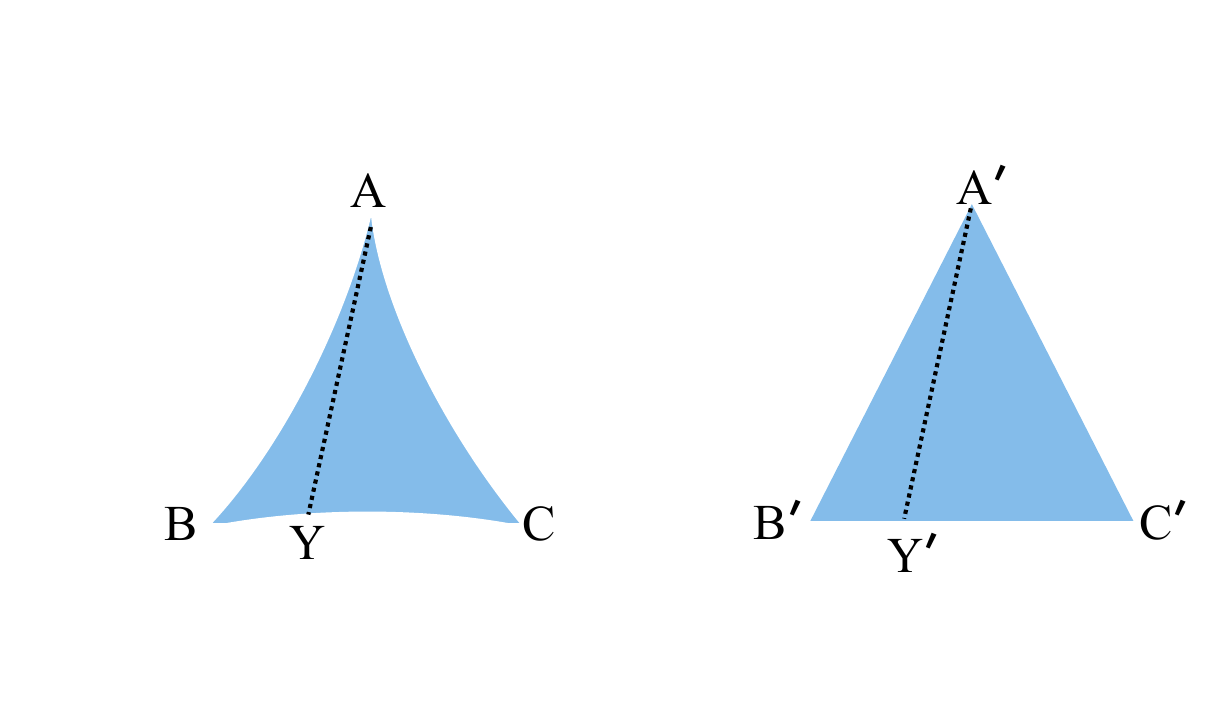}
\caption{The triangle on the left represents a triangle in a CAT(0) space, with its corresponding comparison triangle in Euclidean space on the right.}
\label{fig:comparison_triangles}
\end{figure}

A \emph{polyhedral complex} is a set of convex polyhedra (``cells'') glued together by isometries along their faces.  In this paper we only consider finite polyhedral complexes.
When all of the cells are cubes, then this is called a \emph{cube}, or \emph{cubical}, \emph{complex}.
The length of a path between two points in a polyhedral complex is the sum of the Euclidean lengths of the pieces of the path in each cell of the complex.   The \emph{distance} between two points is defined as the length of the shortest path between them.
%\changed{The sequence of cells that a path passes through is called the  \emph{combinatorial type} of the path.}  %\note{We never use combinatorial type in the rest of the paper, although in the last section we talk about combinatorial shortest paths.  Remove this last sentence?}
% moved up to this paragraph where we talk about paths.
%Finally in a CAT(0) complex, a shortest path passes through some sequence of cells.  We call this sequence of cells the \emph{combinatorial type} of the shortest path.

%
%Within a polyhedral complex, the length of a path between two points is the sum of the lengths of the restriction of that path to each polyhedra in turn, where it is computed using the Euclidean metric.  The distance between two points is defined as the length of the shortest path between them.}  % realized we never explicitly defined our metric

%If $X$ is a cube complex, then Gromov \cite{Gromov} gave an alternative characterization to the CAT(0) inequality:

%\note{Are we using this anywhere? I doubt it, since we never address cubes specifically.  So let's leave it out.}

%\begin{theorem}[Gromov] A cube complex $X$ is CAT(0) if and only if it is simply connected and for every vertex $v \in X$, if $v$ is incident to $k$ edges, any pair of which specify a square, then those $k$ edges also specify a $k$-dimensional cube in $X$.
% alternatively:  simply connected and the link of any vertex is a flag simplicial complex
%\end{theorem}

We will consider polyhedral
complexes that are CAT(0).  The cells are 2D (planar) convex polygons.  These can always be triangulated, so the general setting is when the cells are triangles.
We call this a \emph{2D CAT(0) complex}.
Sometimes we will consider a complex in which all the cells are rectangles, either bounded or unbounded.  We call this a \emph{2D CAT(0) rectangular complex}.
In general we allow the space to have boundary (i.e., an edge that is incident to only one cell).

\note{
%Model of computation and how the input is given.
A 2D CAT(0) complex can be specified by giving its combinatorial information (the list of vertices, edges, and cells together with their incidence relationships) and the geometry of each cell. 
Each cell is a convex polygon that can be given either in local coordinates, or as angles and edge lengths.  
Converting between these two representations requires a real RAM model of computation in general, though for rectangular complexes, the conversion is efficient as measured by bit complexity, since all angles are $90^\circ$.
}

\revised{If a geodesic path in a 2D CAT(0) complex travels through a sequence of cells, we can ``unfold'' the cells into the plane where the path becomes a straight line.    
This result, and its history, is thoroughly discussed by Mitchell et al.~\cite{mitchell1987discrete}, for the case of a 2-manifold (a polyhedral surface). 
The same is true in general 2D CAT(0) complexes since a geodesic path can never revisit an edge, i.e., the sequence of cells traversed by a geodesic path forms a 2-manifold.
}

For any vertex $v$ of a 2D complex, we define the \emph{link graph}, $G_v$ as follows.
The vertices of $G_v$ correspond to the edges incident to $v$ in the complex.
The edges of $G_v$ correspond to the cells incident to $v$ in the complex: if $r$ and $s$ are edges of cell $C$ with $r$ and $s$ incident to $v$, then we add an edge  between vertices $r$ and $s$ in $G_v$ with weight equal to the angle between $r$ and $s$ in $C$.  Every point $p \ne v$ in cell $C$ can be mapped to a point on edge $(r,s)$ in $G_v$: if the angle between $vr$ and $vp$ in $C$ is $\alpha$ then $p$ corresponds to the point along edge $(r,s)$ that is distance $\alpha$ from $r$.
%\note{check if you still do} We use the same notation for a point of $C$ and the corresponding point in the link graph when the meaning is clear from the context.

%We can map every point in the complex except $O$ to a point in $L$:
%for a point on an edge of the complex we take the corresponding
%vertex in $L$, and for a point $p$ interior to a cell $C$ with edges $r$ and $s$, we measure the angular distance $\alpha$ from $r$ to $p$,  and in $L$ we take the point along the edge $rs$ that is distance $\alpha$ from $r$.
%We use the same notation for a point of $P$ and the corresponding point in the link graph when the meaning is clear from the context.
%Every point in $L$ corresponds to a ray of points in the complex.

When we have a 2D polyhedral complex, there is also an alternative condition for determining whether it is CAT(0).

% we don't strictly need finite - we just need that there are a finite number of isometry classes of cells in the complex
% also note that injective is just used to ensure we don't go around a loop twice in the link graph
% This theorem is a combination of Theorem 5.4 and Lemma 5.6 in B&H.

%\begin{theorem}[{{\cite[Theorem II.5.4 and Lemma II.5.6]{BridsonHaefliger}}}]
%A 2-dimensional Euclidean complex $\cal K$ is CAT(0) if and only if it is simply connected and for every vertex $v \in {\cal K}$, every injective loop in the link graph $G_v$ has length at least $2\pi$.
%\end{theorem}

\begin{theorem}[{{\cite[Theorem II.5.4 and Lemma II.5.6]{BridsonHaefliger}}}]
\label{thm:CAT0}
A 2D polyhedral % we have not defined piecewise Euclidean (We could if you want to . . )  MO:  no, I think polyhedral is fine
%piecewise Euclidean 
complex $\cal K$ is CAT(0) if and only if it is simply connected and for every vertex $v \in {\cal K}$, every cycle in the link graph $G_v$ has length at least $2\pi$.
\end{theorem}

\remove{For example, four quadrants arranged in a circle around the origin correspond to the plane, which is CAT(0) and has a link graph of the origin consisting of one cycle of length $2\pi$.  In contrast, three quadrants arranged in a circle around the origin correspond to the corner of a room, and are not CAT(0).  Here, the link graph of the origin is one cycle of length $\frac{3\pi}{2}$, and there are two shortest paths from any point on an edge to any point in the middle of the opposite quadrant, with each path passes through one of the two quadrants adjacent to that edge.}

% thanks for the explanation A.L.
%\note{Do we need ``finite'' -- aren't we assuming that all the time?  Also, what is an ``injective loop''?  Can we restate the result for our single vertex case? --- the link graph has no cycle of length $< 2\pi$.  This is what we use later on, so we really should state it here. A.L.   MO:  I added 'finite' in a few places in the intro to make it clear we are talking about finite complexes, and then removed it here.  Geometric group theorist sometimes consider CAT(0) cube complexes with an infinite number of cubes, so that's why the original theorem specified 'finite'.  As for 'injective loop', they want to make sure that you don't get a loop of $2\pi$ by going around a cycle twice.  This is routinely done in basic algebraic topology, so probably why they are being careful about it here.  I changed "injective loop" to "cycle".}

Some of our results are only for the case where the CAT(0) complex $\cal K$ has a single vertex $O$, which we call the \emph{origin}.  We will call such a complex a \emph{single-vertex complex}.
%\fillin{Do we want a name for these complexes?}
In a 2D single-vertex complex
every cell is a \emph{cone} formed by two edges incident to $O$ with angle at most $\pi$ between them.  There is a single link graph $G = G_O$.  Every point $p$ of the complex except $O$ corresponds to a point $\lambda(p)$ of $G$, and every point of $G$ corresponds to a ray of points in the complex.

\subsection{BHV Tree Space} \label{s:BHV_detail}
As explained in Section~\ref{s:tree_applications} the work on computing convex hulls was motivated by the BHV tree space for trees with 5 leaves, which is a 2D CAT(0) complex.  We will now describe this space, which is denoted $\T_5$, and which contains
all unrooted leaf-labelled, edge-weighted phylogenetic trees with 5 leaves 
(equivalently all such \emph{rooted} trees with 4 leaves).  % I hate footnotes!  they are so distracting.
%\footnote{This space is the same as the space of all \emph{rooted} leaf-labelled, edge-weighted phylogenetic trees with 4 leaves.}. 
For a description of the BHV tree spaces for trees with more than 5 leaves,
%where the trees $n > 5$ leaves and for more details, 
see \cite{BHV}. 
This section is not necessary for understanding the rest of the paper.

A \emph{phylogenetic tree} is a tree in which each interior vertex has degree $\geq 3$ and there is a one-to-one labelling between the leaves (degree 1 vertices) and some set of labels $\cal L$. For this paper, we assume ${\cal L} = \{1,2,3,4,5\}$.  Also, the trees have a positive weight or length on each \emph{interior edge}, which is an edge whose vertices \changed{have} degree $\geq 3$ (that is, are not leaves).  If a phylogenetic tree contains only vertices of degree 1 and 3, then it is called \emph{binary}.

\begin{figure}[htb]
\centering
\includegraphics[width=5in]{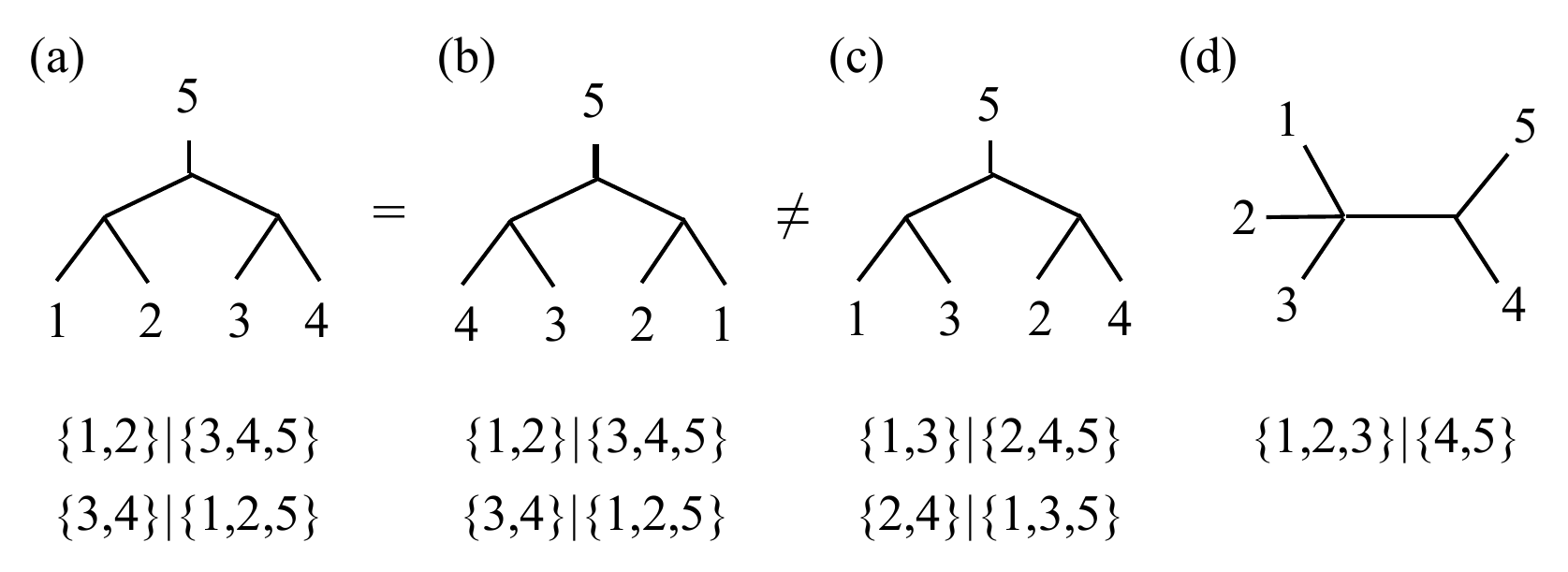}
\caption{Several tree shapes with their constituent splits listed below them.  Tree (a) and tree (b) have the same tree shape, which does not depend on the planar embedding of the tree.  Tree (c) has different splits and hence a different tree shape from trees (a) and (b).  Tree (d) is a non-binary tree shape.}
\label{fig:treeshapes}
\end{figure}

A \emph{split} is a partition of the leaf-set $\cal L$ into two parts $L_1 \cup L_2 = \cal L$ such that $|L_i| \geq 2$ for $i = 1,2$. We write a split as $L_1 | L_2$.  Each interior edge in a phylogenetic tree corresponds to a unique split, where the two parts are the sets of leaves in the two subtrees formed by removing that edge from the tree.  Binary trees with five leaves contain two interior edges, and hence two splits.  There are 10 possible splits, and they can be combined to form 15 different tree shapes.  The \emph{shape} of a tree is defined to be the set of interior edges or splits in that tree, and tell us which species are most closely related.
%\note{Tree shape is not defined.} 
(Figure 2).

%The \emph{shape}\footnote{A tree shape is sometimes also called a tree topology, despite it not being a formal topology.} of a tree is given by the set of interior edges or splits in that tree.  Alternatively, remembering that these trees represent an evolutionary history of the species represented by the leaves, the shape of the tree refers to which species are most closely related (Figure~\ref{fig:treeshapes}).  Not all pairs of splits form a tree.  For example, there are no trees containing edges corresponding to the splits $\{1,2\}|\{3,4,5\}$ and $\{1,3\}|\{2,4,5\}$, because leaf $1$ cannot be most closely related to both leaves $2$ and $3$.  (The tree in which they are all just as closely related would only have the split $\{1,2,3\}|\{4,5\}$, as shown in Figure~\ref{fig:treeshapes}d.)  In this paper, we ignore the edges ending in leaves.  However, if these edges also have lengths, then they can be represented by a point in $\bR^5_{\geq 0}$, and the full tree space becomes the product $\T_5 \times \bR^5_{\geq 0}$.

\begin{figure}[htb]
\centering
\includegraphics[width=5in]{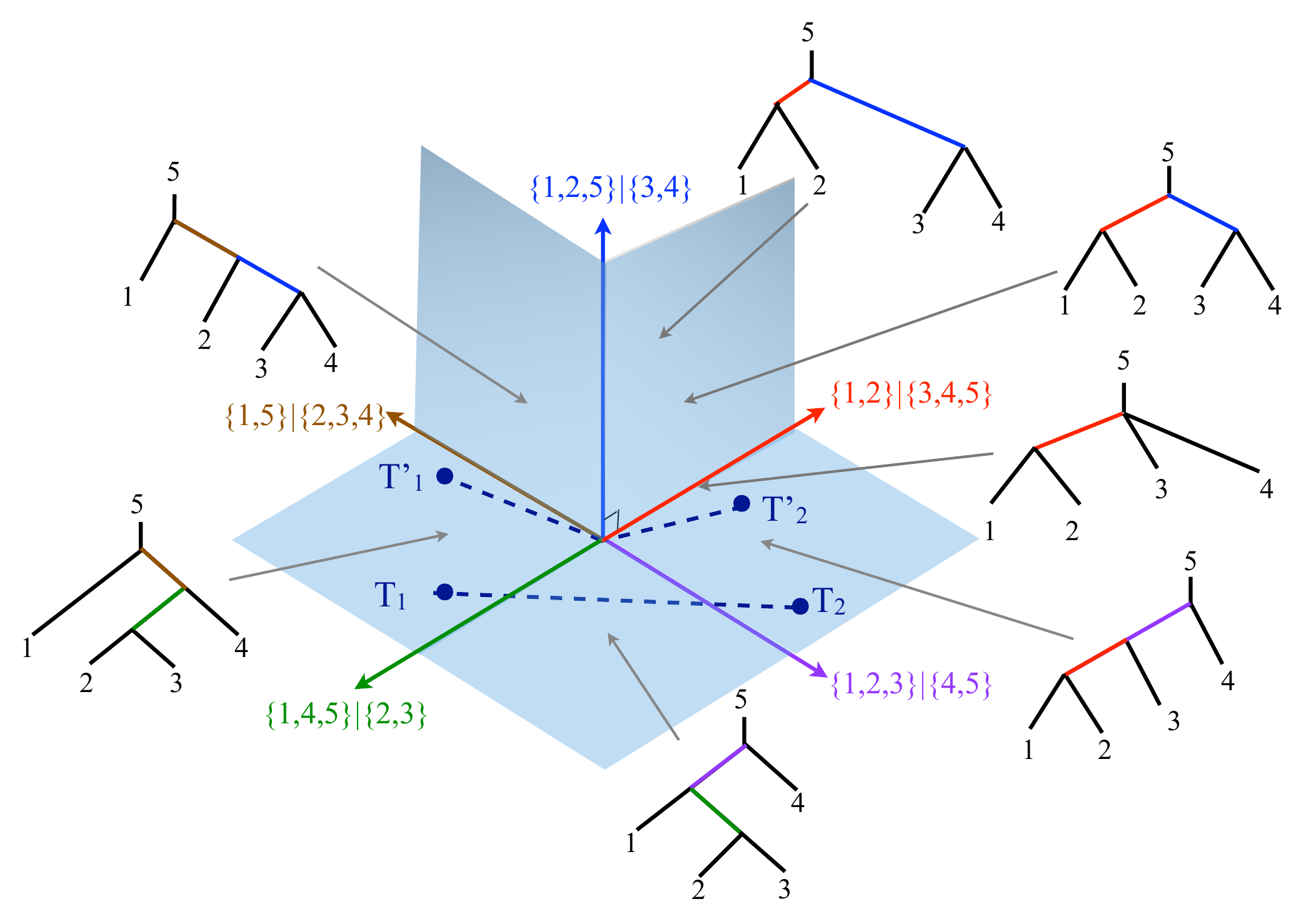}
\caption{Five quadrants in $\T_5$.  
\remove{Technically each split has its own dimension, and these quadrants sit in $\bR^5$.   \note{I don't understand this, but that might be my ignorance.}  However, for ease of visualization, we have drawn two pairs of splits on the same axes.  \note{I only see one split on each axis (??) MO: Each split has its own dimension, and there are five splits, so that means these five quadrants are sitting in 5-dimensional space.  But I drew them in 3-dimensional space, since each split only uses the positive part of its axis, so you can double some of them up (i.e. put a different split on the negative part of the axis).  However, maybe for a geometry audience, it's clear this is what we are doing, so the explanation just make it confusing?}} 
The upper right quadrant illustrates how trees with the same tree shape but different edge lengths lie at different coordinates in the quadrant.  The dashed lines represent geodesics between two pairs of trees where $T_i$ and $T_i'$ have the same tree shape but different edges lengths. %with the same tree shapes, but different edge lengths. 
The geodesic $T_1'$ to $T_2'$ passes through the origin, while the geodesic $T_1$ to $T_2$ passes through an intermediate quadrant. 
%One of these geodesics passes through the origin, while the other passes through an intermediate quadrant.
}
\label{fig:5orthants}
\end{figure}

We now define the space $\T_5$ itself, which consists of exactly one Euclidean quadrant for each of the 15 possible binary tree shapes. 
%Each of the 15 possible tree shapes corresponds to a Euclidean quadrant in $\T_5$.
For each quadrant, the two axes are labelled by the two splits in the tree shape.  A point in the quadrant corresponds to the tree with that tree shape whose interior edges have the lengths given by the coordinates.  We identify axes labelled by the same split, so that if two quadrants both contain an axis labelled by the same split, then they are glued together along that shared axis (see Figure~\ref{fig:5orthants}).  
%In this paper, we ignore the edges ending in leaves.  However, if these edges also have lengths, then they can be represented by a point in $\bR^5_{\geq 0}$, and the full tree space becomes the product $\T_5 \times \bR^5_{\geq 0}$.

The length of a path between two trees in $\T_5$ is the sum of the lengths of the restriction of that path to each quadrant in turn, where it is computed using the Euclidean metric. The \emph{BHV distance} is the length of the shortest path, or geodesic, between the two trees (Figure~\ref{fig:5orthants}).  Billera et al. \cite{BHV} proved that this tree space is a CAT(0) cube complex, which implies that there is a unique geodesic between any two trees in the tree space.
%The length of a path between two trees in $\T_5$ is the sum of the lengths of the restriction of that path to each quadrant in turn, where it is computed using the Euclidean metric.  The \emph{BHV distance} is the length of the shortest path, or geodesic, between the two trees (Figure~\ref{fig:5orthants}).  Billera et al. \cite{BHV} proved that this tree space is a CAT(0) cube complex, which implies that there is a unique geodesic between any two trees in the tree space.  In $\T_5$, the geodesic either passes through the origin, through one axis joining two quadrants, through the two axes bounding one quadrant, or remains in the quadrant.  See Figure~\ref{fig:possible_geos_T5} for illustrations of the first three possibilities.

%\begin{figure}[h]
%\centering
%\includegraphics[width=3in]{figures/possible_geodesics_T5.pdf}
%\caption{This figure illustrates three of the possible four geodesic combinatorial types. The fourth, in which the geodesic stays within the quadrant is not shown.  This figure shows three adjacent quadrants in $\T_5$.}
%\label{fig:possible_geos_T5}
%\end{figure}

To understand how the 15 quadrants in $\T_5$ are connected, consider the link graph of the origin, which is shown in Figure~\ref{fig:petersen} and is the Petersen graph.  The Petersen graph has multiple overlapping 5-cycles, one of which corresponds to the 5 quadrants in Figure~\ref{fig:5orthants}.  Also note that each vertex in the link graph of the origin is \revised{incident} to three edges.  This corresponds to each axis lying in three quadrants in $\T_5$.
Note that this example illustrates that the link graph of even a CAT(0) rectangular complex with a single vertex need not be planar.
%\changed{More generally, the link graph of a 2D single-vertex CAT(0) complex can be any metric graph whose girth is at least $2\pi$ and whose edges have length at most $\pi$, by adding one cell for each edge of the graph.}  

\begin{figure}[htb]
\centering
\includegraphics[width=3in]{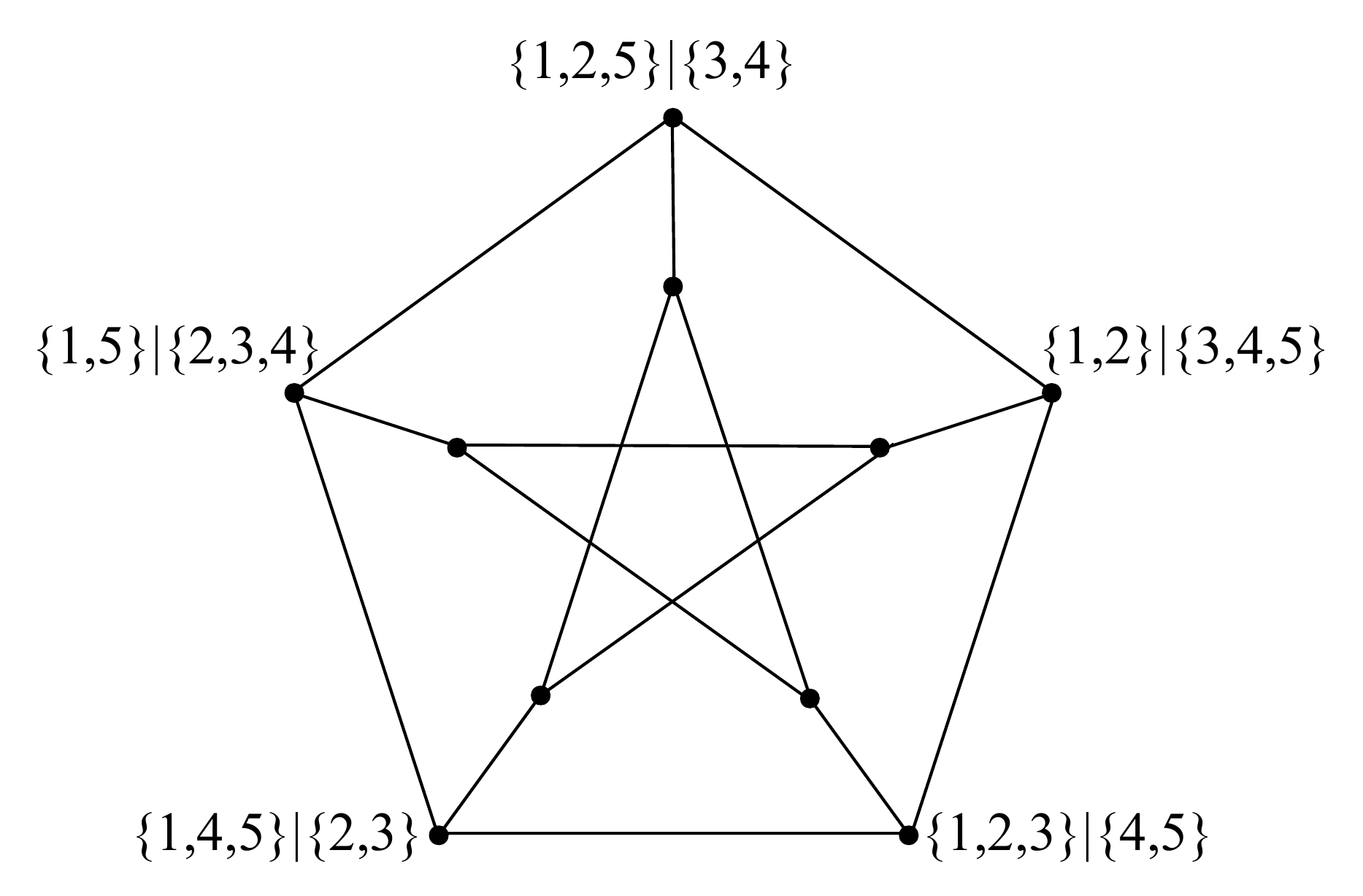}
\caption{The Petersen graph, which is the link graph of the origin of $\T_5$.  The graph edges correspond to tree shapes, and the graph vertices correspond to splits.  The 5 quadrants in Figure~\ref{fig:5orthants} correspond to the outer 5-cycle in the Petersen graph.}
\label{fig:petersen}
\end{figure}

%\newnote{Describe how we can use the Petersen graph to illustrate what angle a point has in the space?   MO:  Just before this subsection in Preliminaries, we explain how to use the link graph to show what angle a point has in  space.   Since we say the Petersen graph is a link graph (and don't have any figures with angles on the Petersen graph), I don't think this needs further explanation. }

%%%%%%%%%%%%%%%%%%%%%%%%%%%%%%%%%%%%%%%%%%%%%%%%%%%%%%%%
%%%%%%%  CONVEX HULLS
%%%%%%%%%%%%%%%%%%%%%%%%%%%%%%%%%%%%%%%%%%%%%%%%%%%%%%%%

\section{Convex Hulls}
\label{sec:convex-hull}

Let $P$ be a finite set of points in a CAT(0) %\changed{polyhedral} %Euclidean
 complex $\cal K$. 
%\note{We do not need say "polyhedral" by note in intro.}
Recall from Section~\ref{sec:introduction} that the convex hull of $P$ is defined to be the minimal set containing $P$ that is closed under taking the shortest path between any two points in the set.
Let CH$(P)$ denote the convex hull of $P$ in $\cal K$.
For algorithmic purposes, there are several ways to specify CH$(P)$.
One possibility is to specify the intersection of CH$(P)$ with each cell of the complex.  In the 
case of 2D CAT(0) complexes, each such set is a convex polygon \revised{which may be open or closed on parts of its boundary.  Our algorithm finds the vertices of each such polygon, and thus finds the closure of $\CH(P)$. } 
%which can be given by its vertices.
%Our algorithm will do this. 
%the sets are polyhedral in general, or might they have curved boundaries?  MO:  I'm really not sure.  At first I thought not (assuming "in general" means arbitrary (instead of just 2D) CAT(0) polyhedral complexes)- that each such set is convex and a polytope, because the convex hull restricted to a cell gives the Euclidean case.  But then maybe weird stuff does happen on the boundaries (since they are not just axes).  I suggest just leaving it at this for the paper. %
We note---although we will not pursue this approach---that
%Although we will not pursue this approach, we note that 
there is another
way to specify CH$(P)$, which might be easier but would still suffice for many applications,
and that is to
%would be to 
give an algorithm to decide if a given query point of $\cal K$ is inside CH$(P)$.
% A.L.  I have commented out the following.  Let's distinguish the definition from the possible ways to finitely represent the convex hull (which is what we need in order to say we have an algorithm to compute the CH).
%Finally, a third possibility is to define the convex hull recursively: $P$ is contained in CH$(P)$ and for every two points $p$ and $q$ in CH$(P)$, the shortest path from $p$ to $q$ is contained in CH$(P)$.  Equivalently, CH$(P)$ is the minimal set containing $P$ that contains the shortest path between any two points of CH$(P)$.

Convex hulls in CAT(0) \revised{spaces} %complexes 
are something of a mystery.  It is not known, for example, whether they are closed sets \cite[Note~27]{Berger}.  
\revised{We do not resolve this, even for our case of a 2D CAT(0) complex with a single vertex.}

We begin in subsection~\ref{sec:CH-examples} by giving some examples to show that various properties of Euclidean convex hulls fail in CAT(0) complexes. 
In subsection~\ref{sec:CH-alg}
we give our main result,  an
%efficient 
%a polynomial time 
algorithm (using linear programming) to find convex hulls in 
any 2D CAT(0) complex with a single vertex $O$. 
Specifically, we prove: 

\begin{theorem}
There is a polynomial-time reduction from the problem of finding the \revised{closure of the} convex hull of a finite set of points $P$ in a 2D CAT(0) complex $\cal K$ with a single vertex $O$ to linear programming.  \changed{The resulting linear program has $O(n+m)$ variables and $O((n+m)^3)$ inequalities, where $n$ is the number of cells in $\cal K$ and $m$ is the number of points in $P$.} 
%an efficient
%a polynomial time 
%algorithm to find the convex hull of a finite set of points $P$ in a 2D CAT(0) complex $\cal K$ with a single vertex $O$.
For the special case of a cube complex this provides a polynomial-time \note{(in bit complexity)} convex hull algorithm.
%The algorithm runs in polynomial time in the special case of a cube complex.
\label{thm:convex-hull-alg}
\end{theorem}
%\note{Boundary?}  Yes, boundary is allowed.

%\revised{\sout{Our proof of Theorem~\ref{thm:convex-hull-alg} 
%implies that the convex hull of a finite set of points in a single-vertex 2D CAT(0) complex is a closed set.}
%}

The idea of our algorithm is to first
use the link graph to test if point $O$ is in the convex hull and to identify the edges of the complex that intersect the convex hull at points other than $O$.  
Then we formulate the exact computation of the convex hull as a linear program whose variables represent the boundary points of the convex hull on the edges of the complex.
There is a polynomial bound on the number of variables and inequalities of the linear program, but whether the linear program can be solved in polynomial time depends on bit complexity issues.
In the general case our reduction uses the real-RAM model of computation, \note{including trigonometric %or square root 
operations.}
%, depending on how the input is given}.
%Our algorithm is efficient but whether it runs in polynomial time depends on bit complexity issues.
There are polynomial-time linear programming algorithms~\cite{khachiyan,karmarkar}, but their run-times depend on the number of bits in the input numbers.  
%Linear programming algorithms run in polynomial time depending on the number of bits of the 
%input~\cite{khachiyan,karmarkar}. 
For cube complexes, which have angles of $90^\circ$, 
\note{our reduction uses standard arithmetic operations}
and the resulting  linear program has coefficients with a polynomial number of bits and so our
convex hull algorithm runs in polynomial time.
However, more generally our algorithm must use the stronger real RAM model of computation in order to perform computations on the angles of the input CAT(0) complex, and we must resort to the simplex method for linear programming~\cite{Dantzig}
\note{which is not known to run in polynomial time}.
%In the special case of a cube complex, our algorithm runs in polynomial time.

\subsection{A Basic Result on Single-Vertex 2D CAT(0) Complexes}
\label{sec:basicCH}

\note{In this section we investigate the correspondence between shortest paths in a 2D CAT(0) complex $\cal K$ with a single vertex $O$ and paths in the link graph $G=G_O$. 
} 

%We can decide whether the geodesic path between two points goes through $O$ by looking at the shortest path between the corresponding points in the link graph, as follows. 
Consider two points $a$ and $b$ in $\cal K$,  distinct from $O$,  and consider the corresponding points 
\note{$\lambda(a)$ and $\lambda(b)$ in $G$.
% (context will distinguish $a$ in $\cal K$ from $a$ in $G$).
Let $\sigma(a,b)$ be the (unique) geodesic path between $a$ and $b$ in the space $\cal K$.  Let $\sigma_G(\lambda(a),\lambda(b))$ be 
%the
\changed{a} shortest path between $\lambda(a)$ and $\lambda(b)$ in $G$.  Let $| \sigma |$ indicate the length of  path $\sigma$.

%\begin{claim}
\begin{proposition}
Exactly one of the following two things holds:
\begin{itemize}
\item $| \sigma_G(\lambda(a),\lambda(b)) | \ge \pi$  and $\sigma(a,b)$ goes through $O$,
\item $| \sigma_G(\lambda(a),\lambda(b)) | < \pi$  and $\sigma(a,b)$ maps to $\sigma_G(\lambda(a),\lambda(b))$ and does not go through $O$.
\end{itemize}
\label{property:path-thru-O}
%\end{claim}
\end{proposition}
For an example, see Figure~\ref{fig:tree-repeat} 
\note{(where corresponding points in $\cal K$ and $G$ are referred to by the same name).}
%we use, e.g.~$p_1$ to refer to a point in $\cal K$ and $G$).}  
Compare the pair $p_1, c$, where \revised{$| \sigma_G(p_1, c) | = 155^\circ$} and $\sigma(p_1, c)$ does not go through the origin, with the pair $p_1, b$, where \revised{$| \sigma_G(p_1, b) | = 205^\circ$} and $\sigma(p_1, b)$ goes through the origin.
\begin{proof}
If $\sigma(a,b)$ does not go through $O$, then $\sigma(a,b)$ travels through some cells, and, by 
\revised{unfolding these, i.e., placing them} one after another in the plane, $\sigma(a,b)$ forms a straight line segment through the cells, which creates a triangle together with point $O$.  The angle of this triangle at $O$ is $| \sigma_G(\lambda(a),\lambda(b)) |$ which is therefore less than $\pi$.

Conversely, if $| \sigma_G(\lambda(a),\lambda(b)) | < \pi$, then the path $\sigma_G(\lambda(a),\lambda(b))$ follows segments of the link graph which correspond to cells of $K$, and when we place these cells one after another in the plane, the angle between segments $Oa$ and $Ob$ in the plane is $| \sigma_G(\lambda(a),\lambda(b)) |$.  Thus the straight line segment from $a$ to $b$ remains in the cells, and forms a geodesic path from $a$ to $b$ that does not go through $O$, and that maps to $\sigma_G(\lambda(a),\lambda(b))$.      
\end{proof}
}

\subsection{Counterexamples for Convex Hulls in CAT(0) complexes}
\label{sec:CH-examples}

In this section we give examples to show that the following properties of the convex hull of a set of points $P$ in Euclidean space do not carry over to CAT(0) complexes,
\note{not even single-vertex CAT(0) cube complexes}.

\begin{enumerate}
\item  Any point on the boundary of the convex hull  of points in 2D is on a shortest path between two points of $P$.
\item  In any dimensional space, the convex hull of three points is 2-dimensional.
\item  Any point inside the convex hull can be written as a convex combination of points of $P$.
%\item \note{Let's delete this one, since this property fails even for the simplest first example. }
%Consider a point $q$ and take a small ball $B$ around it.   For each $p \in P$, let $p'$ be the point where the shortest path from $q$ to $p$ exits $B$.   Then $q$ is in CH$(P)$ iff $q$ is in the convex hull of the points $p'$.
\end{enumerate}

\begin{figure}[htb]
\centering
\includegraphics[width=2in]{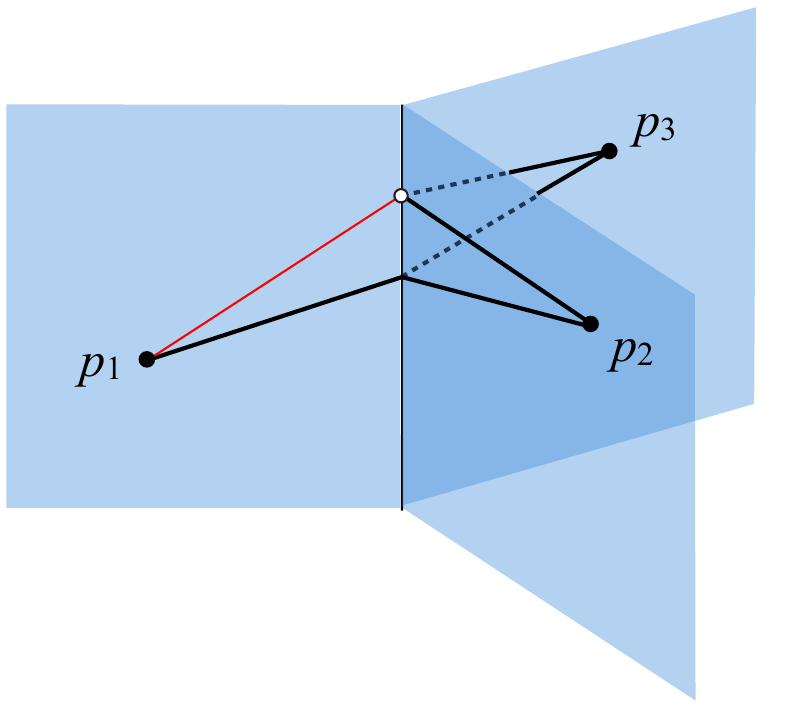}
\caption{The shortest paths between $p_1, p_2, p_3$ (shown in black) do not determine the convex hull because the thin red line is on the boundary of the convex hull but not on any of the shortest paths.}
\label{fig:3pagebook}
\end{figure}

Our first example, shown in Figure~\ref{fig:3pagebook}, has three cells sharing an edge.  Set $P$ contains one point in each cell.  The three shortest paths between pairs of points in $P$ do not determine the convex hull.  This shows that property 1 fails.

\begin{figure}[ht]
\centering
\includegraphics[width=5in]{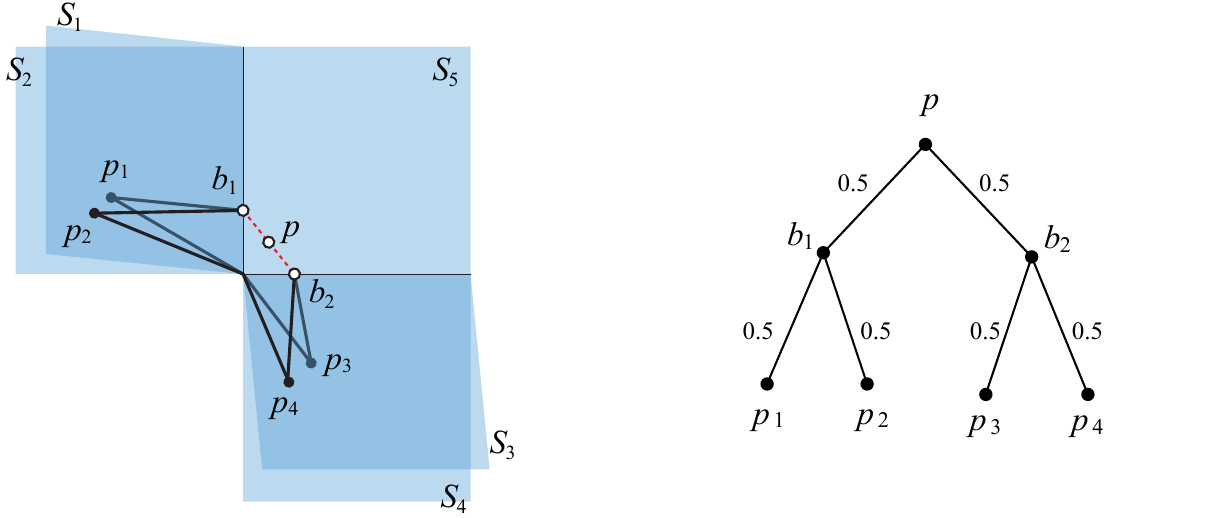}
\caption{(left) A 2D CAT(0) space consisting of 5 quadrants.  Quadrants $S_1, S_2, S_5$ share a vertical edge, and quadrants $S_3, S_4, S_5$ share a horizontal edge.  Point $p_i, i=1,2,3,4$, lies in quadrant $S_i$.   The convex hull of $p_1, p_2, p_3, p_4$ contains points---in particular $p$ and the dashed red line---in a quadrant, $S_5$, that is not entered by any shortest path between points in $P$ (shown as black lines).  
Note that 4 of the 6 shortest path between points of $P$ go through the origin.
(right) The tree representation of point $p$ in terms of $P$.}
\label{fig:remote-quadrant}
\end{figure}

Furthermore, the example in Figure~\ref{fig:remote-quadrant} shows that even in a single-vertex 2D CAT(0) rectangular complex, the convex hull of a set of points $P$ can contain a point in a quadrant that is not entered by any shortest path between points of $P$.
\changed{This example also shows that Carath\'eodory's property may fail, since point $p$ is in the convex hull of the four points $p_1, p_2, p_3, p_4$ but not in the convex hull of any three of the points.}

The example in Figure~\ref{fig:3points3D} shows that the convex hull of three points in a 3D CAT(0) complex may contain a 3D ball\revised{, and thus property 2 fails.  Lin et al. \cite{LinSturmfelsTangYoshida} give a more complicated family of examples in which the top-dimensional cells have dimension $2d$, and there exist 3 points in the space such that their convex hull contains a $d$-dimensional simplex.}

\begin{figure}[htb]
\centering
\includegraphics[width=3.5in]{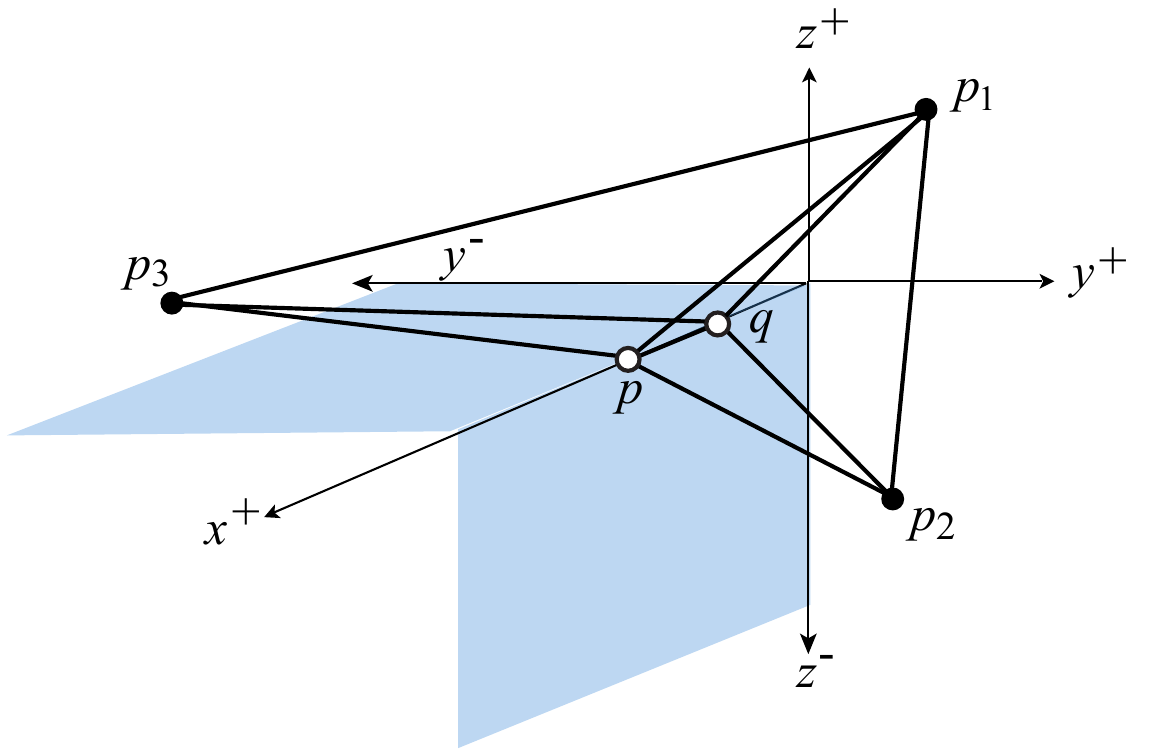}
\caption{\revised{The convex hull of the three points $p_1, p_2, p_3$ in this 3D CAT(0) complex contains a 3D ball.  The complex consists of exactly 3 octants, $x^+ y^+ z^+$, $x^+ y^+ z^-$, and $x^+ y^- z^+$, and is missing the bottom left octant $x^+ y^- z^-$.  Point $p_1$ lies in the $x^+ y^+ z^+$ octant, $p_2$ in the  $x^+ y^+ z^-$ octant, and $p_3$ in the $x^+ y^- z^+$ octant.}  Point $p$ is where the shortest path from $p_2$ to $p_3$ intersects the $x^+$ axis, and point $q$ is where the triangle $p_1 p_2 p_3$ is pierced by the $x^+$ axis.  The convex hull of $P$ consists of the union of two simplices $p_1 p_3 p q$ and $p_1 p_2 p q$.}
%The convex hull of the three points $p_1, p_2, p_3$ in this 3D CAT(0) complex contains a 3D ball.  Point $p_1$ lies in the $x^+ y^+ z^+$ octant, $p_2$ in the  $x^+ y^+ z^-$ octant, and $p_3$ in the $x^+ y^+ z^-$ octant.  The $x^+ y^- z^-$ octant is missing.  Point $p$ is where the shortest path from $p_2$ to $p_3$ intersects the $x^+$ axis, and point $q$ is where the triangle $p_1 p_2 p_3$ is pierced by the $x^+$ axis.  The convex hull of $P$ consists of the union of two simplices $p_1 p_3 p q$ and $p_1 p_2 p q$.}
\label{fig:3points3D}
\end{figure}

\begin{figure}[!ht]
\centering
\includegraphics[width=\linewidth]{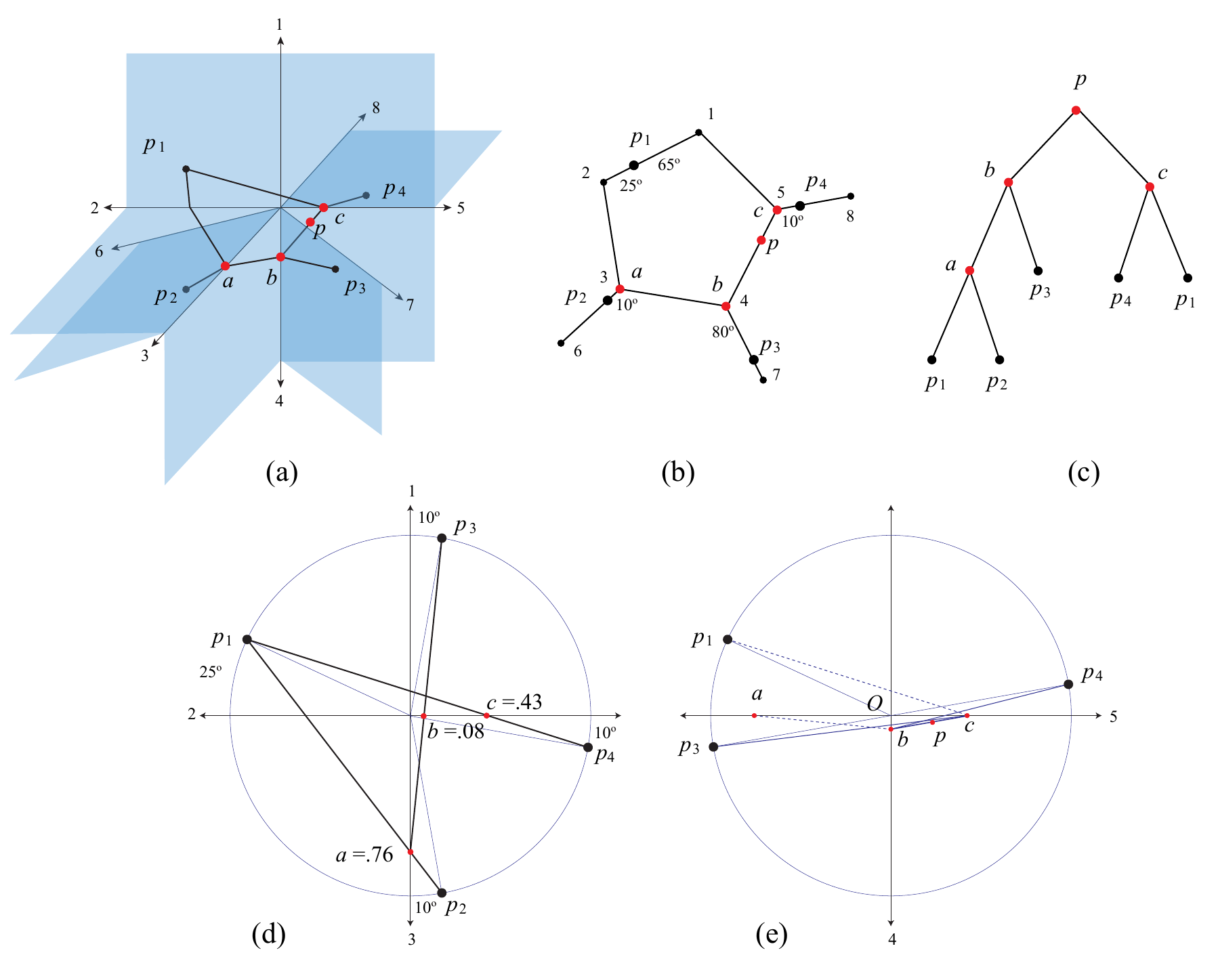}
\caption{An example of a 2D CAT(0) complex where point $p$ in CH$(\{p_1, p_2, p_3,p_4\})$ cannot be represented as a binary tree with distinct leaves.  Points $p_1, p_2, p_3,p_4$ are all distance 1 from the origin, with angles as specified in the link graph.
(a) the complex $\cal K$ (not to scale); (b) the link graph $G$; 
\revised{(c) a binary tree representing $p$ with leaf $p_1$ repeated; (d) construction of points $a, c$ and $b$; (e) $p$ does not lie in either the shortest path from $b$ to a point in $Op_4$ or the shortest path from $c$ to a point in $Op_3$.
} 
%\remove{
%Point $p$ can only be formed by a combination using $p_1$ twice.
%The angle between the points $p_3$ and $p_4$ is \note{$185^\circ > \pi$} radians, so the geodesic between them passes through the origin, and not through $p$. The angle between $p_1$ and $p_4$ is $175^\circ \le \pi$ radians, so the geodesic between them passes though the interior of edge $c$, giving a geodesic between this point at $p_3$ that passes through the orthant containing $p$. \note{However, we now need to explain that $p$ is further out from the origin than this geodesic, and hence we also need to know point $c$ is in the convex hull via $p_1$ again. This might be more subtle than shown in the picture.}
%\changed{I altered the angles to something that should work.  The points $p_i$ should all be the same distance from $O$.  From the angles, the only pairs $p_i$ whose geodesic does not go through $O$ are $p_1 p_2$ and $p_1 p_4$.  Because $p_1$ is no longer centered on its edge in the link graph, the point $a$ will be further from the origin than point $c$.  Then point $b$---as computed on the $p_3 a$ path---will be further from $O$ than the point obtained on axis 4 from the $p_3 c$ path.  We need to write these details and make an extra figure, I think.}}
}
\label{fig:tree-repeat}
\end{figure}

Property 3 must be expressed more carefully for CAT(0) complexes because it is not clear what a convex combination of a set of points means except when the set has two points.   If $p$ and $q$ are two points in a CAT(0) complex, then the points along the shortest path from $p$ to $q$ can be parameterized as
$(1-t)p + tq$ for $t \in [0,1]$.   Based on the definition of the convex hull,   any point in the convex hull of a set of points $P$ can be represented  as a rooted binary tree with leaves labelled by points in $P$ (with repetition allowed) and with the two child edges of each internal node $v$ labelled by two numbers $(1-t_v)$ and $t_v$ for $t_v \in [0,1]$, meaning that the point associated with $v$ is
this combination of the points represented by the child nodes.  For example, see Figure~\ref{fig:remote-quadrant}.

One might hope that every point in the convex hull can be represented by a binary tree whose leaves are labelled by distinct elements of $P$.
If this were true then we could verify that a point is in the convex hull \note{of $m$ points} by giving a binary tree with at most $m$ leaves, and the problem of deciding membership in the convex hull would lie in NP, \note{at least for the case of cube complexes, where the 
sizes of the weights attached to the binary tree are polynomially bounded}.
%If this were true then the problem of deciding whether a point is in the convex hull would lie in the class NP (at least modulo bit complexity issues).  
However, this hope is dashed by the example in Figure~\ref{fig:3points3D}.  
%, but  this is false for the example in Figure~\ref{fig:3points3D}. 
Furthermore, the property may even fail for a 2D CAT(0) complex as we prove below for the example in Figure~\ref{fig:tree-repeat}:

\begin{lemma}
For the 2D CAT(0) complex shown in Figure~\ref{fig:tree-repeat}, the point $p$ cannot be represented by a binary tree with distinct leaves from $\{p_1, p_2, p_3,p_4\}$.
\label{lemma:tree-repeat}
\end{lemma}
\begin{proof}

\revised{Point $p$ can be generated by a binary tree with two leaves labelled $p_1$, as shown in 
Figure~\ref{fig:tree-repeat}(c).  This tree has internal nodes corresponding to points $a, b$ and $c$.  The gist of our argument is to show that $p$ cannot be generated without points $b$ and $c$, and each of those requires $p_1$ to generate it.  

The first part of the argument involves the link graph $G$ shown in Figure~\ref{fig:tree-repeat}(b),
and the second part involves the actual coordinates (the distance from $O$) of the intermediate points.
As shown in Figure~\ref{fig:tree-repeat}(d), starting from quadrant $\{1,2\}$ in the upper left, we can compute $a = .76$ and $c=.42$, and from these (unfolding further quadrants on top of each other) $b = .08$.
Figure~\ref{fig:tree-repeat}(e) 
shows quadrant $\{4,5\}$ containing $b,c,p$ in the lower right  together with neighbouring quadrants.  The property we observe from the figure (and could calculate numerically) is that $p$ lies below segment $bp_4$ and below segment $cp_3$.
With these facts in hand, we now give the details of the proof.

Suppose $p$ is represented by a binary tree $T$ with distinct leaves from $\{p_1, p_2, p_3,p_4\}$.
Let $q_1$ and $q_2$ be the points of the complex corresponding to the children of the root of $T$.
Thus $p$ lies on the shortest path from $q_1$ to $q_2$.
Point $p$ corresponds to a point $\lambda(p)$  that lies in edge $(4,5)$ of the link graph $G$.
By Proposition~\ref{property:path-thru-O}, the shortest path between $\lambda(q_1)$ and $\lambda(q_2)$ must have length less than $180^\circ$, so one of them, say $\lambda(q_1)$, must lie  in edge $(5,1)$ or $(5,8)$ of $G$, and the other, say $\lambda(q_2)$, must lie in edge $(4,3)$ or $(4,7)$.
Furthermore, without loss of generality, the subtree rooted at $q_1$ must include the leaf $p_4$ and the subtree rooted at $q_2$ must include the leaf $p_3$, because without those points we cannot generate anything in the appropriate edges of $G$.

Point $p_1$ may be a leaf of the subtree rooted at $q_1$ (case 1) or the subtree rooted at $q_2$ (case 2), but not both.

In case 1, the subtree rooted at $q_2$ has leaf $p_3$ and possibly $p_2$.  
In $G$, the shortest path between $p_2$ and $p_3$ has length $180^\circ$
so, by Proposition~\ref{property:path-thru-O}, the geodesic between $p_2$ and $p_3$ in the complex 
is a ``cone path'' that
goes through the origin.   Thus $q_2$ must be a point in the segment $Op_3$. 
The subtree rooted at $q_1$ has leaves $p_4, p_1$ and possibly $p_2$. 
As shown in Figure~\ref{fig:tree-repeat}(e), the extreme points we can generate are $p_4$, $c$, and $O$. 
However, for any point $q_1$ in this triangle, and any point $q_2$ in $Op_3$, the shortest path between $q_1$ and $q_2$ does not go through $p$.
This rules out case 1.
 
In case 2, the subtree rooted at $q_1$ has leaf $p_4$ and possibly $p_2$.  Since the path between these points is a cone path, $q_1$ must be a point in the segment $Op_4$.
The subtree rooted at $q_2$ has leaves $p_3, p_1$, and possibly $p_2$. 
As shown in Figure~\ref{fig:tree-repeat}(e), the extreme points we can generate are $p_3$, $b$, and $O$. 
However, for any point $q_2$ in this triangle, and any point $q_1$ in $Op_4$, the shortest path between $q_1$ and $q_2$ does not go through $p$.
This rules out case 2.
 
Therefore, $p$ cannot be represented by a binary tree with distinct leaves from $\{p_1, p_2, p_3,p_4\}$.
}
\remove{
In the link graph (shown in Figure~\ref{fig:tree-repeat}(b)), the 
shortest path between $p_3$ and $p_4$ has length $180^\circ$ 
%angle between $p_3$ and $p_4$ is $180^\circ$ 
so, by Proposition~\ref{property:path-thru-O}, the geodesic between $p_3$ and $p_4$ in the complex 
\revised{is a ``cone path'' that}
goes through the origin.  Similarly, the geodesics between the following pairs 
\revised{are also cone paths}:
%also go through the origin: 
$p_2$ and $p_3$; $p_1$ and $p_3$; $p_2$ and $p_4$.  The only pairs of points whose geodesic does not go through $O$ are
%with angles less than $180^\circ$ are 
$p_1$, $p_2$, and $p_1$, $p_4$.  
%\revised{A binary tree with distinct leaves cannot use both these pairs.}  

\revised{A binary tree with distinct leaves chosen from $\{p_1, p_2, p_3,p_4\}$ has at most
 4 leaves and 3 internal nodes.  
 \rerevised{We first claim that no internal node corresponds to a point on a cone path between its children.  
 
Thus, in particular, any internal node whose children are both leaves must have children  $p_1, p_2$, or $p_1,p_4$.
This implies that the root's children cannot both be internal nodes because their children would then have to be $p_1, p_2$ and $p_1, p_4$, respectively, which means repeating $p_1$. 
} 
Thus, in a binary tree that represents $p$ and has distinct leaves, one child of the root must be a leaf labelled $\ell$ where $\ell$ is one of $p_2$, $p_3$, or $p_4$. 
Furthermore, $p$ must lie on a shortest path between $\ell$ and a point $h$ in the convex hull of the other three points.
From the link graph, we see that one of $\ell$ and $h$ must lie in edge $(5,1)$ or $(5,8)$, and the other must lie in the edge
$(4,3)$ or $(4,7)$.  This rules out $\ell=p_2$ since $p_2$ lies in none of these edges.

It remains to consider $\ell=p_3$ or $p_4$.
If $\ell=p_3$, then $h$ must lie in $(5,1)$ or $(5,8)$, and the best such point in the convex hull of $p_1, p_2, p_4$ is $h=c$.
If $\ell = p_4$,  then $h$ must lie in $(4,3)$ or $(4,7)$, and the best such point in the convex hull of $p_1, p_2, p_3$ is $h=b$.

As shown in Figure~\ref{fig:tree-repeat}(d), we can compute the coordinates (the distance from $O$) of points $a = .76$ and $c=.42$, and from these $b = .08$.
We can then calculate\footnote{calculations done in Geogebra}, as  shown in 
Figure~\ref{fig:tree-repeat}(e), that $p$ does not lie in either the shortest path from $p_4$ to $b$ or the shortest path from $p_3$ to $c$. 
Therefore, $p$ cannot be represented by a binary tree with distinct leaves from $\{p_1, p_2, p_3,p_4\}$.
 } 
The geodesic from $p_1$ to $p_2$ crosses axis 3 at point $a$.  The geodesic from $p_1$ to $p_4$ crosses axis 5 at point $c$.  Because $p_1$ is closer to axis 2, point $c$ is closer to the origin than point $a$ (see Figure~\ref{fig:tree-repeat}(d)).  The geodesic paths from $p_3$ to both $a$ and $c$ cross axis 4 and thus can be used to construct points in the cell bounded by axes 4 and 5 (where $p$ lies).  However, as shown in Figure~\ref{fig:tree-repeat}(e), the geodesic path from $p_3$ to $a$ crosses axis 4 further from the origin at point $b$, and therefore point $p$ can only be constructed using points $a$ and $b$, as shown by the binary tree in \revised{Figure~\ref{fig:tree-repeat}(c)}.  Therefore $p$ can only be represented by a binary tree that repeats the leaf $p_1$.
}
\end{proof}

%\remove{
%the example in  Figure~\ref{fig:3points3D} falsifies this.
%In this example $P$ consists of three points $p_1, p_2, p_3$.  There are three possible shapes of binary trees with three leaves, depending on which of $p_1, p_2, p_3$ is a child of the root.
%The set of points expressed by each shape is a 2-dimensional set.   For example, when $p_1$ is a child of the root, we have the set of points $(1-t_1)p_1 + t_1((1-t_2)p_2 + t_2p_3)$, for $t_1, t_2 \in [0,1]$.
%The union of these three 2-dimensional sets is 2-dimensional,  but the convex hull contains a 3-dimensional ball.  We conjecture that 2D CAT(0) complexes do not permit such counter-examples.

%\begin{conjecture}
%Every point in the convex hull of set $P$ in a 2D CAT(0) complex can be expressed as a binary tree whose leaves are labelled by distinct elements of $P$.
%\end{conjecture}
%}

%%%%%%%%%%%%%%%%%%%%%%%%%%%%%%%%%%%%%%%%%%%%%%%%%%%%
\subsection{Convex Hull Algorithm for a Single-Vertex 2D CAT(0) Complex}
\label{sec:CH-alg}
%\subsection{Testing Whether a Point is in the Convex Hull}
%\label{sec:CH-origin-test}

In this section we prove Theorem~\ref{thm:convex-hull-alg} by reducing the convex hull problem for a single-vertex 2D CAT(0) complex $\cal K$ to linear programming via a polynomial-time reduction in the real RAM model of computation.
Recall that $P$ is the finite set of points whose convex hull we wish to find, and $O$ is the single vertex of the complex.

We will find the convex hull as a union of convex polygons, one for each cell of $\cal K$.
To justify this, observe that the intersection of $\CH(P)$ with cell $C$ of $\cal K$ is a convex polygon (which may be open or closed on parts of its boundary), and $\CH(P)$ is the union of these polygons.

\remove{
Consider one cell $C$ for which $D(C)$ is non-empty.  Let $P(C)$ be the points of $P$ that lie in $C$. 
As shown by the example in Figure~\ref{fig:3pagebook}, 
 in general $D(C)$ is not just the convex hull of $P(C)$. 
Cell $C$ is bounded by two rays $e$ and $f$ incident to $O$.    $D(C)$ intersects each ray in an interval, say the interval $(x_e^{\rm min}, x_e^{\rm min})$ on $e$ and the interval $(x_f^{\rm min}, x_f^{\rm min})$ on $f$ where in both cases ``min'' indicates the point closer to $O$.
When $O$ is in $\CH(P)$ then $x_e^{\rm min} = x_f^{\rm min} = O$.
}

The recursive definition of the convex hull of $P$ involves taking \emph{all} pairs of points $p,q$ in the set and adding \emph{all} points along the unique geodesic from $p$ to $q$.  
%We begin by showing that it suffices to consider
We consider using 
a restricted set of points, namely, those that lie in $P$ and  on the edges of $\cal K$.

Define $S_0$ to be the set $P$ together with point $O$ if it is in the convex hull of $P$.  
For $i=1, 2, \ldots$ we recursively define a finite  set of points $S_i$, called the $i^{\rm th}$ \emph{skeleton}, as follows.  Initialize a set $T_i$ to be $S_{i-1}$.  For each pair of points $p,q$ in $S_{i-1}$, take the shortest path, $\sigma$, in $\cal K$ from $p$ to $q$.  Add to $T_i$ all the intersection points of $\sigma$ with edges of the complex.  Observe that $T_i$ is a finite set.  For any edge $e$ of the complex, if $T_i$ contains more than 2 points of $e$, then discard all but the two extreme points,
$S_i^{\max}(e)$ and $S_i^{\min}(e)$.  In case $O$ is in the convex hull, then $S_i^{\min}(e) = O$.
Now define $S_i$ to be $T_i$.

We can augment each skeleton $S_i$ to a larger subset of $\CH(P)$ as follows.  
For each cell $C$ of the complex, let $S_i(C)$ be the points of $S_i$ that lie in the closed cell $C$.  Thus $S_i(C)$ consists of: the points of $P$ that lie in $C$; point $O$ if it lies in $\CH(P)$; and between 0 and 4 points that lie on the two boundary rays of $C$.   Define $H_i(C)$ to be the Euclidean convex hull of $S_i(C)$ in $C$ and define $H_i = \bigcup \{  H_i(C) : C$ a cell of ${\cal K} \}$.  Observe that $H_i \subseteq \CH(P)$.

We justify the restriction to skeletons by showing that each $H_i$ contains all shortest paths between points of $H_{i-1}$:

\begin{theorem}
Let $p$ and $q$ be points of $H_{i-1}$ and let $\sigma$ be the shortest path in $\cal K$ from $p$ to $q$.
Then $\sigma \subseteq H_i$. 
\label{thm:CH-skeleton}
\end{theorem}
\begin{proof}  Suppose that $p \in H_{i-1}(C_p)$ and $q \in  H_{i-1}(C_q)$ for some cells $C_p$ and $C_q$.  
%\rednote{Anna: Actually, I think the ``wrong'' cell choice works ok. The one case that needs explicit care is if $p$ and $q$ lie in the same edge.  Then the line $L$ in the proof below goes through $O$.  }  
\revised{If $p$ or $q$ lies on a cell boundary, and thus could be assigned to more than one cell, choose the cell assignments to minimize the number of cells traversed by the geodesic between $p$ and $q$.  
In particular, if $p$ and $q$ lie in the same edge, or if one of them is at $O$,
%$p$ or $q$ is $O$, 
then assign both points to the same cell.}
If $C_p = C_q$ then $\sigma \subseteq H_{i-1}(C_p)$ since $H_{i-1}(C_p)$ is a convex set.  
In this case we are done because $H_{i-1}(C_p) \subseteq H_{i}(C_p)$.

If $\sigma$ goes through $O$, then $\sigma$ consists of two subpaths from $p$ to $O$ and from $O$ to $q$ and these subpaths lie in $H_{i-1}(C_p)$ and $H_{i-1}(C_q)$ respectively, so we are also done in this case. 

Otherwise, suppose that $\sigma$ crosses the sequence of cells $C_p = C_0, C_1, \ldots, C_t = C_q$.
We can unfold these cells in the plane so that $\sigma$ becomes a straight line.  
%\revised{This line does not go through $O$ since $p$ and $q$ are in different cells.}
See Figure~\ref{fig:CH-path}.
For $j=1, \ldots, t$
let $e_j$ be the edge (or ray) of $\cal K$ between $C_{j-1}$ and $C_j$,  and let $p_j$ be the point where $\sigma$ crosses $e_j$.  Let $p_0 = p$ and $p_{t+1}=q$.  It suffices to show that $p_j \in H_i$ for all $j$, since this implies that the subpath of $\sigma$ from $p_{j-1}$ to $p_j$ lies in $H_i(C_{j-1})$.

\begin{figure}[ht]
\centering
\includegraphics[width=.5\linewidth]{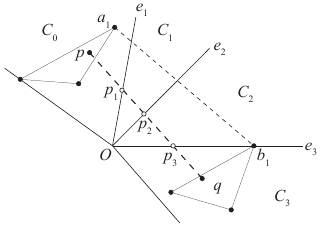}
\caption{Illustration for the proof of Theorem~\ref{thm:CH-skeleton}.
}
\label{fig:CH-path}
\end{figure}

%\revised{If both $p$ and $q$ are in $S_{i-1}$, then by definition, $S_i$ contains $p_j$ for all $j$, and thus so does $H_i$, finishing this case.  Otherwise, without loss of generality, assume $p \notin S_{i-1}$}.  Because $p$ lies in \revised{$H_{i-1}(C_p) - S_{i-1}$} there must be points $a_1, a_2, a_3$\revised{, not necessarily distinct,} in $S_{i-1}(C_p)$ with $p$ inside the \revised{(possibly degenerate)} triangle $a_1 a_2 a_3$. 
%Similarly, t
%\revised{T}here \revised{also} must be points $b_1, b_2, b_3$\revised{, not necessarily distinct,} in $S_{i-1}(C_q)$ with $q$ inside \revised{or on the boundary of} the \revised{(possibly degenerate)} triangle $b_1 b_2 b_3$.  

Because $p$ lies in \revised{$H_{i-1}(C_p)$, which is the Euclidean convex hull of $S_{i-1}(C_p)$,} 
there must be points $a_1, a_2, a_3$ in $S_{i-1}(C_p)$ with $p$ inside the triangle $a_1 a_2 a_3$. 
\revised{We note that the triangle may degenerate to a line segment (or even to a point, in case $p \in  S_{i-1}(C_p)$).} 
Similarly, there must be points $b_1, b_2, b_3$ in $S_{i-1}(C_q)$ with $q$ inside the \revised{(possibly degenerate)} triangle $b_1 b_2 b_3$. 

In the planar unfolding of cells $C_0, C_1, \ldots, C_t $, extend  $\sigma$ to a straight line $L$.  
\revised{Note that $L$ does not go through $O$ otherwise $p$ and $q$ would lie in the same cell.} Let the ``near side'' of $L$ be the (closed) side containing $O$, and the ``far side'' be the other (closed) side.  At least one of the $a_k$'s, say $a_1$, must be on the far side of $L$.  Similarly, at
%\revised{A}t 
least one of the $b_k$'s, say $b_1$, must be %\revised{$q$ or} 
on the far side of $L$.   Then the shortest path from $a_1$ to $b_1$ in $\cal K$ unfolds to the straight line segment $a_1 b_1$.  
\revised{(To justify this, note that the angle $a_1Ob_1$ is less than $\pi$ because $a_1$ and $b_1$ are on the far side of $L$.)}
The line segment  $a_1 b_1$ crosses each edge $e_j$ on the far side of $L$.  By the definition of $S_{i}$ each such crossing point, or a point even farther along $e_j$, becomes a point of $S_i$.

We can make the same argument about the near side of $L$.  Suppose points $a_2$ 
%\revised{is on the near side of $L$ and $b_2$ is either $q$ or on the near side of $L$.}
and $b_2$ are on the near side of $L$.  
The shortest path from $a_2$ to $b_2$ either goes through $O$, or becomes a straight line segment in the unfolding.  In either case, every edge $e_j$ contains a point of $S_{i}$ that is on the near side of $L$.  

Because point $p_j$ is between two points of $S_{i}$ on $e_j$, thus $p_j \in H_i$.
\end{proof}

%Along each edge $e$ of the complex, the convex hull of $P$ is bounded, and the one or two points of $S_i$ on $e$ define an interval that is monotonically increasing in $i$.  Thus the sets $S_i$ have limit points on each edge $e$.  Define $S$ to be the limit of the $S_i$'s, define $S(C)$ to be the points of $S$ in cell $C$, and define $H = \bigcup \{  H(C) : C$ a cell of ${\cal K} \}$.  Then we have:
%
%\begin{theorem}  $\interior(H) \subseteq \CH(P) \subseteq H$.
%\label{thm:CH-characterization}
%\end{theorem}

\begin{corollary}
$\CH(P) = \bigcup_i H_i$.
\label{cor:CH-characterization}
\end{corollary}

Alternatively, we can express $\CH(P)$ in terms of a set $S$ that is the limit of the $S_i$'s.  For each edge $e$ of $\cal K$, 
the sequence of points $S_i^{\min}(e), i=0, 1, \ldots$ is decreasing and bounded below by $O$.  The set  $S_i^{\max}(e), i=0, 1, \ldots$ is increasing and bounded above.  (Note that no point of $\CH(P)$ will be further from $O$ than the furthest point of $P$.)  Thus the limit points $S^{\min}(e)$ and $S^{\max}(e)$ exist.  Define $S$ to be the \revised{union of $P$ and the} set of limit points on all edges $e$.

Similar to the definition of $H_i$ from $S_i$, we define $H(C)$ to be the \revised{Euclidean} convex hull of $S \cap C$ for each cell $C$, and then define $H$ to be  $\bigcup \{  H(C) : C$ a cell of ${\cal K} \}$. 
Certainly $\CH(P) \subseteq H$.  Whether they are equal is the same as the 
question of whether  $\CH(P)$ is closed, \revised{as the following proposition shows.}

\revised{
\begin{proposition} $H$ is the closure of $\CH(P)$. 
\label{prop:closure}
\end{proposition}
\begin{proof} 
$H$ is a closed set containing $\CH(P)$, so $H$ contains the closure of $\CH(P)$.  In the other direction, the closure of $\CH(P)$ contains $S$ and therefore contains $H$.
\end{proof}
}

Our approach to proving Theorem~\ref{thm:convex-hull-alg} is to \revised{compute $H$, the closure of $\CH(P)$, by capturing}
the \revised{limit skeleton $S$} via linear programming.  
%\revised{\sout{This will show that $\CH(P)$ is closed.}}

\revised{
A more obvious approach to computing $\CH(P)$ would be to compute the sequence of  $S_i$'s.
Such a procedure is finite if and only if $\CH(P)$ is closed:}

\revised{
\begin{proposition}  $\CH(P)$ is closed if and only if 
the sequence $S_0, S_1, \ldots$ is finite (i.e., $S_k = S_{k+1}$ for some $k$).
\end{proposition}
} %end revised
\begin{proof}
If $\CH(P)$ is closed then for each edge $e$ of $\cal K$, the extreme points of $\CH(P)$ on $e$ must enter $S_i$ for some $i$.  The set of such extreme points is finite (there are at most two per edge), so all of them are contained in $S_k$ for some $k$.  Then $S_k = S_{k+1}$. 

In the other direction, if $S_k = S_{k+1}$ then $H_k = H_{k+1}$.  Also 
\revised{$S= S_{k}$ and $H = H_{k}$.  
By Corollary~\ref{cor:CH-characterization}, $\CH(P) = \bigcup_i H_i$.  Then 
 $\bigcup_i H_i = \bigcup_i^k H_i = H_k = H$, so $\CH(P) = H$.
By Proposition~\ref{prop:closure}, 
$H$ is the closure of $\CH(P)$ so this implies that $\CH(P)$ is closed.
}
 \end{proof}

%The question of whether $\CH(P)$ is closed is the question of whether the $S_i$'s (and the $H_i$'s) form a finite sequence:

\remove{
\begin{lemma}  
\revised{The following conditions} are equivalent:
\begin{enumerate}
\item $\CH(P)$ is closed,
\item the sequence $S_0, S_1, \ldots$ is finite (i.e., $S_k = S_{k+1}$ for some $k$),
\item $\CH(P)= H$.
\end{enumerate}
%The sequence $S_0, S_1, \ldots$ is finite (i.e., $S_k = S_{k+1}$ for some $k$) if and only if $\CH(P)$ is closed. 
\label{lemma:CH-closed} 
\end{lemma}
\begin{proof}
(1 $\rightarrow$ 2) If $\CH(P)$ is closed then for each edge $e$ of $\cal K$, the extreme points of $\CH(P)$ on $e$ must enter $S_i$ for some $i$.  The set of such extreme points is finite (there are at most two per edge), so all of them are contained in $S_k$ for some $k$.  Then $S_k = S_{k+1}$. 

(2 $\rightarrow$ 3) 
If $S_k = S_{k+1}$ then $H_k = H_{k+1}$.  Also \revised{$S= S_{k}$} and \revised{$H = H_{k}$.  
By Corollary~\ref{cor:CH-characterization}, $\CH(P) = \bigcup_i H_i$.  Then 
 $\bigcup_i H_i = \bigcup_i^k H_i = H_k = H$, so $\CH(P) = H$.}
%By Theorem~\ref{thm:CH-skeleton} $H_k$ is closed under taking shortest paths.  Thus $\CH(P) = H$. 

(3 $\rightarrow$ 1) $H$ \revised{is defined as the finite union of $H(C)$, and each $H(C)$ is a Euclidean convex hull and therefore closed.  Thus} $H$ is a closed set.
\end{proof}
}

%Our proof of Theorem~\ref{thm:convex-hull-alg} will show that the convex hull is closed and thus, by Lemma~\ref{lemma:CH-closed}, the sequence of $S_i$'s is finite.
%Our approach is to use linear programming to find the limiting skeleton  $S$.

\revised{
We conjecture that $\CH(P)$ is closed, and thus that we can compute  $\CH(P)$ by computing each $S_i$ until no further changes occur.  This}
%The more obvious approach of successively computing each $S_i$ 
would be efficient if there were a good bound on the length of the sequence.   We conjecture that there is such a bound: 

\begin{conjecture}
$S_k = S_{k+1}$ for some $k$ that is polynomially bounded in $n$ and $m$.
\label{conj:iteration}
\end{conjecture}

\remove{
Suppose first that the origin $O$ is inside the convex hull.  The case where the origin is not inside the convex hull will be dealt with later.
To find the convex hull
it suffices to describe the intersection of the convex hull with each cell of the complex.
Consider a cell $C$ bounded by rays $e$ and $f$ incident to $O$.  The part of the convex hull inside $C$ is a convex polygon determined by its vertices which consist of: point $O$; \changed{possibly some or all of} the points of $P$ inside $C$; and
two points $x_e$ and $x_f$ on edges $e$ and $f$, respectively, that are on the boundary of the  convex hull.
Note that we do not make any assumption about whether $x_e$ and $x_f$ are in the convex hull, because we are not making any assumption about whether the convex hull is open or closed.
If we knew the points $x_e$ and $x_f$, then we could easily compute the part of the convex hull inside $C$, since it is simply the Euclidean convex hull of $x_e$, $x_f$, $O$, and the points of $P$ inside $C$.

There is an obvious iterative approach to finding the points $x_\ell$ where the boundary of the convex hull intersects each edge $\ell$ of the complex:
%\changed{To find the values of $x_\ell$ in the convex hull, t}here is an obvious iterative approach: 
Initialize $H_0$ to be the set $P$.  For $i=1, 2, \ldots$, initialize $H_i$ to $H_{i-1}$ and then take every pair of  points from $H_{i-1}$, compute the shortest path $\sigma$ between them, and add to $H_i$ all the intersection points of $\sigma$ with edges of the complex. 
(When two points lie on the same edge, we can discard the one closer to $O$.)
\rednote{We keep only the extreme points along each edge, where the min is $O$ if $O$ lies in the CH.}

We observe that this process is finite if and only if the convex hull is closed.
As mentioned, our proof of Theorem~\ref{thm:convex-hull-alg} shows that the convex hull is closed and thus this iterative process is finite.
In order for the iterative process to be an efficient algorithm we would need a bound on the number of iterations.
We conjecture that there is a polynomial bound:
%, but we do not even have a finite bound \note{(actually I think we do)}, although we make the following conjecture: 

\begin{conjecture}
$H_k = H_{k+1}$ for some polynomially bounded $k$.
\label{conj:iteration}
\end{conjecture}
} % end remove

\subsubsection{Combinatorics of the Convex Hull}

For our linear programming approach we need to know whether $O$ is in the convex hull, and we need to identify the edges of $\cal K$ that contain points of the convex hull other than $O$.
We give algorithms for these  using the link graph $G = G_O$.
Let $V$ be the set of vertices of $G$; recall that these correspond to the edges of $\cal K$.
%
%In this subsection we give algorithms to decide
%if $O$ is in the convex hull and 
%to identify the edges of $\cal K$ that contain points of the convex hull other than $O$.  
%
%We will do this using the link graph $G = G_O$.   
%Let $V$ be the set of vertices of $G$; recall that these correspond to the edges of $\cal K$.
%Our algorithms take polynomial time in any model of computation where we can find shortest paths in the link graph and compare their lengths to $\pi$. 

We begin by showing that we can compute the projection of the convex hull of $P$ on the link graph.  We introduce some notation to make this formal.  For any point \revised{$p \in {\cal K} - \{O\}$}, denote the corresponding point in the link graph by $\lambda(p)$.  We extend this notation to subsets of $\cal K$\revised{---for a set $S \subseteq {\cal K}$,
define $\lambda(S)$  as $\{\lambda(p) : p \in S - \{O\} \}$.
In particular,}
$\lambda(\CH(P))$ denotes the projection of the convex hull of $P$ on the link graph.   

We introduce the \emph{link convex hull},   
$\LCH(P)$, a subset of $G$ defined  recursively as follows: (1)  $\lambda(P)$ is contained in $\LCH(P)$; (2)  For any two points $a,b$ in $\LCH(P) \cap (\lambda(P) \cup V)$, if the shortest path $\sigma_G(a,b)$ has length less than $\pi$ then all the points of $\sigma_G(a,b)$ are contained in $\LCH(P)$. 
In other words, we take the closure of $\lambda(P)$ in the link graph 
under the operation of taking shortest \revised{paths} between pairs of points, but only when the points  are vertices or correspond to points in $P$, and only when the paths have length less than $\pi$. 
Note that $\LCH(P)$ is not a subgraph of $G$ because in general it includes portions of edges.

We will show that $\lambda(\CH(P)) = \LCH(P)$, and that $\LCH(P)$ can be computed in a straight-forward way.  From $\LCH(P)$ we can readily identify the edges of $\cal K$ that contain points of the convex hull other than $O$---these correspond to vertices of the link graph in $\LCH(P)$.  We will also show how to use $\LCH(P)$ to decide if $O$ is in $\CH(P)$.

\begin{lemma} $\lambda(\CH(P)) = \LCH(P)$.
\label{lemma:CH-equiv}
\end{lemma}
\begin{proof}
%\revised{The proof of each inclusion is by structural induction based on the recursive definitions of $\CH(P)$ and $\LCH(P)$.}
We first prove $\LCH(P) \subseteq \lambda(\CH(P))$ 
\revised{by structural induction based on the recursive definition of $\LCH(P)$.}
As the base case we have  $\lambda(P) \subseteq \lambda(\CH(P))$.  
For the recursive step 
 let $a,b$ be two points in $\LCH(P) \cap (\lambda(P) \cup V)$ such that $|\sigma_G(a,b)| < \pi$.  Assume 
 \revised{by induction} 
 that $a,b \in \lambda(\CH(P))$, in particular, that $a = \lambda(a')$ and $b = \lambda(b')$ with $a', b' \in \CH(P)$.   Then \revised{$\sigma(a',b') \subseteq  \CH(P)$}.  
Since $|\sigma_G(a,b)| < \pi$, Proposition~\ref{property:path-thru-O} implies that 
$\lambda(\sigma(a',b')) = \sigma_G(a,b)$.  Therefore  $\sigma_G(a,b) \subseteq \lambda(\CH(P))$.

Next we prove $\lambda(\CH(P)) \subseteq \LCH(P)$
\revised{by structural induction based on the recursive definition of $\CH(P)$}.  
As the base case we have $\lambda(P) \subseteq \LCH(P)$.  For the recursive step, let $a$ and $b$ be two points in $\CH(P)$.  Assume 
\revised{by induction}
that $\lambda(a), \lambda(b) \in \LCH(P)$.  We must show that $\lambda(\sigma(a,b)) \subseteq \LCH(P)$.   If $|\sigma_G(\lambda(a), \lambda(b))| \ge \pi$ then by Proposition~\ref{property:path-thru-O} the shortest path from $a$ to $b$ in $\cal K$ goes through $O$.  In this case $\lambda(\sigma(a,b))$ consists only of points $\lambda(a)$ and $\lambda(b)$, which are in $\LCH(P)$ by assumption.  Thus we can restrict attention to the case  where $|\sigma_G(\lambda(a), \lambda(b))| < \pi$.  By Proposition~\ref{property:path-thru-O} this implies that $\lambda(\sigma(a,b)) = \sigma_G(\lambda(a), \lambda(b))$.  Thus we must show that $\sigma_G(\lambda(a), \lambda(b)) \subseteq \LCH(P)$, i.e.~that every point in $\sigma_G(\lambda(a),\lambda(b))$ lies in $\LCH(P)$.   

If $\lambda(a)$ and $\lambda(b)$ lie in $\lambda(P) \cup V$, this is immediate, but otherwise 
we must 
\revised{examine why $\lambda(a)$ and $\lambda(b)$ are in $\LCH(P)$.} 
%consider different endpoints.  
If $\lambda(a) \in \lambda(P) \cup V$, let $A  = A' = \lambda(a)$.  Otherwise, suppose that $\lambda(a)$ is an internal point of edge $e_a$ of $G$.  
%For ease of discussion, orient edge $e_a$ left to right so that $\sigma_G(\lambda(a),\lambda(b))$ contains points of $e_a$ just to the right of $\lambda(a)$.  
\revised{By the definition of $\LCH(P)$, there must be points $A$ and $A'$ in $\LCH(P) \cap (\lambda(P) \cup V)$
such that $\sigma_G(A,A')$ has length less than $\pi$ and includes $\lambda(a)$.
Choose such $A$ and $A'$ so that $| \sigma_G(A,A') | $ is minimum.  Then}
%$\lambda(a)$ was placed in $\LCH(P)$ as a point of the shortest path from $A$ to $A'$, say, where $A$ and $A'$ are points of $\lambda(P) \cup V$ that had already been placed in $\LCH(P)$, and $\sigma_G(A,A')$ has length less than $\pi$. 
%If $A$ or $A'$ lie outside $e_a$, we can replace them with endpoints of $e_a$ (which lie in $V \cap \LCH(P)$).  Thus we may assume that 
each of $A$ and $A'$ is either an endpoint of $e_a$ or a point of $\lambda(P)$ internal to $e_a$.   
%Suppose that $A$ is to the right of (or equal to) $A'$.

We now do the same for $b$.   If $\lambda(b) \in \lambda(P) \cup V$, let $B = B' = \lambda(b)$.  Otherwise, suppose that $\lambda(b)$ is an internal point of edge $e_b$ of $G$.  
%For ease of discussion, orient edge $e_b$ left to right so that $\sigma_G(\lambda(a),\lambda(b))$ contains points of $e_b$ just to the left of $\lambda(b)$. 
\revised{By the definition of $\LCH(P)$, there must be points $B$ and $B'$ in $\LCH(P) \cap (\lambda(P) \cup V)$
such that $\sigma_G(B,B')$ has length less than $\pi$ and includes $\lambda(b)$.
Choose such $B$ and $B'$ so that $| \sigma_G(B,B') | $ is minimum.  Then}
 %$\lambda(b)$ was placed in $\LCH(P)$ as a point of the shortest path from $B$ to $B'$ where $B$ and $B'$ are points of $\lambda(P) \cup V$ that had already been placed in $\LCH(P)$, and 
%$\sigma_G(B,B')$ has length less than $\pi$. 
%As above, we may assume that 
each of $B$ and $B'$ is either an endpoint of $e_b$ or a point of $\lambda(P)$ internal to $e_b$.  
%Suppose that $B$ is to the left of (or equal to) $B'$.

\revised{If $e_a \ne e_b$ then there must be a point of $\{A,A' \}$, say $A$, in $\sigma_G(\lambda(a),\lambda(b))$, and there must be a point of $\{B,B' \}$, say $B$, also in $\sigma_G(\lambda(a),\lambda(b))$.  
See Figure~\ref{f:link-graph-paths}.
Then $\sigma_G(A,B)$ has length less than $\pi$ since it is a subpath of $\sigma_G(\lambda(a),\lambda(b))$.
Now we have $\sigma_G(\lambda(a),\lambda(b)) \subseteq \sigma_G(A,A') \cup \sigma_G(A,B) \cup \sigma_G(B,B')$.}
All three of these paths have length less than $\pi$ and have endpoints in $\LCH(P) \cap (\lambda(P) \cup V)$.  Thus the paths lie in $\LCH(P)$, so
$\sigma_G(\lambda(a), \lambda(b)) \subseteq \LCH(P)$.
\revised{If $e_a = e_b$ we may still have points $A$ and $B$ in $\sigma_G(\lambda(a),\lambda(b))$---in which case the previous argument applies.  And otherwise $\sigma_G(\lambda(a),\lambda(b)) \subseteq \sigma_G(A,A')$, and this still gives $\sigma_G(\lambda(a), \lambda(b)) \subseteq \LCH(P)$.}
%
%Consider the union of $\sigma_G(A',A)$ and $\sigma_G(B,B')$.  If the union contains $\sigma_G(\lambda(a),\lambda(b))$ then we are done---every point of $\sigma_G(\lambda(a),\lambda(b))$ lies in $\LCH(P)$.  
%Otherwise (see Figure~\ref{f:link-graph-paths}), the order of points along $\sigma_G(\lambda(a),\lambda(b))$ is $\lambda(a), A, B, \lambda(b)$, so $\sigma_G(A,B)$ is contained in  $\sigma_G(\lambda(a),\lambda(b))$ and has length less than $\pi$.  
%Now $\sigma_G(\lambda(a),\lambda(b))$ is contained in the union of $\sigma(A',A)$, $\sigma(A,B)$ and $\sigma(B,B')$.  
%All three of these paths have length less than $\pi$ and have endpoints in $\LCH(P) \cap (\lambda(P) \cup V)$.  Thus the paths lie in $\LCH(P)$, so
%every point of $\sigma_G(\lambda(a), \lambda(b))$ lies in $\LCH(P)$.
\end{proof}

\begin{figure}[htb]
\centering
\includegraphics[scale=0.55]{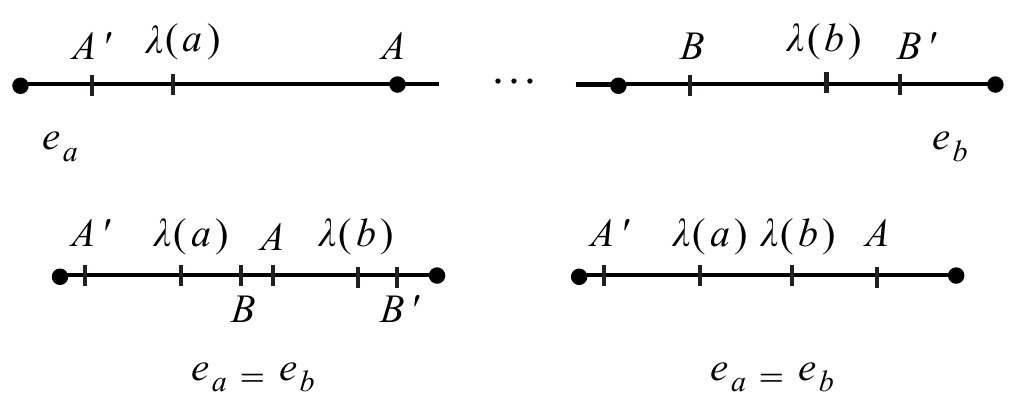}
\caption{Illustration for the proof of Lemma~\ref{lemma:CH-equiv}.  $\sigma_G(\lambda(a),\lambda(b)) \subseteq \sigma_G(A',A) \cup \sigma_G(B,B') \cup \sigma_G(A,B)$.  Top: $e_a \ne e_b$.  Bottom: two cases where $e_a=e_b$.}
\centering
\label{f:link-graph-paths}
\end{figure}

\medskip
\noindent{\bf Computing $\LCH(P)$.}
We can find the points of $\LCH(P)$ in $\lambda(P) \cup V$ as follows.  
We build up a set $A \subseteq \lambda(P) \cup V$.
Initially $A$ will just be the input set of points $\lambda(P)$, and at the end of the algorithm, $A$ will 
be the required set.
%also contain all vertices of the link graph that correspond to points in the convex hull.  
We will also keep a subset $F$ of $A$ that represents the ``frontier" that we still need to explore from.  Initially $F = A = \lambda(P)$.

The general step is to remove one element $v$ from $F$.  We then explore the part of the link graph within distance $< \pi$ from $v$.  This can be done by a depth-first search in $O(n+m)$ time.  A search tree to distance $\pi$ will find no cycles, and will therefore find shortest paths from $v$.   We remove the part of the depth-first tree that is beyond the deepest point of $A$ on each branch.  Then for every vertex $w$ of the link graph that is in the depth-first search tree, we check if $w$ is already in $A$---if not then we add $w$ to $A$ and to $F$.

The size of $A$ is bounded by $n+m$ where $n$ is the number of cells in $\cal K$ and $m$ is the size of $P$.
Note that the amount of work we do for one element of $F$ is $O(n+m)$.  Thus the algorithm runs in time $O((n+m)^2)$.

\medskip
\noindent{\bf Testing if $O$ is in the convex hull.}
If $O \in P$ we are done, so we must just deal with the case when $O \not\in P$.
%Then $O$ is in $\CH(P)$ if and only if there is a shortest path between two points in $\CH(P) - O$ that goes through $O$, and Claim~\ref{property:path-thru-O} gives the conditions for a path to go through $O$.  Thus:

%\begin{observation}
%$O$ is in $\CH(P)$ if and only if there are two points $a$ and $b$ in $\CH(P) - O$ such that the distance between $\gamma(a)$ and $\gamma(b)$ in the link graph is at least $\pi$.
%\label{obs:distance-pi}
%\end{observation}

\begin{lemma}
Suppose that $O \not\in P$.  Then $O$ is in $\CH(P)$ if and only if there are two points in $\LCH(P)$ such that the distance between them is at least $\pi$.
\label{lemma:distance-pi}
\end{lemma}
\begin{proof}
By definition, $O$ is in $\CH(P)$ if and only if there are two points $a$ and $b$ in $\CH(P) - O$ such that the shortest path between them goes through $O$.  By Proposition~\ref{property:path-thru-O} this is equivalent to there being two points in $\lambda(\CH(P))$ whose distance in $G$ is at least $\pi$.  By  Lemma~\ref{lemma:CH-equiv} $\lambda(\CH(P)) = \LCH(P)$ which gives the desired result.
\end{proof}

We can test if there are two points of $P$ whose shortest path goes through $O$.  
The remaining case is solved by the following lemma:

\begin{lemma}
\label{lem:origin-in-CH}
Suppose that $O \not\in P$ and no shortest path between two points of $P$ goes through $O$.  Then $O \in \CH(P)$ if and only if $\LCH(P)$ contains a cycle.
\end{lemma}

\begin{proof}
Suppose $\LCH(P)$ contains a cycle.  Because the space is CAT(0), the cycle has length at least $2\pi$, so it must contain two points $a,b$ whose minimum distance in the cycle is $\pi$.  
We claim that the shortest path from $a$ to $b$ in the link graph \Anna{has length} $\pi$---if there were a shorter path then, together with the path \megan{in the cycle} of length $\pi$ we would get a \megan{second} cycle of length $< 2\pi$.
Thus the shortest path from $a$ to $b$ has length $\pi$ and by Lemma~\ref{lemma:distance-pi}, $O$ is in $\CH(P)$.

For the other direction, suppose $\LCH(P)$ does not contain a cycle.  $\LCH(P)$ is connected, so it must be a tree.
We claim that the leaves of the tree are points of $\lambda(P)$:
If $d$ is a point of $\LCH(P)$ that is not in $\lambda(P)$, then $d$ 
was placed in $\LCH(P)$ because it is the internal point of some shortest path between points in $\LCH(P)$,
so $d$ has degree at least 2 in  $\LCH(P)$, so it is not a leaf.

Let $a$ and $b$ be points of $\LCH(P)$.  The path between $a$ and $b$ in the tree $\LCH(P)$ can be extended to a path between leaves of $\LCH(P)$, and, since the leaves are in $\lambda(P)$, this path has length less than $\pi$.  
Thus by Lemma~\ref{lemma:distance-pi}, $O$ is not in the convex hull.
%Thus there is a path between $a$ and $b$ in $\LCH(P)$ of length less than $\pi$.  This must be the shortest path between $a$ and $b$ in $G$.   Thus, $\LCH(P)$ is closed under taking shortest paths.  We will never get two points at distance $\pi$ or more.  Therefore $O$ is not in the convex hull.
\end{proof}

\remove{  %old version
First consider the problem of testing whether $O$ is in the convex hull.
If $O \in P$ we are done, so assume that $O \not\in P$.
Then $O$ is in $\CH(P)$ if and only if there is a path between two points in $\CH(P) - O$ that goes through $O$, and Proposition~\ref{property:path-thru-O} gives the conditions for a path to go through $O$.  Thus:

\begin{observation}
$O$ is in $\CH(P)$ if and only if there are two points $a$ and $b$ in $\CH(P) - O$ such that the distance between $a$ and $b$ in the link graph is at least $\pi$.
\label{obs:distance-pi}
\end{observation}

We test if $O$ is in $\CH(P)$ as follows.
If there are two points $p,q \in P$ whose distance in the link graph is at least $\pi$, then $O$ is in $\CH(P)$.
Otherwise, let $G[P]$ be the subset of the link graph that is the union of all shortest paths $\sigma_G(p,q), \ p,q \in P$.  Note that $G[P]$ is not a subgraph of $G$ in the usual sense because in general it includes portions of edges.
Every point of $G[P]$ is the image of some point in $\CH(P)$.  By the following lemma, it suffices to test if $G[P]$ has a cycle.

\begin{lemma}
\label{lem:origin-in-CH}
Suppose that no path between two points of $P$ goes through $O$.  Then $O \in \CH(P)$ if and only if $G[P]$ contains a cycle.
\end{lemma}

\begin{proof}
Suppose $G[P]$ contains a cycle.  Because the space is CAT(0), the cycle has length at least $2\pi$, so it must contain two points $a,b$ whose minimum distance in the cycle is $\pi$.  Then the length of the shortest path between $a$ and $b$ in the link graph must be $\pi$ (otherwise we would have a cycle of length less than \note{$2\pi$}).  By Observation~\ref{obs:distance-pi}, $O$ is in $\CH(P)$.

For the other direction, suppose $G[P]$ does not contain a cycle.  $G[P]$ is connected, so it must be a tree.
We claim that the leaves of the tree are points of $P$:
If $d$ is a point of $G[P]$ that is not in $P$ then $d$ is an internal point of some path $\sigma_G(p,q), \ p,q \in P$, so $d$ has degree at least 2 in  $G[P]$, so it is not a leaf.

Let $a$ and $b$ be points of $G[P]$.  The path between $a$ and $b$ in $G[P]$ can be extended to a path between leaves of $G[P]$, and, since the leaves are in $P$, this path has length less than $\pi$.  Thus there is a path between $a$ and $b$ in $G[P]$ of length less than $\pi$.  This must be the shortest path between $a$ and $b$ in $G$.   Thus, $G[P]$ is closed under taking shortest paths.  We will never get two points at distance $\pi$ or more.  Therefore $O$ is not in the convex hull.
\end{proof}

The algorithm to test if $O$ is in CH$(P)$ has a straight-forward implementation that runs in time $O(m(n+m))$ where $n$ is the number of cells in $\cal K$ and $m$ is the size of $P$.
For each point $p \in P$, we do a depth-first search in the link graph to find paths to all the points of $P$ within distance $\pi$ of $p$.  Note that a search to distance $\pi$ will not find any cycles in the link graph, and therefore finds shortest paths from $p$.  If some point $q \in P$ is not reached then we know that $O$ is in CH$(p)$. Otherwise, we construct $G[P]$ as the union of all these depth-first search trees, and test if $G[P]$ contains a cycle.  Each depth-first search takes time $O(n+m)$ so the total time is $O(m(n+m))$.  Constructing and exploring $G[P]$ takes linear time.

\medskip
\noindent{\bf Finding edges of $\cal K$ intersecting the convex hull.}

We now show how to identify the edges of the complex $\cal K$ that contain points of CH$(P) - O$.  Equivalently, we will identify the vertices of the link graph that correspond to points in the convex hull.  
Our method applies whether or not $O$ is inside the convex hull.

Let $V$ be the set of vertices of the link graph.  Then $V$ corresponds to the set of edges of $\cal K$.  Let $B$ be the vertices of the link graph that correspond to points in the convex hull of $P$.  We want to find $B$.
The convex hull of $P$ in the complex $\cal K$ is defined to be the closure of $P$ under shortest paths.  
We observe that it suffices to consider shortest paths between points that are in $P$ or lie on the edges of $\cal K$. \note{Should we prove this?}

In the link graph this corresponds to taking shortest paths between points that lie in $P \cup V$.  Furthermore, by Proposition~\ref{property:path-thru-O} we should only take shortest paths in the link graph that have length less than $\pi$.

Therefore $P \cup B$ is the closure of the set $P$ in the set $P \cup V$ under the operation of taking shortest paths of length less than $\pi$ in the link graph. 

Our algorithm finds this closure by building up a set $S \subseteq P \cup V$.
%We will build up a set $S$ of points and vertices in the link graph.  
Initially $S$ will just be the input set of points $P$, and at the end of the algorithm, $S$ will 
be the required set $ P \cup B$.
%also contain all vertices of the link graph that correspond to points in the convex hull.  
We will also keep a subset $F$ of $S$ that represents the ``frontier" that we still need to explore from.  Initially $F = S = P$.

The general step is to remove one element $v$ from $F$.  We then explore the part of the link graph within distance $\pi$ from $v$.  This can be done by a depth-first search in $O(n+m)$ time.  As noted above, a search tree to distance $\pi$ will find no cycles, and will therefore find shortest paths from $v$.   We remove the part of the depth-first tree that is beyond the deepest point of $S$ on each branch.  Then for every vertex $w$ of the link graph that is in the depth-first search tree, we check if $w$ is already in $S$---if not then we add $w$ to $S$ and to $F$.

The size of $S$ is bounded by $n+m$ where $n$ is the number of cells in $\cal K$ and $m$ is the size of $P$.
Note that the amount of work we do for one element of $F$ is $O(n+m)$.  Thus the algorithm runs in time $O((n+m)^2)$.

The algorithm is correct because the final set $S$ is closed under taking shortest paths of length less than $\pi$ in the link graph.  Each time we add a point to $S$ we explicitly check all points of $V$ within distance $\pi$ from the point.
} % end remove old version

%%%%%%%%%%%%%%%%%%%%%%%%%%%%%%%%%%%%%%%%%
\subsubsection{Finding the Convex Hull via Linear Programming}
\label{sec:LP-alg}

In this section we give 
%an efficient
%%a polynomial time 
%algorithm to explicitly find the convex hull of a finite point set $P$ in a 
%2D CAT(0) complex $\cal K$ with a single vertex $O$.
a polynomial-time algorithm to construct a linear program to find the convex hull of a finite point set $P$ in a 
2D CAT(0) complex $\cal K$ with a single vertex $O$.
%\revised{\sout{Our algorithm implies that the convex hull is closed in this case.}}

Recall the skeletons, $S_i$, from Section~\ref{sec:CH-alg}, which give an iterative way of computing the extreme points of the convex hull along each edge of the complex.  As $i$ increases, the extreme points expand outwards.
The idea of our linear program is to have two variables for each edge $e$ that represent the 
two extreme points 
%limits (max and min) of the points 
of $S_i$ on $e$.  The linear constraints will express the closure-under-shortest-paths property that was used to construct $S_{i+1}$ from $S_i$.  
Thus feasible solutions to the linear program will represent limit points $S$ of the $S_i$'s.
%Since the set of feasible solutions to a linear program is a closed set, 
%In this way, solving the linear program will give the limit points of $S_i$ on each edge $e$. 
From these points, we can compute $\CH(P)$, as justified by Corollary~\ref{cor:CH-characterization}.
Furthermore, the computation of $\CH(P)$ from the set $S$ is efficient since it simply involves computing the Euclidean convex hull inside each cell $C$, and Euclidean (planar) convex hulls can be computed in polynomial time~\cite{ChanCH}.
 We now fill in the details of this plan.

%Suppose first that the origin $O$ is inside the convex hull.  The case where the origin is not inside the convex hull will be dealt with later.
%

\remove{
To find the convex hull
it suffices to describe the intersection of the convex hull with each cell of the complex.
Consider a cell $C$ bounded by rays $e$ and $f$ incident to $O$.  The part of the convex hull inside $C$ is a convex polygon determined by its vertices which consist of: point $O$; \changed{possibly some or all of} the points of $P$ inside $C$; and
two points $x_e$ and $x_f$ on edges $e$ and $f$, respectively, that are on the boundary of the  convex hull.
Note that we do not make any assumption about whether $x_e$ and $x_f$ are in the convex hull, because we are not making any assumption about whether the convex hull is open or closed.
If we knew the points $x_e$ and $x_f$, then we could easily compute the part of the convex hull inside $C$, since it is simply the Euclidean convex hull of $x_e$, $x_f$, $O$, and the points of $P$ inside $C$.
}

\remove{  % this material moved to an earlier section
There is an obvious iterative approach to finding the points $x_\ell$ where the boundary of the convex hull intersects each edge $\ell$ of the complex:
%\changed{To find the values of $x_\ell$ in the convex hull, t}here is an obvious iterative approach: 
Initialize $H_0$ to be the set $P$.  For $i=1, 2, \ldots$, initialize $H_i$ to $H_{i-1}$ and then take every pair of  points from $H_{i-1}$, compute the shortest path $\sigma$ between them, and add to $H_i$ all the intersection points of $\sigma$ with edges of the complex. 
(When two points lie on the same edge, we can discard the one closer to $O$.)
We observe that this process is finite if and only if the convex hull is closed.
As mentioned, our proof of Theorem~\ref{thm:convex-hull-alg} shows that the convex hull is closed and thus this iterative process is finite.
In order for the iterative process to be an efficient algorithm we would need a bound on the number of iterations.
We conjecture that there is a polynomial bound:
%, but we do not even have a finite bound \note{(actually I think we do)}, although we make the following conjecture: 
%\changed{(which implies that the convex hull is closed)}:

\begin{conjecture}
$H_k = H_{k+1}$ for some polynomially bounded $k$.
\label{conj:iteration}
\end{conjecture}

Lacking a proof of the conjecture, we
%will find the values for $x_\ell$ using linear programming. 
will use linear programming to find the points where the boundary of the convex hull crosses the edges of the complex.
}

Our algorithm and our notation will be simpler
if the points of $P$ all lie on edges of the complex.  
In particular, the convex hull inside a cell $C$, if non-empty,  will be a triangle or quadrilateral (depending on whether $O$ is in the convex hull) since no points of $P$ will be internal to $C$. 
We can achieve this by constructing a new edge $e_p$ from $O$ through each point $p \in P$ (except the point $O$). 
Each such edge divides a cell in two.
Point $p$ is then represented in local coordinates by the distance along edge $e_p$ from $O$ to point $p$. We will use $p$ to refer both to the point and to its local coordinate, i.e.,~its distance from $O$ (in the same way that we refer to a point on the real line as a number).

Let $B$ be the set of edges of the complex that contain points other than $O$ inside the convex hull.  
These correspond to points of $\LCH(P)$ in $V$ and can be found as described in the previous section.

%Let $B' = B \cup \{e_p\}_{p \in P}$.

For clarity of presentation we will separate into two cases depending on whether the origin is inside the convex hull.

\medskip
\noindent
{\bf When the origin $O$ is inside the convex hull.}
For each $\ell \in B$ our linear program will have %make
a variable $x_\ell \in {\mathbb R}$ representing the 
distance from $O$ to the
point on $\ell$ that is on the boundary of the convex hull.  
Then $x_\ell > 0$.  
Note that, like $p$,  $x_\ell$ refers both to a distance and a point.

Our inequalities are of two types.  First, for any point 
$p \in P$
%, \Anna{where $p$ lies on edge $\ell$ and the distance from $p$ to $O$ is $p_d$,}
lying on an edge $\ell \in B$ 
 we include the inequality:  
 
\begin{align}
x_\ell \ge p
%x_\ell \ge p
\label{eq:LP1}
\end{align}

Inequalities of the second type will be determined by pairs of elements from the set $B$.
%Our will be determined by pairs of elements from the set $L = P \cup B$.
%Notation will be eased by viewing elements of $P$ and of $B$ uniformly, which we can do by   
%constructing an edge $e_p$ through each point of $p \in P$ (subdividing its cell) and associating $p$ with the %variable $x_{e_p}$ which has a known constant value.  
For any two edges $e$ and $f$ of $B$, %For any two variables $x_e$ and $x_f$, 
such that the angle between $e$ and $f$ is $< \pi$,
consider the shortest path $\sigma$  between the corresponding points $x_e$ and $x_f$.
We will add a constraint for each edge $\ell$ of $B$ crossed by $\sigma$, expressing the fact that the convex hull includes the point where $\sigma$ crosses $\ell$.  The constraint has the form $x_\ell \ge t$ where $t$ is 
the distance from $O$ to 
the point where $\sigma$ crosses $\ell$.
We will use $t$ to refer to both the distance and to the point.
%(More precisely, $t$ is the distance from $O$ to the crossing point.)
We can express $t$ in terms of known quantities.  The set-up is illustrated in Figure~\ref{f:general_constraints}. 
Note that there may be several polyhedral cells separating $e$ and $f$, but we can %lay them down 
\revised{unfold them in the plane} to form a triangle.

%\remove{
%$$d_\ell \ge \chi_\ell(p,q)$$

%We will give an expression for $\chi_\ell(p,q)$ below.  Before that, we show how to identify the edges of $\cal K$ crossed by $\sigma(p,q)$ using the link graph $G = G_O$.  Points $p$ and $q$ in $B$ correspond to points of $G$.  Any point's distance from $O$ is lost in this mapping, and any point on ray $\ell$ maps to vertex $\ell$ in $G$.  We use the notation $\sigma_G(p,q)$ to denote the shortest path between $p$ and $q$ in $G$.
%If $\sigma_G(p,q)$ has length $\ge \pi$ then $\sigma(p,q)$ goes through the origin and does not cross any edges.  Otherwise, the edges crossed by $\sigma(p,q)$ in $\cal K$ correspond to the vertices along $\sigma_G(p,q)$ in $G$.
%}

\begin{figure}[htb]
\centering
\includegraphics[scale=0.45]{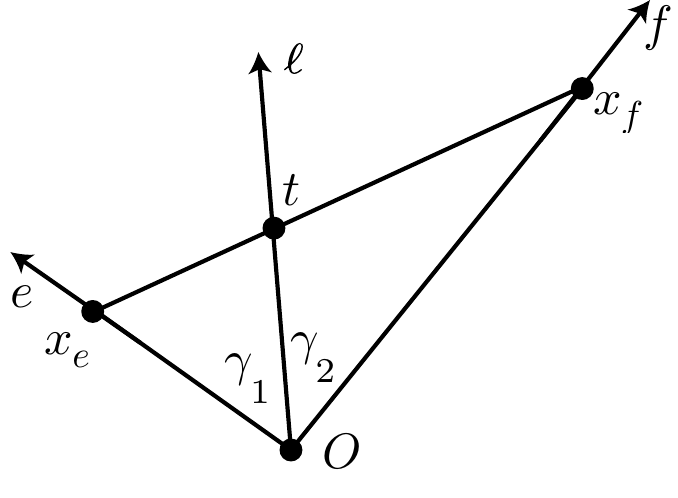}
\caption{Expressing the intersection point $t$ in terms of known quantities.}
\centering
\label{f:general_constraints}
\end{figure}

Let $\gamma_1$ be the angle between $e$ and $\ell$, and let $\gamma_2$ be the angle between $\ell$ and $f$.
Then we have:   %\note{Note: I added this as a claim.}

\begin{claim}
\begin{align*}
t = \frac{x_e x_f\sin(\gamma_1 + \gamma_2)}{x_e \sin \gamma_1 + x_f \sin \gamma_2}
\end{align*}
\end{claim}
\begin{proof}
\revised{Let $A$ denote the area of a triangle.  Observe that $A(O,x_e,t) + A(O,t,x_f) = A(O,x_e,x_f)$. 
Applying the sine law for area of a triangle, we obtain 
\begin{align*}
\frac{1}{2}x_e t \sin \gamma_1 +  \frac{1}{2}t x_f \sin \gamma_2  =  \frac{1}{2}x_e x_f \sin (\gamma_1 + \gamma_2).
\end{align*}
Rearranging gives the required formula for $t$.
}
\end{proof}

\remove{ %old proof
We apply the sine law several times.
%For completeness we include a derivation of the above formula.  
Let  $\alpha$ and $\beta$ be the angles opposite sides $Ox_e$ and $Ot$, respectively, in the triangle $O x_e t$, and let $c$ be the distance from $x_e$ to $x_f$.
By the sine law
\begin{align*}
t =  \frac{x_e \sin \beta}{\sin \alpha}, \quad \sin \beta =  \frac{x_f \sin(\gamma_1 + \gamma_2)}{c}. 
\end{align*}
Combining these we get
\begin{align*}
t = \frac{x_e x_f \sin(\gamma_1 + \gamma_2)}{c \sin \alpha}.
\end{align*}
Let $c_1$ and $c_2$ be the parts of $c$  
opposite $\gamma_1$ and $\gamma_2$, respectively.
By the sine law, we also have $c_1 \sin \alpha  = x_e \sin \gamma_1$ and $c_2 \sin(\pi - \alpha) = c_2 \sin \alpha = x_f \sin \gamma_2 $.  So $c \sin \alpha  = (c_1 + c_2) \sin \alpha = x_e \sin \gamma_1  + x_f \sin \gamma_2$.  Substituting into the above expression for $t$ gives the desired formula.
}

Using the above claim, the constraint $x_\ell \ge t$ becomes
\begin{align*}
x_\ell \ge \frac{x_e x_f \sin(\gamma_1 + \gamma_2)}{x_e \sin \gamma_1 + x_f \sin \gamma_2}
\end{align*}

This is not a linear inequality, but substituting 
$y_\ell = {1 \over x_\ell} $ yields 

\begin{align}
y_\ell \le  y_f \frac{\sin \gamma_1}{\sin(\gamma_1 + \gamma_2)}  + y_e \frac{\sin \gamma_2}{\sin(\gamma_1 +  \gamma_2)}
\label{eq:LP}
\end{align}

Since the $\gamma_i$'s are constant, this is a linear inequality.  

The inequalities (\ref{eq:LP1}) for point $p \in P$ on edge $\ell$ become
\begin{align}
y_\ell \le  \frac{1}{p}
\label{eq:LP1'}
\end{align}

\begin{lemma} %\note{[I made this into a Lemma. AL.]} 
Maximizing $\sum y_\ell$ subject to the inequalities (\ref{eq:LP}), (\ref{eq:LP1'}) and $y_\ell \ge 0$ gives the 
\revised{closure of the}
convex hull of $P$, \note{i.e.,~gives the points $x_\ell = \frac{1}{y_\ell}$ where the \revised{closure of the} convex hull intersects each edge $\ell \in B$}. 
\end{lemma}
\begin{proof}  
\note{
Recall the definition of the limit set $S$ from the beginning of Section~\ref{sec:CH-alg}.
Note that the points  $S$ provide 
a feasible solution  to the linear system, because 
they satisfy 
the constraints $x_\ell \ge p$ and $x_\ell \ge t$ which we used to construct our inequalities. 
Denote this solution by $y^\CH_\ell, \ell \in B$. 

Next, note that any other solution $y'_\ell, \ell \in B$, has $y'_\ell \le y^\CH_\ell$
 for all $ \ell \in B$, \note{i.e.~$x'_\ell \ge x^\CH_\ell$}---in other words, any other solution includes
 $S$. 
% $\CH(P)$.  
 This is because the points of $P$ are included, and the inequalities of our linear system enforce closure under shortest paths. 
Therefore, the solution that maximizes $\sum y_\ell$ gives the points $S$, and thus 
\revised{the set $H$, which is the closure of $\CH(P)$.}
%the convex hull of $P$.
} 
\end{proof} 

%\begin{corollary} % better to defer to the end
%$\CH(P)$ is closed.
%\end{corollary}

%\remove{
%\vspace{.2in}
%\note{I do not think we should list the following as a theorem.  Because we already stated the main theorem above.  Also if we do list the following, then we need another theorem for the case when $O$ is outside the CH.  Also in the shortest path section, we don't have a Theorem to express our algorithm.  I don't mind making the above argument a formal Claim if you want.  A.L.  MO:  I see now that you put a theorem with our result at the beginning of the section, which takes care of my concern about citing this.} 

%\changed{
%\begin{theorem}  
%Given a CAT(0) 2D complex $\cal K$ with a single-vertex $O$ and input points $P$, if $O$ is in the convex hull, then the vertices of the convex hull $\CH(P)$ are given by the solution $\{y_{\ell}\}_{\ell \in L}$ to the linear programming problem described above.  
%\end{theorem}

%\begin{proof}
%Since $L = P \cup B$ is a set of edges in the complex, Let $Y$ by the set in $\cal K$ with vertices $\{y_{\ell}\}_{\ell in L}$.   By construction, the convex hull $\CH(P)$ is contained in $Y$. ...
%\end{proof}

%} }

%Thus we have reduced the problem of finding the convex hull to linear programming in the case when $O$ is in the convex hull. 
This completes the reduction to linear programming when $O$ is in the convex hull.

%%%%%%%%%%%%%%%%
% now the case when O is not in the convex hull (later on comes the timing analysis

\note{
\medskip
\noindent
{\bf When the origin $O$ is not inside the convex hull.}
}
%We now deal with the case where the origin $O$ is not inside the convex hull.
In this case, by Lemma~\ref{lem:origin-in-CH}, the subgraph of the link graph corresponding to the convex hull is a tree, and it seems even more plausible that an efficient iterative approach can be used to find the convex hull.  However, we leave this as an open question, and give a linear programming approach like the one above.  

%As before, we will construct a new edge through every point of $P$, and 
%let $B'$ be the set of edges (including the newly constructed ones) that have points other than $O$ inside the convex hull.  
For each $\ell \in B$ we will make two variables, $x_\ell^{\rm min}$ and $x_\ell^{\rm max}$ in ${\mathbb R}$ representing the minimum and maximum points on $\ell$ that are on the boundary of the convex hull. 
To find the 
\revised{closure of the} 
convex hull, it suffices to find the values of these variables.

We want to ensure that $x_\ell^{\rm max}$ is larger than 
any point of $P$ and 
%any path's crossing point on edge $\ell$ 
\changed{any point at which a geodesic between $x_e^{\rm max}$ and $x_f^{\rm max}$, for any $e, f \in B$, crosses edge $\ell$.
Similarly, we want to ensure }
that $x_\ell^{\rm min}$ is smaller than any point of $P$ and
% any path's crossing point on edge $\ell$.  
\changed{any point at which a geodesic between $x_e^{\rm min}$ and $x_f^{\rm min}$, for any $e, f \in B$, crosses edge $\ell$.}
Using the same notation and set-up as above with edges $e$ and $f$, and using the inverse variables
$y_\ell^{\rm min}= {1 \over x_\ell^{\rm min}}$ and  
$y_\ell^{\rm max}= {1 \over x_\ell^{\rm max}}$ 
the inequalities corresponding to (\ref{eq:LP}) are: 
%these properties are captured by the following two inequalities:

\begin{align*}
y^{\rm max}_\ell \le  y^{\rm max}_f \frac{\sin \gamma_1}{\sin(\gamma_1 + \gamma_2)}  + y^{\rm max}_e \frac{\sin \gamma_2}{\sin(\gamma_1 +  \gamma_2)}\\
y^{\rm min}_\ell \ge  y^{\rm min}_f \frac{\sin \gamma_1}{\sin(\gamma_1 + \gamma_2)}  + y^{\rm min}_e \frac{\sin \gamma_2}{\sin(\gamma_1 +  \gamma_2)}
\end{align*}

The inequalities corresponding to (\ref{eq:LP1'}) are:
\begin{align*}
y^{\rm max}_\ell \le  \frac{1}{p} \le y^{\rm min}_\ell
\end{align*}

If we maximize the objective function $\sum (y^{\rm max}_\ell - y^{\rm min}_\ell)$ subject to the above inequalities and $y^{\rm min}_\ell \ge y^{\rm max}_\ell \ge 0$ then, by a similar argument to the one above, this gives \revised{the closure of} the convex hull of $P$.
Thus we have reduced the problem of finding the convex hull to linear programming.  

%%%%%%%%%%%%%%%%%%%
% timing analysis

\medskip
\noindent
{\bf Running time.}
%We now discuss the running time of the algorithm.  
We will concentrate on the case where $O$ is in the convex hull---the other case is similar.
Recall that $n$ is the number of cells in the complex and $m$ is the number of points in $P$.
At the beginning of the algorithm we test if $O$ is in the convex hull, and find the set $B$ of edges of the complex that contain points of the convex hull other than $O$.  This takes $O((n+m)^2)$ time as discussed in the previous section.   The set $B$ has size $O(m+n)$ because it includes an edge of the complex through every point of $P$.   The linear program has $O(n+m)$ variables.  The number of inequalities is $O((n+m)^3)$ since we consider each pair of elements, $e, f$  from $B$, and add an inequality for each edge of the complex crossed by the shortest path from $e$ to $f$. 
We can construct the linear program in polynomial time assuming a real RAM model of computation.
%\rednote{Fix this some more.}
\note{We need more than arithmetic operations in our real RAM, but what we need depends on exactly how the input is given.
The algorithm, as written above, assumes that the input triangles are given in terms of angles.  In that case, our real RAM must be able to compute sines of those angles.  An alternative is that the input is given to us with each triangle expressed in a local coordinate system.  In that case, the sines can be computed from the local coordinates so long as our real RAM includes the square root operation.
}  
%that allows us to do computations on the input angles (including computing the sine of angles).  

\medskip
\noindent{\bf The special case of a cube complex}

In the special case of a cube complex, it is more natural to give each input point using $x$- and $y$-coordinates relative to the quadrant containing the point.  In this case, we claim that all the low-level computations described above can be performed in polynomial time when measuring bit complexity.  
We will not construct new edges through points of $P$ since that introduces new angles.  
Our variables are $x_e$ for $e$ an edge of the complex, and we add constraints for shortest paths between  
%Instead, we will explicitly look at shortest paths between 
pairs of points in $P \cup \{x_e\}$.
%$P \cup \{x_e : e \in B\}$.	// changed format for consistency with next paragraph (where I wanted to make the expression more compact)

We give a few more details for the computation of $t$ in Figure~\ref{f:general_constraints} in this case.  
Figure~\ref{f:general_constraints} shows a path between $x_e$ and $x_f$ crossing an edge $\ell$ at $t$.  In the current setting, $\ell$ will be an edge of $\cal K$, and $x_e$ and $x_f$ may be variables or input points.  
If $x_e$ and $x_f$ %in the Figure 
correspond to input points, then $t$ is just the point where a line between two known points crosses an axis.  
Then the right-hand-side of the corresponding constraint~(\ref{eq:LP}) is a constant whose bit complexity is polynomially bounded in terms of the input bit complexity.  
The case when both $x_e$ and $x_f$ are variables cannot arise because they would be distance $\pi$ apart in the link graph.  Thus the only case we must take care of is when $x_e$ corresponds to an input point and $x_f$ is a variable (or vice versa).  The situation is shown in Figure~\ref{f:cube_constraints}, with $x_e$ being an input point $p$ with coordinates $(h,v)$ as shown.  Then $t=v x_f / ( x_f + h)$,  
%$t=v {x_f \over x_f + h}$, 
so constraint~(\ref{eq:LP}) becomes $y_\ell \le (h y_f + 1) / v$. 
%$y_\ell \le {1 \over v} (1 + h y_f)$.
Thus the coefficients in our linear constraints are rationals whose bit complexity is polynomially bounded in terms of the input bit complexity.
This means that we can use polynomial-time linear programming algorithms~\cite{khachiyan,karmarkar}
to find the value of $x_e$ for each edge $e$ of $\cal K$.  

\remove{
\rednote{This should go/change.}
Finally, as noted at the beginning of this section, we can use the values $x_e$ to compute the convex hull in each quadrant $C$.  We repeat the argument here. 
Suppose $C$ is bounded by edges $e$ and $f$.  
% I changed this so that it matches the earlier explanation more exactly
Then the convex hull inside $C$ is the Euclidean convex hull of points $x_e$, $x_f$, $O$, and the points of $P$ that lie inside $C$,
and since Euclidean (planar) convex hulls can be computed in polynomial time~\cite{ChanCH}, this allows us to compute the convex hull in $C$, and thus the whole convex hull, in polynomial time.  
% if we quote bounds it should be as follows
%A standard planar convex hull algorithm~\cite{ChanCH} will find the convex hull inside $C$ in 
%time $O(m_C \log m_C)$ where $m_C$ is the number of points of $P$ inside $C$, or $O(n + m \log m)$ for all quadrants.  Thus we can find the whole convex hull in polynomial time.
}

\remove{
For each quadrant $C$, let $P_C$ be those points in $P \cup \{x_e\}_{e \in B} \cup O$ contained in its interior or boundary.  We compute their convex hull, $CH(P_C)$, using an optimal algorithm for computing the convex hull of points in the plane.  This can be done in time $O(m \log m)$ for each quadrant, or $O(nm \log m)$ for all quadrants.  Thus we can find the convex hull in polynomial time.
}

\begin{figure}[htb]
\centering
\includegraphics[scale=0.45]{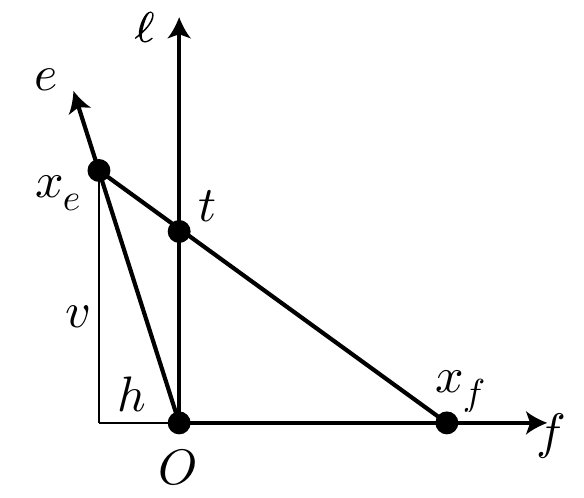}
\caption{Expressing the intersection point $t$ in terms of known quantities in the case of a cube complex.}
\centering
\label{f:cube_constraints}
\end{figure}
  
This completes the proof of Theorem~\ref{thm:convex-hull-alg}.

\section{Shortest Paths}
\label{sec:shortest-paths}

In this section we explore the possibilities and limitations of 
using the shortest path map to solve the single-source shortest path problem in a 2D CAT(0) complex.
%which divides the space into regions where shortest paths from $s$ are combinatorially the same (precise definition below).}
The input is a 2D CAT(0) complex, $\cal K$,  composed of $n$ triangles, and a ``source'' point $s$ in $\cal K$.
We denote the shortest path from $s$ to $t$ by $\sigma(s,t)$.
\note{Throughout this section, the cells of $\cal K$ will be called ``faces''.  }

In general, the \emph{shortest path map} partitions the space into regions in which all points have shortest paths from $s$ that have the same \emph{combinatorial type}.  Specialized to 2D CAT(0) complexes, two shortest paths have the same \emph{combinatorial type} if they \note{traverse}
% cross
 the same sequence of edges, vertices, and faces.
 A basic approach to the single source shortest path problem is to compute the whole shortest path map from $s$.

We first show that the shortest path map may have exponential size \note{for a general 2D CAT(0) complex}.  
This contrasts with the fact that the shortest path map has size $O(n^2)$ in the two special cases 
 where the single-source shortest path problem is known to be efficiently solvable: when the complex is a topological 2-manifold with boundary, which we will call a 2-manifold for short~\cite{Maftuleac}; and when the complex is rectangular~\cite{Chepoi-Maftuleac}.    

 We then show that for any 2D CAT(0) complex there is a structure called the ``last step shortest path map'' that coarsens the shortest path map, has size $O(n)$, and allows us to find the shortest path $\sigma(s,t)$ to a given target point $t$ in time proportional to the number of triangles and edges traversed by the path.  
 Although we do not know how to find the last step shortest path map in polynomial time for general 2D CAT(0) complexes, we can obtain it from the shortest path map.

From this, we obtain efficient algorithms for the single-source shortest path problem in 2D CAT(0) complexes that are 2-manifold or rectangular.  
Both cases had been previously solved, but the techniques used in the two cases were quite different.  Our approach is the same in both cases and opens up the possibility of solving other cases.
We need $O(n^2)$ preprocessing time and space to construct a structure that uses $O(n)$ space and allows us to  find the shortest path $\sigma(s,t)$ to a given target point $t$ in time proportional to the number of triangles and edges traversed by the path.
\note{This matches the previous time bounds.}
%This improves the bounds for 2-manifolds, and is competitive for rectangular complexes (where the published bounds are for two-point queries, so the situation is not quite comparable).
%\note{The improvement for rectangular complexes is pretty minimal.  They use $O(n^2)$ preprocessing time and space and get a structure of size $O(n^2)$ and answer queries in maybe a bit more time(?).  We only reduce the size of the structure (but we still need the $O(n^2)$ space at the start). Also, they solve all-pairs (which is more general).  How much should we say?}

\subsection{The Shortest Path Map}

% intro to shortest path map
Typically in a shortest path problem, the difficulty is to decide which of multiple geodesic (or locally shortest) paths to the destination is shortest.
This is the case, for example, for shortest paths in a planar polygon with holes, or for shortest paths on a terrain, and is a reason to use a Dijkstra-like approach that explores
paths to all target points in order of distance.   For shortest paths on a terrain, Chen and Han~\cite{Chen-Han} provided an alternative that uses a Breadth-First-Search (BFS) combined with a clever pruning when two paths reach the same target point.

When geodesic paths are unique, however, it is enough to explore all geodesic paths, and there is no need to explore paths in order of distance or in BFS order.  This is the case, for example, for shortest paths in a polygon, where the ``funnel'' algorithm~\cite{Guibas-sh-path-87,Hershberger-Snoeyink} achieves $O(n)$ processing time and storage, and $O(\log n)$ query time (plus output size to produce the actual path).
Similarly, in CAT(0) spaces, the uniqueness of geodesic paths means we can obtain a correct algorithm by simply exploring all geodesic paths without any ordering constraints.
%A basic approach to the single source shortest path problem is to compute the whole shortest path map from $s$.  
%The \emph{shortest path map} partitions the space into regions in which all points have shortest paths from $s$ that have the same \emph{combinatorial type}.  Specialized to 2D CAT(0) complexes, two shortest paths have the same \emph{combinatorial type} if they \note{traverse}
% the same sequence of edges, vertices, and faces.

For a vertex $v$ in a 2D CAT(0) complex, we define the \emph{ruffle} of $v$ to be the set of points $p$ in the complex such that the shortest path from $s$ to $p$ goes through $v$.  See Figure~\ref{fig:ruffle} for an example in the case of a rectangular complex.
The points of the ruffle of $v$ in a small neighbourhood of $v$ can be identified from the link graph of $v$ together with the incoming ray which is the last segment of the shortest path $\sigma(s,v)$.  In particular, the points of $v$'s ruffle close to $v$ are those points $p$ for which the segment $vp$ makes an angle of at least $\pi$ with the incoming ray. 
Using the link graph, the boundary rays of the ruffle of $v$ can be identified in time proportional to the number of faces incident to $v$.

\begin{figure}[htb]
\centering
\includegraphics[width=0.75\textwidth]{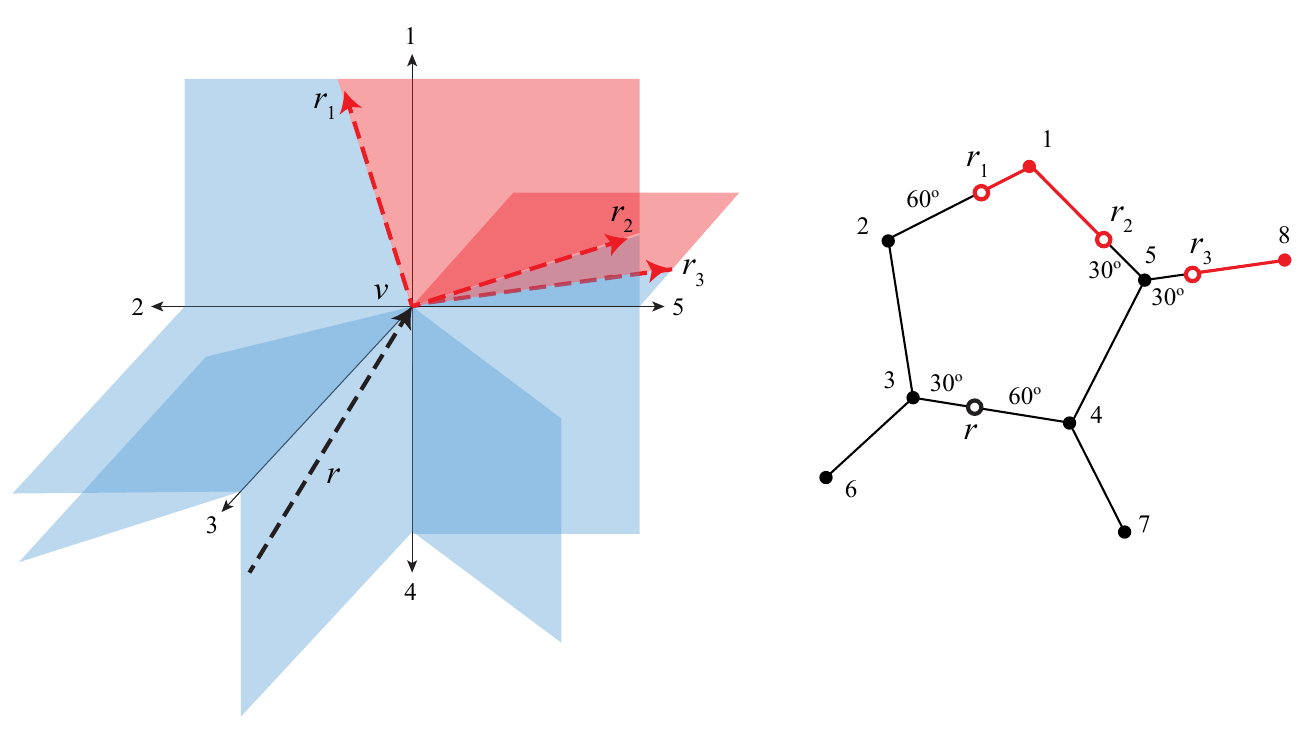}
\caption{The ruffle (in red) of vertex $v\in\mathcal K$ with respect to incoming ray $r$, shown in $\cal K$ (left) and in the link graph $G_v$ (right). The boundary rays of the ruffle are $r_1$, $r_2$, and $r_3$.}
\label{fig:ruffle}
\end{figure}

Consider one region of the shortest path map, and the set, 
$C$, of shortest paths  to points in the region.  
%If $C$ consists of just one path, we call it a \emph{singleton} and otherwise we call it \emph{full}.
The paths in $C$ all go through the same sequence, $S_C$, of faces and edges and vertices.  
Let $v$ be the last vertex in the sequence $S_C$
 (possibly $v=s$).  
There is a unique geodesic path from $s$ to $v$, and all the paths of $C$ traverse this same path from $s$ to $v$.  After that, the points of the paths of $C$ all lie in the ruffle of $v$.  Since the paths 
traverse the same sequence of edges and faces they can be laid out in the plane to form a cone with apex $v$.   See Figure~\ref{fig:cone}.
Observe that the boundary rays of the cone may or may not lie in the set $C$.  If the boundary of the cone is the boundary of the ruffle of $v$ then it is included in $C$; but if the boundary of the cone is determined by another vertex, then beyond that vertex, the boundary is not included.
Note however, that the boundary ray is a shortest path---just not of the same combinatorial type since it goes through another vertex.  

\begin{figure}[htb]
\centering
\includegraphics[width=3in]{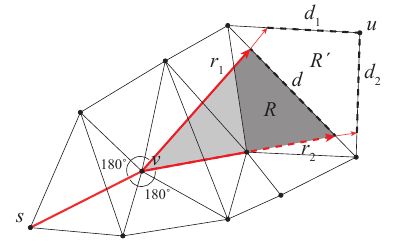}
\caption{The structure of shortest paths to one region $R$ (shown darkly shaded) of the shortest path map.
The set $C$ of shortest paths to points in the region forms a path $\sigma(s,v)$ together with a cone (lightly shaded) with apex $v$ bounded by rays $r_1$ and $r_2$.  Region $R$ is closed on the $r_1$ boundary and open on the $r_2$ boundary.  
Shortest paths exit $R$ through segment \note{$d$}.  \note{The figure shows one region $R'$ of the shortest path map beyond $d$.  Shortest paths exit $R'$ through two segments $d_1$ and $d_2$ and a vertex $u$.}
Note that the angles of the triangles incident to $v$ are not drawn accurately since they should sum to more than $2 \pi$.}
\label{fig:cone}
\end{figure}

\subsubsection{Computing the shortest path map}

We will show that if the shortest path map has $M$ regions, then it can be computed in time $O(M)$.
Regions of the shortest path map may have dimension 0, 1, or 2.
Each 2-dimensional region of the shortest path map is bounded by: 
two boundary rays;
%two segments of boundary rays (either or both of which may degenerate to a point); 
a vertex or a  
segment of an edge through which shortest paths enter the region; and one or two segments of edges and possibly a vertex through which shortest paths exit the region.  See Figure~\ref{fig:cone}.  With each region, we will store its boundary rays and vertices/segments.
Each vertex of the complex is a 0-dimensional region of the shortest path map.  An edge may form a 1-dimensional region of the shortest path map (for example any edge $(v,w)$ inside the ruffle of $v$).
%, so that we can link to all adjacent regions.  -- I don't think we need these links. AL.

The algorithm builds the regions of the shortest path map working outwards from $s$. In general, we will have a set of vertices and segments (portions of edges) that form the ``frontier'' of the known regions, and at each step of the algorithm, we will advance the known regions beyond one frontier vertex or segment.  

The algorithm is initialized as follows.  Assume that $s$ is a vertex of the complex (if necessary, by triangulating the face containing $s$ \megan{or the %two  %Anna: there may be more than two!
neighbouring faces if $s$ is on an edge}).  Each edge incident to $s$ becomes a region of the shortest path map.  Each face $f$ incident to $s$ becomes a region of the shortest path map with the two edges of $f$ that are incident to $s$ as its boundary rays.
The two vertices of $f$ different from $s$ enter the frontier, along with the edge of $f$ not incident to $s$.

At each step of the algorithm we take one vertex or segment out of the frontier set and we find all the regions for which shortest paths enter through this vertex or segment. 

Consider first the case of removing segment \note{$d$} from the frontier.  We wish to find the regions of the shortest path map for which shortest paths enter through segment $d$.
If segment $d$ lies in edge $e$, then the faces containing the new regions are those incident to $e$, not including the face from which shortest paths arrive at $d$.
 (See segment $d$ and region $R'$ in Figure~\ref{fig:cone} for example.) Each such region $R'$ gives rise to one or two segments and possibly a vertex through which shortest paths exit the region.
 We add these segments and vertex to the frontier.  In case there is a vertex, $u$,  (such as in Figure~\ref{fig:cone})  
we must find the shortest path to the vertex.    
This can be done by placing the boundary rays of $R'$ in the plane, computing their point of intersection, $p$, and constructing the ray from $p$ to $u$. Note that we do not need to know the sequence of faces traversed by shortest paths to region $R'$---local information suffices. 
This provides us with the shortest path to $u$ and also the boundary rays of the segments incident to $u$.

We next consider the case where a vertex $v$ is removed from the frontier.  
We must find the regions of the shortest path map for which shortest paths enter through vertex $v$.
These lie in the ruffle of $v$.  
Knowing the shortest path $\sigma(s,v)$, we can search the link graph $G_v$ of $v$ to find all the boundary rays of the ruffle of $v$.  
Any edge incident to $v$ that lies in the ruffle forms a 1-dimensional region of the shortest path map, and we add its other endpoint to the frontier.
For each face $f$ incident to $v$, we can identify the region of the shortest path map that lies in face $f$ and interior to the ruffle of $v$.   We can also identify the segments and vertices through which shortest paths exit the new region, and add these to the frontier. 

This completes the high-level description of the algorithm.  We spend constant time per region of the shortest path map, plus $O(n)$ time to search the faces incident to each vertex, for a total of $O(M)$.

If we want to use the shortest path map to answer shortest path queries, we also need a way to locate, given a target point $t$ that lies in face $f$, which region of the shortest path map contains $t$.
This necessitates building a search structure for the shortest path regions that face $f$ is partitioned into, which takes more time and space. (Results of Mount~\cite{Mount} might give a solution better than the obvious one for this.)  We will not pursue this solution because we will present an alternative solution in Section~\ref{sec:last-step}.

%%%%%%%%%%%%%%%%%%%%%%%%%%%%%%%%%%%%%%%%%%%
\subsubsection{Properties of the shortest path map}

For our remaining results, we need
some properties of shortest paths in a 2D CAT(0) complex.

We begin with the observation that shortest path rays diverge in any face, \Anna{i.e., if we place the face in the plane and extend the two rays backwards, they meet}.  This is obvious (see Figure~\ref{fig:cone}) for rays in one region of the shortest path map, and follows more generally from the fact that regions of the shortest path map partition any face.

\Anna{
\begin{observation}
\label{obs:paths-diverge}
Any two shortest path rays in a face diverge.
\end{observation}
}
\remove{ below is a proof of the above observation except for parallel rays
\Anna{
\begin{observation}
\label{obs:paths-diverge}
Any two shortest path rays in a face diverge.
\end{observation}
\begin{proof} Suppose shortest path rays $r$ and $s$ converge in face $f$.  If they intersect in $f$, this contradicts the uniqueness of shortest paths.  Next, suppose that rays $r$ and $s$ exit the same edge $e$ of face $f$.  We will augment the complex by adding a new face $f'$ incident to edge $e$.  The other two edges of face $f'$ can be chosen so that rays $r$ and $s$ intersect in $f'$.  This provides a contradiction so long as the new complex is CAT(0).  But observe that 
the addition of face $f'$ does not alter the link graph of any vertex, and the complex remains simply connected.  Thus, by Theorem~\ref{thm:CAT0}, the new complex is CAT(0), which gives our contradiction.

On the other hand, if  rays $r$ and $s$ exit face $f$ on different edges, then by Lemma~\ref{lemma:shortest-paths-to-face} they must enter on the same edge $e'$, and we make the same argument by adding new face $f'$ on edge $e'$. 
\end{proof}
}
}
\begin{lemma}
Let $e$ be an edge of a 2D CAT(0) complex.  Either all the shortest paths to internal points of $e$ travel along $e$, or they all reach $e$ from one incident face.
\label{lemma:shortest-paths-to-edge}
\end{lemma}
\begin{proof}
If the shortest path to some internal point \note{$p$ of edge $e$ travels along $e$ (i.e.~arrives at $p$ from one of the endpoints of $e$)}, then so do the shortest paths to all internal points of $e$. 

Otherwise \note{shortest paths to internal points of $e$ arrive from faces incident to $e$.}
Consider the (finitely many) combinatorial types of shortest paths to points of $e$, and let $C_1, C_2, \ldots, C_k$ be the corresponding sets of shortest paths, ordered according to the order of points along $e$.  We will prove that paths in all the $C_i$'s arrive at points of $e$ from the same incident face.  For otherwise, there would be some $C_i$ and $C_{i+1}$ that arrive from different incident faces.  
\note{Let $p$ be the point on $e$ that is the boundary between points reached by paths of $C_i$ and points reached by paths of  $C_{i+1}$.  Point $p$ must be reached by a ray in one of $C_i$ or $C_{i+1}$, say $C_{i+1}$.}
%The boundary ray  between $C_i$ and $C_{i+1}$, must be part of one or the other, say $C_{i+1}$. 
 But observe that when $C_i$ is laid out in the plane, the boundary ray of its cone \note{that is incident to $p$} is still a shortest path, and still arrives at $e$ from the same incident face as $C_i$ does.  But this contradicts $C_{i+1}$ arriving from a different face. 
\end{proof}

We next characterize how shortest paths can enter a face (a triangle) of the complex.  See Figure~\ref{fig:triangle-sh-path}.

\begin{lemma}
Shortest paths enter a triangular face either through
one edge, or one vertex, or one edge and an incident vertex, or two edges and their common vertex.
%from at most two edges and their common vertex. 
\label{lemma:shortest-paths-to-face} 
\end{lemma}
\begin{proof}
We cannot have shortest paths entering a face from all three edges, nor from an edge and the opposite vertex, otherwise we would have shortest paths to two points on the same edge arriving from different faces, in contradiction to Lemma~\ref{lemma:shortest-paths-to-edge}. 
\end{proof}

\begin{figure}
\centering
\includegraphics[width=6in]{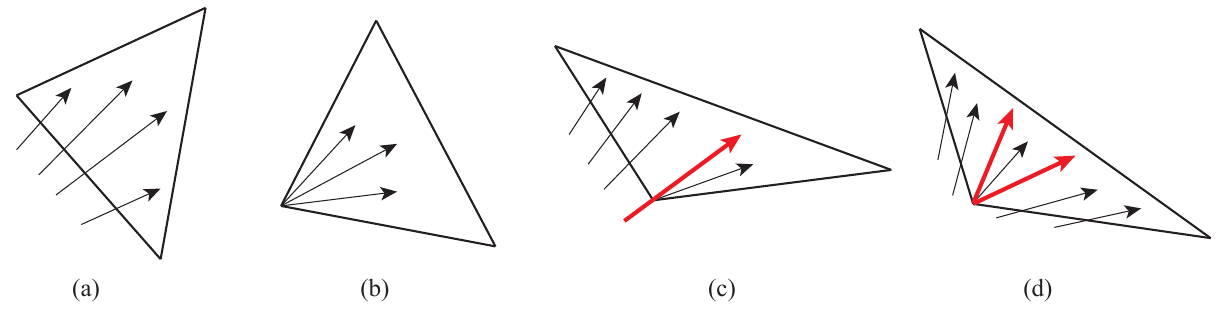}
\caption{Shortest paths may enter a face through: (a) one edge (type \emph{E}); (b) one vertex (type \emph{V}); (c) one edge and an incident vertex (type \emph{EV}); or (d) two edges and their common vertex (type \emph{EVE}).
%For type \emph{EV} and \emph{EVE} faces we store the rays (shown in bold red) that partition the face based on the edge/vertex through which shortest paths enter. 
}
\label{fig:triangle-sh-path}
\end{figure}

%%%%%%%%%%%%%%%%%%%%%%%%%%%%%%%%%%%%%%%%%%%
\subsubsection{Size of the shortest path map}

A boundary ray between adjacent regions of the shortest path map starts out as a boundary ray of the ruffle of some vertex.  By Lemma~\ref{lemma:shortest-paths-to-face}, each face originates at most two such rays.  In a general 2D CAT(0) complex, such a ray can bifurcate into two or more branches when it hits an edge that is incident to more than two faces.  There is one branch for each new incident face.
See Figure~\ref{fig:exponential}(a) for an example.  
The collection of all branches that originate from one boundary ray of a ruffle is called a \emph{boundary tree}.  Observe that it is a tree---no two branches can intersect because geodesic paths are unique.
There are $O(n)$ boundary trees because each face originates at most two boundary trees.
If the complex is a 2-manifold (i.e.,~every edge is in at most two faces) then 
no bifurcations can occur, so each boundary tree consists of only one branch, which implies that the size of the shortest path map is $O(n^2)$.  This was proved 
by Maftuleac~\cite{Maftuleac} (where 2-manifold complexes are called ``planar''), but we include a proof because we wish to observe a generalization.

\begin{lemma}[\cite{Maftuleac}]
In a 2D CAT(0) complex that is a 2-manifold the size of the shortest path map is $O(n^2)$.  
%There are examples that achieve this bound.
\label{lemma:size-shortest-path-map}
\end{lemma}
\begin{proof} As noted above, every boundary tree consists of only one branch, or ray.  If such a ray entered a face twice then the second entry would not be a shortest path, since we could short-cut across the face from the first entry.  Therefore no ray enters a face twice, and the number of boundary tree branches cutting any face is $O(n)$.  Then the number of regions of the shortest path map within one face is $O(n)$ and the overall number of regions is $O(n^2)$.
%
%Finally, for the lower bound, see Figure~\ref{fig:n2-shortest-path-map}.
%\note{Is this in \cite{Maftuleac}?}
\end{proof}

\revised{A general 2D CAT(0) complex may have the property that} 
%In a general 2D CAT(0) complex it may happen that 
no two branches of one boundary tree cross the same face, 
\revised{in which case} 
%which implies that 
the shortest path map still has size $O(n^2)$.
%and in this case the shortest path map still has size $O(n^2)$ as argued in the proof of Lemma~\ref{lemma:size-shortest-path-map}. 
We prove that this is the case for 2D CAT(0) rectangular complexes:

\begin{lemma}
In a 2D CAT(0) rectangular complex, no two branches of one boundary tree can enter the same face, and from this it follows that 
the shortest path map has size $O(n^2)$.
\label{lemma:rectangular}
\end{lemma}
\begin{proof}
Suppose, by contradiction, that two branches $r_1$ and $r_2$ of the same boundary tree enter a common face.  Let $f$ be the first face \megan{along branch $r_1$ that} they both enter.  
\note{Let $f_0$ be the face just before the two branches diverge.
The sequence of faces from $f_0$ to $f$ traversed by $r_1$ can be laid out in the plane so that $r_1$ forms a straight line.  After a suitable rotation, the edges crossed by $r_1$ alternate between horizontal and vertical.  
Then the angle between $r_1$ and any horizontal edge it crosses is the same, say $\alpha$, and the angle between $r_1$ and any vertical edge it crosses is $\pi/2 -\alpha$. 
See Figure~\ref{fig:rectangle-ray}.
The same is true for $r_2$, and the angle $\alpha$ must be the same for $r_1$ and $r_2$ because the two rays match in face $f_0$. }  
%Edges of a rectangular complex lie in two classes, which we call ``horizontal'' and ``vertical'', such that
%\note{the edges around any face alternate between horizontal and vertical.} 
%Observe that if a ray makes an angle of $\alpha$ with some horizontal edge, then it makes the same angle $\alpha$ with every horizontal edge that it crosses in the rectangular complex, and it makes the same angle $\pi/2 -\alpha$ with every vertical edge that it crosses.  
%See Figure~\ref{fig:rectangle-ray}.
This means that
$r_1$ and $r_2$ are parallel \megan{or perpendicular} in $f$. 
\Anna{If $r_1$ and $r_2$ are parallel then they do not diverge, which contradicts Observation~\ref{obs:paths-diverge}.}   
%\megan{If $r_1$ and $r_2$ are parallel then c}onsider a line segment $b$ joining a point of $r_1$ and a point of $r_2$ in $f$, and consider the shortest paths that arrive at points of $b$.  The last segments of all these shortest paths must be parallel.  This contradicts the fact that the set of shortest paths corresponding to any region of the shortest path map form a cone of rays. 
%
\megan{If $r_1$ and $r_2$ are perpendicular, then either they intersect in $f$, which is a contradiction, or there is some edge $e$ of $f$ that both rays pass through.  If $r_1$ enters $f$ at edge $e$, then by Lemma~\ref{lemma:shortest-paths-to-edge}, branch $r_2$ must reach $e$ from this same face, contradicting $f$ being the first face along branch $r_1$ that both branches enter.  Thus $r_1$ must exit $f$ through edge $e$, and again by Lemma~\ref{lemma:shortest-paths-to-edge}, $r_2$ must also exit through edge $e$.}
%Branch $r_1$ must exit $f$ at edge $e$, otherwise it would contradict $f$ being the first face along branch $r_1$ that they both enter. } 
%\Anna{Then by Lemma~\ref{lemma:shortest-paths-to-edge} branch $r_1$ must also exit $f$ at edge $e$.  
\Anna{But then the two rays converge, which contradicts Observation~\ref{obs:paths-diverge}.
}
%Following branch $r_1$ as it exits $f$, lay out the faces that both it and $r_2$ pass through.  If $r_2$ was entering $f$ at edge $e$, then $r_2$ will pass into only one new face as it crosses each edge, however, since $r_1$ is exiting $f$ at edge $e$, it will pass into all new faces as it crosses each edge.  Thus we can lay out a sequence of faces such that they all contain both branches.   Branches $r_1$ and $r_2$ will eventually intersect in one of these faces, which is a contradiction.}

Therefore the branches of one boundary tree enter a face at most once.  Since there are $O(n)$ boundary trees, this means that the 
number of boundary tree branches cutting any face is $O(n)$.  Then the number of regions of the shortest path map within one face is $O(n)$ and the overall number of regions is $O(n^2)$.
\end{proof}

\begin{figure}
\centering
\includegraphics[width=3in]{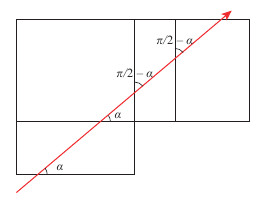}
\caption{If a ray makes an angle of $\alpha$ with some horizontal edge in a 2D CAT(0) rectangular complex, then it makes the same angle $\alpha$ with every horizontal edge that it crosses, and it makes an angle $\pi/2 -\alpha$ with every vertical edge that it crosses.
}
\label{fig:rectangle-ray}
\end{figure}

In a general 2D CAT(0) complex, two branches of one boundary tree may cross the same face---see \revised{Figure~\ref{fig:exponential}} for an example---and the size of the shortest path map may grow exponentially:

\begin{proposition} 
The size of the shortest path map of a 2D CAT(0) complex may be exponential in $n$, the number of faces.
\label{lemma:exponential}
\end{proposition}
\begin{proof}
Figure~\ref{fig:exponential}(a--c) show how one boundary ray of a ruffle can bifurcate into two branches which then enter the same face $g_2$.  Figure~\ref{fig:exponential}(d) shows how this process can be repeated.  With each addition of three faces, $g_i, f_i'$, and $f_i''$, the number of branches doubles.  Thus after adding $3n$ faces, the number of branches is $2^n$\revised{---so long as the angles are small enough that the process can be repeated $n$ times.  
To justify this, we need to be more precise about the angles.

For the initial set-up, let the angle between edge $e_1$ and the initial ray  
be $\alpha_1 = \varepsilon$, with $\varepsilon$ to be chosen later.   
Define $\beta_1$, the angle between $e_2$ and [the extension of] $e_1$ to be $\beta_1 = 2\varepsilon$.
More generally, define 
$\beta_i$, the angle between edge $e_{i+1}$ and [the extension of] $e_i$ to be $2^i \varepsilon$.
Note that in our construction the sum of the angles of $f'_i, f''_i$ and $g_{i+1}$ at the point where they meet is $2\pi$, so the angle between $e_{i+1}$ and [the extension of] $e_i$ is well-defined.

We claim that in 
the general situation, as shown in Figure~\ref{fig:exponential-2}(a), we have an edge $e_i$ and a fan of $2^i$ pairs of branches that meet in pairs along $e_i$, and form an increasing sequence of angles from $\alpha_1$ to $\alpha_i = (2^i -1)\varepsilon$.  We prove this by induction on $i$.
It is true initially with $i=1$.   For the induction step from $i$ to $i+1$, it suffices to examine the outer pair of branches, since they determine the two extreme rays intersecting $e_{i+1}$.  Refer to Figure~\ref{fig:exponential-2}(b).  
From $\alpha_i = (2^i -1)\varepsilon$ and $\beta_i = 2^i \varepsilon$, we 
%apply the fact that angles in a triangle sum to $\pi$ to 
calculate that the maximum angle between a branch and $e_{i+1}$ is $\alpha_{i+1} = (2^{i+1} -1)\varepsilon$ and the minimum angle between a branch and $e_{i+1}$ is $\alpha_1$.  The remaining branches have slopes and intersection points on $e_{i+1}$ that lie between these two extremes. 
These $2^i$ branches are then reflected in $e_{i+1}$ to form a fan of $2^{i+1}$ pairs of branches, which completes the induction proof. 

The induction step to $i+1$ can be carried out so long as  $\beta_i  + \alpha_i < \pi$.  
Thus, by choosing $\varepsilon < \pi / 2^{n}$ we guarantee $\beta_{n-1} + \alpha_{n-1} < 2^n\varepsilon < \pi$, and we can continue the branching process for $n$ steps.
We note that this construction produces an exponential-sized shortest path map only by using $O(n)$ bits for the angles.
}
\end{proof}

\begin{figure}[htb]
\centering
\includegraphics[width=6in]{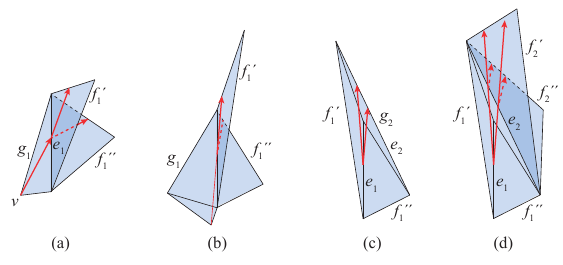}
\caption{(a) A boundary ray of a ruffle (shown in red with arrows) originates from vertex $v$ in face $g_1$, and bifurcates when it reaches edge $e_1$, branching into two rays, one in face $f_1'$ and one in face $f_1''$. 
(b) The same situation but with sharper angles.  
(c) The two resulting branches enter face $g_2$ that is incident to $f_1'$ and $f_1''$, and arrive at edge $e_2$.  Note that (b) and (c) show opposite sides of face $f_1'$.
(d) Two more faces $f_2'$ and $f_2''$ are incident to edge $e_2$, so the two branches bifurcate into a total of four branches.  In the next iteration, the four branches will enter a face $g_3$ incident to $f_2'$ and $f_2''$.  
The process can be continued, and the number of rays doubles each time we add three faces.
%, which shows that the shortest path map may have an exponential number of regions in one face.
}
\label{fig:exponential}
\end{figure}

\begin{figure}[htb]
\centering
\includegraphics[width=6in]{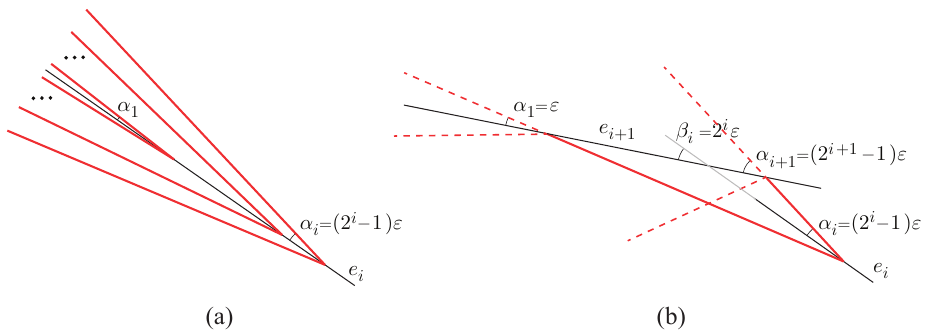}
\caption{\revised{A more detailed view of the angles used in the construction from Figure~\ref{fig:exponential}: (a) the general set-up; 
(b) the next step. 
} }
\label{fig:exponential-2}
\end{figure}

Note that an exponential size shortest path map does not preclude polynomial time algorithms for computing shortest paths.  In the tree space and its generalization, orthant space, the shortest path map, and indeed the number of regions in a face, can have exponential size \cite{Owen,MillerOwenProvan}, but there is still a polynomial time algorithm for computing geodesics in these spaces \cite{OwenProvan,MillerOwenProvan}.

%%%%%%%%%%%%%%%%%%%%%%%%%%%%%%%%%%%%%
\subsection{The Last Step Shortest Path Map}
\label{sec:last-step}

Although the shortest path map for single-source shortest paths in a 2D CAT(0) complex may have exponential size, there is a structure, called the ``last step shortest path map,'' that has linear size and can be used to find the shortest path to a queried target point in time proportional to the combinatorial size of the path (i.e.,~the number of faces, edges, and vertices traversed by the path).

The \emph{last step shortest path map}, first introduced in~\cite{touring},  partitions the space into regions where points $p$ and $q$ are in the same region if the shortest paths $\sigma(s,p)$ and $\sigma(s,q)$ have the same last vertex, edge, or face, i.e.,~the combinatorial type of the two paths matches on the last element.   
Thus, the last step shortest path map is a coarsening of the shortest path map.
\note{By Lemma~\ref{lemma:shortest-paths-to-face}  each face 
has one of the types shown in Figure~\ref{fig:triangle-sh-path}, and is partitioned into one, two, or three regions. 
We store the type of each face, and for 
type \emph{EV} and \emph{EVE} faces we store the rays that partition the face based on the edge/vertex through which shortest paths enter.
}

For the purpose of answering shortest path queries, we store with each region of the last step shortest path map the last vertex, edge, or face with which shortest paths enter the region.  We call this the \emph{incoming information} (``in-info'') for the region.
\note{By Lemmas~\ref{lemma:shortest-paths-to-edge} and~\ref{lemma:shortest-paths-to-face} the possible regions and possible in-info are as follows:
\begin{itemize}
\item a vertex $v$, with in-info a vertex $u$ (via edge $(u,v)$), or a face $f$
\item an edge $e$, with in-info an endpoint $u$ of $e$, or a face $f$
\item a face, partitioned into one, two, or three regions, each with in-info a vertex $u$ or an edge $e$
\end{itemize}
} 

\note{For any 2D CAT(0) complex the last step shortest path map has size $O(n)$.  The incoming information also has size $O(n)$.}

%\note{By Lemma~\ref{lemma:shortest-paths-to-edge} each edge is a region of the last step shortest path map.  Each vertex is a region by itself.  
%By  Lemma~\ref{lemma:shortest-paths-to-face} each face is partitioned into one, two, or three regions of the last step shortest path map. With each face we store its partition into regions.}
%Thus for any 2D CAT(0) complex the last step shortest path map has size $O(n)$. 

%%%%%%%%%%%%%%%%%%%%%%%%%%%%%%%%%%%%%%%%%%%
\subsubsection{Answering shortest path queries using the last step shortest path map}

We show that the last step shortest path map, together with the in-info described above, is sufficient to recover the path from $s$ to any point $t$ in time proportional to the number of faces \megan{and edges} on the path. 
A query point $t$ is given as a vertex, or a point on an edge, or a point (in local coordinates) in a face.
\note{We find the path working backwards from $t$.}

\note{If $t$ is a vertex or a point on an edge and the in-info is a face, then we treat $t$ as a point in the face.  For a point in a face, we test the partition of the face to determine which region contains $t$.  

If the in-info for $t$'s region is a vertex $u$ then we replace $t$ by $u$ and recurse. }
 
\remove{
\note{We first identify the region of the last step shortest path map that contains $t$.  In particular, if $t$ is a vertex or a point on an edge, then that vertex or edge is the region, and if $t$ is a point in a face, then we test the partition of the face to determine which region contains $t$.
Next we examine the incoming information associated with the region that contains $t$.  If the incoming information tells us that a shortest path reaches $t$'s region from a vertex $u$ or along an edge from $u$, then we replace $t$ by $u$ and recurse.}
%We first test whether a shortest path reaches $t$ along an edge---this happens if and only if $t$ is a vertex or a point on an edge and the incoming information attached to the vertex or edge is an edge, say the edge from $u$ to $v$.   In this case, we replace $t$ by $u$ and recurse. 
Otherwise, 
%$t$ lies in a face, edge or vertex and 
the incoming information tells us that a shortest path reaches $t$'s region through a face, say $f$.  Refer to Figure~\ref{fig:triangle-sh-path}.
If $f$ is of type \emph{V}, we replace $t$ by the incoming vertex of $f$ and recurse.  If $f$ is of type \emph{VE} or type \emph{EVE} we locate $t$ relative to the rays in $f$.  From this we can tell if the shortest path to $t$ goes through a vertex of $f$ or not.  If it does, then we replace $t$ by that vertex and recurse.
}

\begin{figure}[htb]
\centering
\includegraphics[width=2.5in]{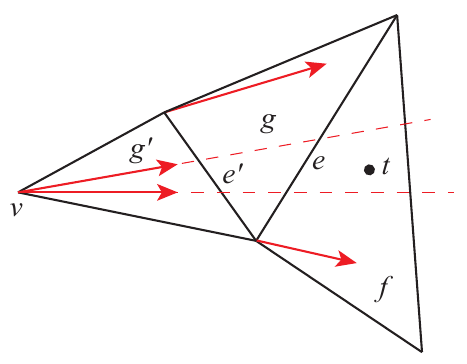}
\caption{Finding the shortest path from $s$ to point $t$ in face $f$.  In this example, $f$ is of type \emph{VE}.  Testing the ray of $f$, we find that the shortest path to $t$ enters from edge $e$ which has incoming face $g$ of type \emph{VE}.  Testing the rays of $g$, we find that the shortest path to $t$ enters from edge $e'$ which has incoming face $g'$ of type \emph{EVE}.  Finally, testing the rays of $g'$ we find that the shortest path to $t$ comes from vertex $v$.  We recursively find the shortest path to $v$.}
\label{fig:recover-path}
\end{figure}

%We are left with the case where the shortest path to $t$ enters face $f$ through some edge, say edge $e$.
Otherwise, $t$'s region is part of a face $f$, and in-info is an edge $e$.
\note{We place $f$ in the plane (arbitrarily) and enter the main loop} of the algorithm (see Figure~\ref{fig:recover-path}):
Let $g$ be the incoming face for edge $e$.  Attach triangle $g$ on the other side of edge $e$ of face $f$ in the plane.  Note that the placement of $g$ is uniquely determined.
%Now we enter the main loop of the algorithm (see Figure~\ref{fig:recover-path}):
If $g$ is of type \emph{V}, we replace $t$ by the incoming vertex of $g$ and recurse.  If $g$ is of type \emph{VE} or type \emph{EVE} we locate $t$ relative to the rays that partition $g$ (although $t$ is not in $g$ we just extend the rays to do the test).  From this we can tell if the shortest path to $t$ goes through a vertex of $g$ or not.  If it does, then we replace $t$ by that vertex and recurse.  Otherwise the shortest path to $t$ enters $g$ through an edge, and we repeat the loop with the incoming face of that edge.

This algorithm finds the shortest path from $s$ to $t$ in time proportional to the number of triangles and edges on the path.  In the worst case this is $O(n)$.
 
%%%%%%%%%%%%%%%%%%%%%%%%%%%%%%%%%%%%%%%%%%%
\subsubsection{Computing the last step shortest path map}

We do not know how to compute the last step shortest path map in polynomial time.  
More broadly, we do not know of a polynomial-time algorithm to compute shortest paths in a 2D CAT(0) complex.  
On the other hand, the problem does not seem to be amenable to NP-hardness proofs like the ones for shortest paths in 3D Euclidean space with polyhedral obstacles~\cite{canny}, or for shortest paths that visit a sequence of non-convex polygons in the plane~\cite{touring}. Furthermore, we have the example of orthant spaces as CAT(0) complexes with exponential shortest path maps, but a polynomial time algorithm for computing shortest paths \cite{MillerOwenProvan}.

It is tempting to think
%\footnote{As we did in a preliminary version of this paper.} 
that the last step shortest path map can be computed in a straight-forward way by propagating incoming information outward from the source.  The trouble with this approach is that faces of type \emph{EVE} need incoming information from two edges.  This can result in dependencies that form a cycle, with each edge/face waiting for incoming information from some other face/edge.  See Figure~\ref{fig:incoming-cycle} for an example.

\begin{figure}[htb]
\centering
\includegraphics[width=3in]{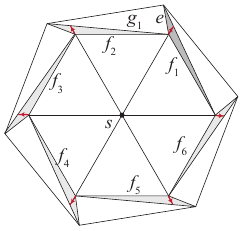}
\caption{A cycle of incoming information in a 2D CAT(0) complex that lies in the plane.  Face $f_1$ (darkly shaded) is of type \emph{EVE} with incoming edge $e$, which has incoming face $g_1$, which depends on incoming information from face $f_2$.  Similarly, each face $f_i$ (lightly shaded) depends on incoming information from face $f_{i+1}$, and $f_6$ depends on incoming information from face $f_1$, which creates a cycle.  \Anna{The red arrows indicate the ruffles of the vertices.}}
\label{fig:incoming-cycle}
\end{figure}

We end this section with one positive, though weak, result.
The last step shortest path map can be computed from the shortest path map in time $O(M)$, where $M$ is the size of the shortest path map. 
For each edge, we 
can identify the incoming edge or face from any of the shortest path regions containing portions of the edge (by Lemma~\ref{lemma:shortest-paths-to-edge} these all give the same information).
Since we have the shortest path to each vertex $v$, we can recover or recompute the boundary rays of the ruffle of $v$, which gives us the type (\emph{E}, \emph{V}, \emph{EV}, or \emph{EVE}) of each face incident to $v$, and the incoming information for the face. 

We summarize the implications for special cases of the single-source shortest path problem in 2D CAT(0) complexes:

\begin{proposition}
For a 2D CAT(0) complex  that is a 2-manifold  or is rectangular, we can solve the single-source shortest path problem using $O(n^2)$ time and space to produce a structure (the last step shortest path map) of size $O(n)$ that allows us to answer shortest path queries in time proportional to the number of triangles and edges on the path. 
\label{lemma:good-single-course}
\end{proposition}

\remove{
In the remainder of this section we will give a polynomial-time algorithm to compute the last step shortest path map  for the special case when there are no faces of type \emph{EVE}.  Our algorithm can test this property.  

We begin by noting that rectangular complexes  \note{2D CAT(0) rectangular?} 
fit into this special case.

\begin{lemma}  In a triangulation of a rectangular complex, there are no \emph{EVE} faces.
\note{This is FALSE.}
\label{lemma:rectangular-EVE}
\end{lemma}
\begin{proof}  Consider the link graph of any vertex $v$ in a triangulation of a rectangular complex.  The length of any cycle in the link graph is a multiple of $90^\circ$.  
This means that the angle between two rays of a ruffle of $v$ is either $0$, or is greater than $90^\circ$.      
\note{Well, the $0$ causes trouble.}  
\end{proof}

\note{Give algorithm to construct last step shortest path map when there are no EVE faces.}

\note{Give example to show that the algorithm does not work in general.}
}

%%%%%%%%%%%%%%%%%%%%%%%%%%%%%%%%%%%%
%%%%%%%%%%%%%%%%%%%%%%%%%%%%%%%%%%%%

\remove{
\subsection{\note{The old stuff.}}

This section is about the single-source shortest path problem in a 2D CAT(0) complex.
The input is a 2D CAT(0) complex, $\cal K$,  composed of $n$ triangles, and a point $s$ in $\cal K$.
We denote the shortest path from $s$ to $t$ by $\sigma(s,t)$.

\begin{theorem}
There is a 
%We will build a 
data structure and algorithm that will
%that allows us to 
find $\sigma(s,t)$ for any query point $t$.
%, for any query point $t$, the shortest path from $s$ to $t$, denoted $\sigma(s,t)$.  
%We will achieve $O(n^2)$ preprocessing time, $O(n)$ storage, and query time proportional to the number of triangles and edges traversed by the path.
The data structure can be built in $O(n^2)$ preprocessing time and requires $O(n)$ storage.  The query time is proportional to the number of triangles and edges traversed by the path.
\end{theorem}

Typically in a shortest path problem, the difficulty is to decide which of multiple geodesic (or locally shortest) paths to the destination is shortest.
This is the case, for example, for shortest paths in a planar polygon with holes, or for shortest paths on a terrain, and is a reason to use a Dijkstra-like approach that explores
paths to all target points in order of distance.   For shortest paths on a terrain, Chen and Han~\cite{Chen-Han} provided an alternative that uses a Breadth-First-Search (BFS) combined with a clever pruning when two paths reach the same target point.

When geodesic paths are unique, however, it is enough to explore all geodesic paths, and there is no need to explore paths in order of distance or in BFS order.  This is the case, for example, for shortest paths in a polygon, where the ``funnel'' algorithm~\cite{Guibas-sh-path-87,Hershberger-Snoeyink} achieves $O(n)$ processing time and storage, and $O(\log n)$ query time (plus output size to produce the actual path).
\changed{Similarly, in CAT(0) spaces, the uniqueness of geodesic paths means we can simply explore all geodesic paths without any ordering constraints.

A basic approach to the single source shortest path problem is to compute the whole \emph{shortest path map} from $s$.  The \emph{shortest path map} partitions the space into regions in which all points have shortest paths from $s$ that have the same \emph{combinatorial type}.  Specialized to 2D CAT(0) complexes, two shortest paths have the same \emph{combinatorial type} if they cross the same sequence of edges, vertices, and faces.

Consider a set $C$ of shortest paths of the same combinatorial type in our 2D CAT(0) complex.  If $C$ consists of just one path, we call it a \emph{singleton} and otherwise we call it \emph{full}.  The paths in $C$ all go through the same sequence of faces \note{and edges and vertices}, so we can lay out these faces in the plane.  \note{See Figure X.}  Then a singleton $C$ becomes a straight line segment, and if $C$ is full then $C$ is either a cone with apex $s$, or $C$ consists of a shortest path (a line segment in the plane) from $s$ to some vertex $v$, plus a cone with apex $v$.
The two boundary rays of the cone still correspond to shortest paths, although they do not belong to $C$ since they have a different combinatorial type, going through a vertex rather than through an internal point of an edge. 
When $C$ contains a cone with apex $v$, the angle between the incoming segment at $v$ and any outgoing segments at $v$ is at least $\pi$ (otherwise the segments would not form a shortest path through $v$).
}

\changed{
\begin{claim}
Consider an edge $e$ of a 2D CAT(0) complex, and consider the shortest paths from $s$ to all internal points of $e$.  Either all these shortest paths travel along $e$, or they all reach $e$ from one incident face.
\label{claim:shortest-paths}
\end{claim}
\begin{proof}
If the shortest path to some internal point of edge $e$ travels along $e$, then so do the shortest paths to all internal points of $e$. 

Otherwise consider the (finitely many) combinatorial types of shortest paths to points of $e$, and let $C_1, C_2, \ldots, C_k$ be the corresponding sets of shortest paths, ordered according to the order of points along $e$.  We will prove that paths in all the $C_i$'s arrive at points of $e$ from the same incident face.  For otherwise, there would be some $C_i$ and $C_{i+1}$ that arrive from different incident faces.  At least one of the two, say $C_i$, must be full.  Now observe that when $C_i$ is laid out in the plane, the boundary ray of its cone on the $C_{i+1}$ side is still a shortest path and still arrives at $e$ from the same incident face.  But this contradicts $C_{i+1}$ arriving from a different face. 
\end{proof}

A consequence of this claim is that shortest paths enter a face (a triangle) from at most two edges and their common vertex.  See Figure~\ref{fig:triangle-sh-path}.  In particular, we cannot have shortest paths entering a face from all three edges, nor from an edge and the opposite vertex, otherwise we would have shortest paths to two points on the same edge arriving from different faces.     
}

%In a CAT(0) complex, geodesic paths are unique.  One approach to the single source shortest path problem
%is to compute the whole shortest path map from $s$, partitioning each face into regions where all points have the same combinatorial shortest path.  

\note{The following is very FALSE.}
\changed{It is not difficult to show that for a 2D CAT(0) complex,} 
%In a 2D CAT(0) complex 
the shortest path map has worst case size $\Theta(n^2)$, and can be computed with $O(n^2)$ time and storage, so that the shortest path to query point $t$ can be found in time proportional to the number of faces along the path.

We do not see how to improve this to linear preprocessing time %as in the funnel algorithm, 
\changed{as for shortest paths in a simple polygon},
nor how to find the distance to a query point in logarithmic time, but we will improve the storage.
Our algorithm will achieve $O(n^2)$ preprocessing time, and $O(n)$ storage, and we will recover a shortest path to query point $t$ in time proportional to the number of faces along the path.

%\smallnote{Move to Intro: Result of Chepoi and Maftuleac~\cite{Chepoi-Maftuleac}. Specialized to rectangular complexes.  $O(n^2)$ storage.  Preprocessing time? (they don't say).  Two-point queries.  Also \cite{Maftuleac}.}

The high-level idea is as follows.
Starting with the faces containing $s$, we expand to adjacent faces, constructing the
\emph{last-step shortest path map} in which
two points $p$ and $q$ inside a face are \emph{equivalent} if
$\sigma(s,p)$ and $\sigma(s,q)$
%the shortest path from $s$ to $p$ and the shortest path from $s$ to $q$
enter the face on the same edge/vertex.
We thus store a constant amount of information for each face.  See Figure~\ref{fig:triangle-sh-path}.
We show that this information is sufficient to recover the path from $s$ to any point $t$ in time proportional to the number of faces on the path.   This involves ``unfolding''  the faces along the path into the plane.
%Starting with the triangle(s) containing $s$, we expand outwards to adjacent triangles.  (We do not even need to specialize to %depth-first or breadth-first search.)
The idea of storing in each face only the combinatorial information about the last step of the shortest path comes from~\cite{touring}.  %(Also implicit in Chen and Han, as part of contracting long paths in the tree.)

We now fill in the details of our algorithm.

The algorithm will categorize faces according to how shortest paths enter them.
\changed{As justified above,} 
shortest paths may enter a face through: one edge (type \emph{E}); one vertex (type \emph{V}); one edge and an incident vertex (type \emph{EV}); or two edges and their common vertex (type \emph{EVE}).
%There are no other possibilities otherwise there would be a point in the face reached by more than one shortest path.   
See Figure~\ref{fig:triangle-sh-path}.
\changed{By Claim~\ref{claim:shortest-paths}, }
shortest paths may reach an edge from an incident face, or from an endpoint of the edge when the edge itself lies on a shortest path.
%Note that different points on an edge $e$ cannot be reached by shortest paths through different 
%faces incident to $e$\changed{. In particular, the set of points on edge $e$ that are reached by geodesic paths from face $f$ is a closed set, so if $e$ contains two points $a$ and $b$ reached by geodesics from different faces, then there would be a point $c$ on edge $e$ between them reached by two different geodesic paths.}

The algorithm will discover how shortest paths enter each face, edge and vertex.
For each face we \changed{will find and} store the \emph{incoming} edge(s)/vertex through which shortest paths enter the face.
For faces of types \emph{EV} or \emph{EVE} we store the one or two rays that form the boundaries of the part of the face reached by shortest paths through the vertex, as shown in Figure~\ref{fig:triangle-sh-path}.
For each edge we \changed{find and} store the \emph{incoming} face/vertex through which shortest paths reach the edge.
For each vertex we \changed{find and} store the \emph{incoming} face or edge that contains the last segment of the shortest path to the vertex, and in the case of an incoming face we store the last segment of the shortest path to the vertex.   In general, a ray or segment is given in local coordinates of the face in which it lies (i.e.,~in terms of vertices of the face).
Note that the \emph{incoming} information has constant size per face/edge/vertex, and therefore linear size overall.

At the beginning of the algorithm every face/edge/vertex is unmarked.  
\changed{When we have complete information about shortest paths to all points interior to a face then the face is marked \emph{explored}. 
Vertices and edges will have two possible markings---an edge or vertex marked \emph{frontier} is one that we know the shortest paths to (via \emph{incoming} information), but have not explored beyond.  Once we explore shortest paths leaving a frontier vertex or edge, then the vertex or edge will be marked \emph{explored}.  
The general step is to take an edge or vertex out of the frontier and ``explore'' beyond it, moving some incident faces/edge/vertices out of the unmarked category into the frontier or explored category.
The algorithm terminates when the frontier is empty.  We will prove that at this point, all vertices, edges, and faces will have been explored, i.e.~that we have shortest paths to all points.
}

\remove{
When we have complete information about shortest paths to a face/edge/vertex then it is marked \emph{explored}.  We will have a third category for edges and vertices---an edge or vertex marked \emph{frontier} is one that we know the shortest path to (via \emph{incoming} information), but have not explored beyond.
\changed{Note that faces are only marked as \emph{explored} in the process of ``exploring" edges and vertices in the frontier.}
The general step is to take an edge or vertex out of the frontier and ``explore'' beyond it, moving some incident faces/edge/vertices out of the unmarked category into the frontier or explored category.
The algorithm terminates when every face/edge/vertex is marked \emph{explored}.
}

%\remove{
%\note{Write this better.}
%The goal of our algorithm is to \emph{explore} all the faces, edges, and vertices.  We store the type of each explored face.
%For each face we store the \emph{incoming} edge(s)/vertex through which shortest paths enter the face.
%For faces of types \emph{EV} or \emph{EVE} we store the one or two rays that form the boundaries of the part of the face reached by shortest paths through the vertex.
%%For a Type 1 face we store the \emph{incoming} edge or vertex through which shortest paths reach the face.
%%For a Type 2 face we store the \emph{incoming} edge(s) and vertex through which shortest paths reach the face, and we store the one or two rays that form the boundaries of the part of the face reached by shortest paths through the vertex.
%For each edge we store the \emph{incoming} face/vertex through which shortest paths reach the edge.
%%For a Type 1 or Type 2 edge we store the \emph{incoming} face through which shortest paths reach the edge (as noted above, this face is unique), and for a Type 0 edge we store the \emph{incoming} vertex through which shortest paths reach points along the edge.
%For each vertex we store the ray on which the shortest path enters $v$ and the \emph{incoming} face or edge that contains this ray.
%
%In general, we have a \emph{frontier} of edges and vertices that we know the shortest path to (via \emph{incoming} information), but have not explored beyond.
%}

We now give the details of the algorithm to build the data structure for shortest path queries.
The algorithm to answer shortest path queries is described later on.  The two methods are entwined, because we need to answer shortest path queries in order to build the data structure.

\medskip\noindent
{\bf Initialization.}
Assume that $s$ is a vertex of the complex (if necessary, by triangulating the face containing $s$).
For each edge $e =(s,v)$ incident to $s$, mark $e$ as ``explored'' with incoming vertex $s$ and put $v$ into the frontier with entering ray $sv$.
For each face $f$ incident to $s$,
mark $f$ as ``explored'', and as type \emph{V} with incoming vertex $s$ and put the edge of $f$ not incident to $s$  into the frontier with incoming face $f$.

\medskip\noindent
{\bf General Step.}
Until the frontier is empty, take a vertex or edge out of  the frontier and explore beyond it as specified in the following cases.
Before we take an edge out of the frontier, there are special conditions that must be met,
which we describe below.

\medskip\noindent
{\bf I. Taking a vertex $v$ out of the frontier.}  Mark $v$ ``explored''.  Let $r$ be the incoming ray to $v$.  Starting from point $r$ in $v$'s link graph $G_v$, we search the link graph to identify all points within distance $\pi$ from $r$.  The complementary set (all points in $G_v$ of distance $\ge \pi$ from $r$) correspond to points in  $\cal K$ that have shortest paths that go through $v$, and we call this set the \emph{ruffle} of $v$.
Note that this includes the case where $v$ is on the boundary.
See Figure~\ref{fig:ruffle}.
For each edge $e=(v,u)$ incident to $v$, if $e$ is in the ruffle then we mark $e$ ``explored'' with incoming vertex $v$ and we put $u$ in the frontier with incoming ray $vu$.

\noindent
For each face $f$ incident to $v$ we consider several cases depending on how $f$ intersects the ruffle of $v$.
Let $e=(a,b)$ be the edge of $f$ not incident to $v$.

\noindent
{\bf Case 0.} No point of $e$ is inside the ruffle.  Do nothing, as no shortest paths to this face pass through $v$.

\noindent
{\bf Case 1.}
Both $a$ and $b$ are inside the ruffle.
Note that all of $f$ is in the ruffle because $f$ corresponds in the link graph $G_v$ to an edge whose endpoints (corresponding to $a$ and $b$) are distance $\ge \pi$ from $r$, so all points internal to the edge are also distance $\ge \pi$ from $r$.
Mark $f$ ``explored'' of type $V$ with incoming vertex $v$, and put $e$ in the frontier with incoming face $f$.

\noindent
{\bf Case 2.}
Exactly one of $a$ or $b$ (say $a$) is inside the ruffle.  Mark $f$ of type \emph{VE} with incoming vertex $v$ and incoming edge $(v,b)$.
A boundary ray of the ruffle goes from $v$ to a point on the edge $(a,b)$.  We store this ray with the face $f$.
Note that we do not yet mark $f$ as ``explored''---we will only do that after exploring edge $(v,b)$.

\noindent
{\bf Case 3.} Neither $a$ nor $b$ is inside the ruffle but some interior point(s) of $e$ are in the ruffle.  Mark $f$ of type \emph{EVE} with incoming vertex $v$ and incoming edges  $(v,b)$ and $(v,a)$.
If a single point of $e$ is inside the ruffle then exactly one ray goes from $v$ to a point on the edge $e$, and otherwise
two boundary rays of the ruffle go from $v$ to points on the edge $e$.  We store these rays with the face $f$.
Note that we do not yet mark $f$ as ``explored''---we will only do that after exploring edges $(v,b)$ and $(v,a)$.

%\noindent
%Otherwise we do nothing.

\medskip\noindent
{\bf II. Taking an edge $e$ out of the frontier.}
We only take an edge $e=(u,v)$ out of the frontier if both vertices $u$ and $v$ have already been explored.
%(Observe that an edge only enters the frontier when both endpoints are in the frontier or already explored, so we can always remove something from the frontier unless it is empty.)
Mark $e$ as ``explored''.  Let $g$ be the incoming face for $e$.
For each face $f \ne g$ incident to $e$, let $w$ be the third vertex of $f$ and do the following.

\noindent
{\bf Case 1.}
If $f$ is not already marked \emph{VE} or \emph{EVE} then we mark $f$ ``explored'' of type $E$ with incoming edge $e$.
  We put the edges $(u,w)$ and $(v,w)$ in the frontier with incoming face $f$.
Using the method described below, we \changed{query to} find the shortest path from $s$ to $w$ and the segment $r$ along which the shortest path reaches $w$. We put vertex $w$ in the frontier with entering segment $r$ and incoming face $f$.

\noindent
{\bf Case 2.}
If $f$ is marked \emph{VE} with, say, $v$ as an incoming vertex then we mark $f$ ``explored'' and put edge $(u,w)$ in the frontier with incoming face $f$.  (The case when $u$ is the incoming vertex is symmetric.)

\noindent
{\bf Case 3.}
If $f$ is marked \emph{EVE} with, say, $v$ as an incoming vertex then if edge $(v,w)$ is already explored we mark $f$ ``explored'' and put edge $(u,w)$ in the frontier with incoming face $f$.
(The case when $u$ is the incoming vertex is symmetric.)

\medskip

\changed{
Note that no problems are caused by our condition that an edge is removed from the frontier only after both its vertices have been explored.   Since an edge only enters the frontier when both endpoints are in the frontier or already explored, we can always remove something from the frontier unless it is empty. 

To show that the algorithm is correct, we must show that when the frontier is empty, all vertices, edges, and faces have been marked explored and that when a vertex/edge/face is marked explored, we have the correct incoming information. 

\note{fill in more} 

%When the frontier is empty, then we claim that every \emph{element} (i.e.~every vertex/edge/face) is marked explored, and that we have the correct incoming information.  We prove this by induction on the combinatorial length of the shortest path to the element, i.e.~the number of vertices, edges, and faces crossed by the shortest path to the element.  The base case, when the combinatorial distance is 1, is handled correctly by the Initialization step of the algorithm.
%So consider an element $x$ at combinatorial distance $k$.  Let $y$ be the element before $x$ along the shortest path to $x$. 

The ``incoming'' relationship defines a directed graph $G$ whose nodes are the vertices, edges, and faces of the complex $K$, with a directed edge from node $x$ to node $y$ if there is a geodesic path in $K$ that travels through a point of $x$ and immediately after that a point of $y$.  There is a path in $G$ from $s$ to every node, and $s$ has only outgoing edges.  We claim that $G$ is acyclic.  A node corresponding to a vertex has only one incoming edge in $G$ (because there is a unique geodesic path to the vertex), and a node corresponding to an edge of $K$ has only one incoming edge in $G$ by Claim~\ref{claim:shortest-paths}.  So if there is a cycle in $G$ then there must be some face $F$ that . . .

}

\remove{
Correctness of the algorithm is straightforward by induction on the number of faces.  The one thing worth commenting on is our condition about not removing an edge from the frontier until both its vertices have been explored.  Since an edge only enters the frontier when both endpoints are in the frontier or already explored, we can always remove something from the frontier unless it is empty.  When the frontier is empty, all faces, edges and vertices will be explored.
}

We now analyze the run time of the algorithm.  Each edge/vertex enters the frontier only once.  The time to process a vertex (step I) is proportional to the number of incident edges and faces, so this is linear overall.  The time to process an edge (step II) is proportional to the number of incident faces times the time to recover a shortest path to a point (point $w$ in case 1).
% Perhaps we can save a bit by processing all such $w$'s in one batch, but it's not clear, and would require more work.
As shown below, it takes $O(n)$ time to recover a shortest path.  Thus the total run time is $O(n^2)$.  Storage is $O(n)$.

\medskip
We now describe how to answer a query for the shortest path to a point $t$.
%To complete the algorithm, we must show how to recover the shortest path to a point $t$.
We can do this as soon as the face/edge/vertex containing $t$ has been explored---the algorithm need not have terminated.
If $t$ is a vertex, or a point on an edge, we
can tell from the incoming information if a shortest path reaches $t$ along an edge, say the edge from $u$ to $v$.  In this case, we replace $t$ by $u$ and recurse.

In the more general case $t$ lies in a face, edge or vertex and we know that a shortest path reaches $t$ through a face, say $f$.
%suppose
%$t$ is a point in the interior of face $f$, or $t$ is a point of an edge with incoming face $f$, or $t$ is a vertex with incoming face $f$.
If $f$ is of type \emph{V}, we replace $t$ by the incoming vertex of $f$ and recurse.  If $f$ is of type \emph{VE} or type \emph{EVE} we locate $t$ relative to the rays in $f$.  From this we can tell if the shortest path to $t$ goes through a vertex of $f$ or not.  If it does, then we replace $t$ by that vertex and recurse.

\remove{ % moved earlier
\begin{figure}[htb]
\centering
\includegraphics[width=2.5in]{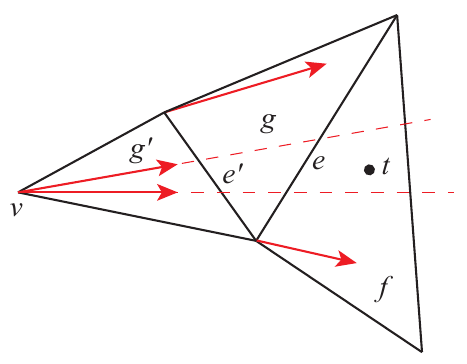}
\caption{Finding the shortest path from $s$ to point $t$ in face $f$.  In this example, $f$ is of type \emph{VE}.  Testing the ray of $f$, we find that the shortest path to $t$ enters from edge $e$ which has incoming face $g$ of type \emph{VE}.  Testing the rays of $g$, we find that the shortest path to $t$ enters from edge $e'$ which has incoming face $g'$ of type \emph{EVE}.  Finally, testing the rays of $g'$ we find that the shortest path to $t$ comes from vertex $v$.  We recursively find the shortest path to $v$.}
\label{fig:recover-path}
\end{figure}
}

We are left with the case where the shortest path to $t$ enters face $f$ through some edge, say edge $e$.
Let $g$ be the incoming face for edge $e$.  We place $f$ in the plane and attach triangle $g$ to edge $e$.  The placement of $f$ is arbitrary, but then $t$ and $g$ are fixed.
Now we enter the main loop of the algorithm (see Figure~\ref{fig:recover-path}):
If $g$ is of type \emph{V}, we replace $t$ by the incoming vertex of $g$ and recurse.  If $g$ is of type \emph{VE} or type \emph{EVE} we locate $t$ relative to the rays in $g$ (although $t$ is not in $g$ we just extend the rays to do the test).  From this we can tell if the shortest path to $t$ goes through a vertex of $g$ or not.  If it does, then we replace $t$ by that vertex and recurse.  Otherwise the shortest path to $t$ enters $g$ through an edge, and we repeat with the incoming face of that edge.

This algorithm will find the shortest path from $s$ to $t$ in time proportional to the number of triangles and edges on the path, which is $O(n)$ in the worst case.

To wrap up, the whole algorithm takes time $O(n^2)$ to preprocess the complex from $s$, and results in a structure of space $O(n)$ that
allows searching for a path to $t$ in time proportional to the size of the path.

It is possible that results of Mount~\cite{Mount} might provide an algorithm that uses $O(n^2)$ time, $O(n \log n)$ space, and answers queries for the distance in $O(\log n)$ time\footnote{Thanks to Stefan Langerman for this suggestion}.  The rough idea is to store the whole shortest path map via nested trees that allow us to search all the rays entering a triangle.  

%\note{This is a suggestion from Stefan Langerman.  There is a paper by David Mount,``Storing the Subdivision of 
%a Polyhedral Surface"~\cite{Mount} that we might be able to use to give an algorithm that uses $O(n^2)$ time, $O(n \log n)$ space (the space goes up), and answers queries for the distance in $O(\log n)$ time.  The rough idea is to ``keep'' the whole shortest path map via nested trees that allow us to search all the rays entering a triangle.}

%\smallnote{Perhaps you can cut down on the work in practice by doing rough computation of shortest path in dual graph of triangles.   This is basically what Chepoi and Maftuleac do in the case of rectangles.  In their case, the shortest path will be one of the shortest paths in the dual, but that is not true in our case.}

}

%%%%%%%%%%%%%%%%%%%%%%%%%%%%%%%%%%%%%%%%%%%%%%%%%%%%%%%%
%%%%%%%  CONCLUSIONS
%%%%%%%%%%%%%%%%%%%%%%%%%%%%%%%%%%%%%%%%%%%%%%%%%%%%%%%%
\section{Conclusions}
\label{sec:conclusions}

We have given  
an algorithm for computing the \revised{closure of the} convex hull of a set of points in a 2D CAT(0) polyhedral complex with a single vertex.  Our algorithm relies on linear programming.  
\revised{The main open questions are:
%We leave it as an open question to
%find a combinatorial algorithm, or prove that a simple iterative approach takes polynomial time (Conjecture~\ref{conj:iteration}). 
\begin{itemize}
\item Is there a polynomial-time combinatorial algorithm to compute the convex hull of a set of points in a 2D CAT(0) polyhedral complex with a single vertex?

\item Is such a convex hull closed? We conjecture that it is.

\item Does the simple iterative approach of computing successive skeletons $S_k$ run in polynomial time?  I.e., does $S_k = S_{k+1}$ for some $k$ that is polynomially bounded in the size of the complex (Conjecture~\ref{conj:iteration})? 

\item   Is there a polynomial time algorithm to test if given point is in the convex hull of a given point set in a CAT(0) polyhedral complex?  This may be easier than computing the convex hull, and 
would be sufficient for most applications, including %an algorithm that computes 
computing a geometric centre by peeling convex hulls. 

\item Does our linear programming solution  extend to 2D CAT(0) complexes with more than one vertex, or to
single-vertex higher dimensional CAT(0) complexes?  The latter problem seems hard because in 3D and beyond 
the boundary between two maximal cells has dimension at least 2,  and it is not clear that the intersection of the convex hull with a boundary face is even a polytope.
% 
%We do not see how to extend our linear programming solution to 2D CAT(0) complexes with more than one vertex, nor to single-vertex higher dimensional CAT(0) complexes.
%In 2D, the boundary between maximal cells is an edge, and the  convex hull intersects the edge in an interval.
%In higher dimensions, the boundary between two maximal cells has dimension 2 or greater,  and it is not clear that the intersection of the convex hull with a boundary face is even a polytope. 
\end{itemize}  
}

%Rather than computing the convex hull explicitly,  
%it might be easier to find an algorithm that tests whether a point is in the convex hull or not.
%  This would be sufficient for most applications, including 
%computing a geometric centre by peeling convex hulls. 

For the single-source shortest path problem in a 2D CAT(0) complex, we have shown that the shortest path map may have exponential size, and that the last step shortest path map is a better alternative.
\revised{The main open questions are:
%It is an open question to compute the last step shortest path map in polynomial time.  Alternatively, one might try to prove that the shortest path problem is NP-hard for 2D CAT(0) complexes.
\begin{itemize}
\item Can the last step shortest path map be computed in polynomial time?
\item Is the shortest path problem NP-hard for 2D CAT(0) complexes?
\end{itemize}
}

%%%%%%%%%%%%%%%%%%%%%%%%%%%%%%%%%%%%%%%%%%%%%%%%%%%%%%%%
%%%%%%%  ACKNOWLEDGEMENTS
%%%%%%%%%%%%%%%%%%%%%%%%%%%%%%%%%%%%%%%%%%%%%%%%%%%%%%%%
\section{Acknowledgements}
The authors thank Sean Skwerer for the example showing convex hulls of 3 points can be 3-dimensional (Figure~\ref{fig:3points3D}), and Aasa Feragen, Steve Marron, Ezra Miller, Vinayak Pathak, Scott Provan, and Sean Skwerer for helpful discussions about convex hulls in tree space.
\revised{We are extremely grateful to some incredibe anonymous reviewers, who greatly improved the paper and caught an error.}  % well, actually two errors.  AL.
%We thank anonymous reviewers for helpful comments.
 
%\note{Should this be in a footnote, like the one for Stefan Langerman?  Up to you.  A.L.}   
MO acknowledges the support of the Fields Institute.  
Research \revised{of all authors was} supported by NSERC, the Natural Sciences and Engineering Research Council of Canada.
%covers all of us I think.
%\note{should Waterloo or any of your grants be included for my support?}  

\bibliographystyle{abuser}   % This style wasn't working with sharelatex.
\bibliography{CAT0}

\end{document}